\documentclass[11pt]{article}

\usepackage[margin=1in]{geometry} 

\usepackage{setspace}
\usepackage{mathtools}
\usepackage{amsmath, amssymb, amsfonts}
\usepackage{amsthm}
\usepackage{mathrsfs}
\usepackage{bbm}
\usepackage{accents}
\newcommand{\widebar}[1]{\overline{#1}}
\usepackage{scalerel}
\usepackage{pifont}

\DeclareFontFamily{U}{mathx}{\hyphenchar\font45}
\DeclareFontShape{U}{mathx}{m}{n}{
      <5> <6> <7> <8> <9> <10>
      <10.95> <12> <14.4> <17.28> <20.74> <24.88>
      mathx10
      }{}
\DeclareSymbolFont{mathx}{U}{mathx}{m}{n}
\DeclareMathAccent{\widecheck}{0}{mathx}{"71}

\theoremstyle{plain} 
\newtheorem{theorem}{Theorem}[section]
\newtheorem{proposition}[theorem]{Proposition}
\newtheorem{corollary}[theorem]{Corollary} 
\newtheorem{lemma}[theorem]{Lemma} 

\theoremstyle{definition}
\newtheorem{definition}[theorem]{Definition} 
\newtheorem{example}[theorem]{Example} 

\theoremstyle{remark}
\newtheorem{remark}[theorem]{Remark}

\theoremstyle{definition}
\newtheorem{assumption}{Assumption}

\theoremstyle{definition}
\newtheorem{condition}{Condition}
\usepackage{booktabs}
\usepackage{multirow}
\usepackage{enumitem}

\usepackage{graphicx}
\usepackage{float}
\graphicspath{{./art/}}

\usepackage{caption}
\usepackage[numbers]{natbib} 
\usepackage[colorlinks=true,
            linkcolor=blue,
            citecolor=blue,
            urlcolor=magenta]{hyperref}

\usepackage[linesnumbered, ruled, vlined]{algorithm2e}
\usepackage{tikz}
\usetikzlibrary{graphs, positioning}

\DeclareMathOperator*{\argmin}{arg\,min}

\def\T{{ \mathrm{\scriptscriptstyle T} }}
\def\v{{\varepsilon}}

\title{\bfseries Assumption-lean Inference for\\ Network-linked Data}

\makeatletter
\renewcommand{\thefootnote}{\fnsymbol{footnote}}
\makeatother

\author{%
  Wei Li\footnotemark[1] \and
  Nilanjan Chakraborty\footnotemark[2] \and
  Robert Lunde\footnotemark[1]
}

\date{} 

\begin{document}
\maketitle

\footnotetext[1]{Department of Statistics and Data Science, Washington University in St.~Louis, USA.}
\footnotetext[2]{Department of Mathematics and Statistics, Missouri University of Science and Technology, USA.}

\makeatletter
\renewcommand{\thefootnote}{\arabic{footnote}}
\setcounter{footnote}{0}
\makeatother

\begin{abstract}
We consider statistical inference for network-linked regression problems, where covariates may include network summary statistics computed for each node. In settings involving network data, it is often natural to posit that latent variables govern connection probabilities in the graph. Since the presence of these latent features makes classical regression assumptions even less tenable, we propose an assumption-lean framework for linear regression with jointly exchangeable regression arrays. We establish an analog of the Aldous-Hoover representation for such arrays, which may be of independent interest.

Moreover, we consider two different projection parameters as potential targets and establish conditions under which asymptotic normality and bootstrap consistency hold when commonly used network statistics, including local subgraph frequencies and spectral embeddings, are used as covariates. In the case of linear regression with local count statistics, we show that a bias-corrected estimator allows one to target a more natural inferential target under weaker sparsity conditions compared to the OLS estimator. Our inferential tools are illustrated using both simulated data and real data related to the academic climate of elementary schools.
\end{abstract}


\section{Introduction}
In many modern applications, data linked through an underlying network, commonly referred to as network-linked data (e.g., \citealp{li2019prediction}, \citealp{le2022linear}, \citealp{lunde2023conformal}), is becoming increasingly common. These datasets consist of nodal units such as individuals, products, or institutions. Typically, each nodal unit is associated with attributes and an outcome which are structurally connected through an observed network that encodes relational, transactional, or social ties. Across disciplines such as social science (\citealp{goldsmith2013social}, \citealp{bramoulle2009identification}), economics (\citealp{manski1993identification}, \citealp{graham2017econometric}), finance (\citealp{acemoglu2015systemic}, \citealp{hochberg2007whom}), biology (\citealp{segal2003module}) and chemistry (\citealp{hansen2015machine}), network structures have been leveraged in regression problems to study relationships between nodal outcomes and individual-level features. In these settings, one may be interested in predictive or inferential tasks where the goal is to relate nodal responses to both conventional covariates and network-derived features, a framework broadly known as network-assisted regression. For instance, in portfolio management, studies such as \cite{hochberg2007whom} show that network-derived metrics, particularly centrality, can serve as informative predictors of financial outcomes to aid in identifying and mitigating systemic vulnerabilities. In e-commerce applications, a customer’s purchase behavior may depend not only on personal demographics but also on their position in a social or product co-view network. In both examples, the network is expected to provide meaningful structural information. For such data, network-assisted regression models offer a natural and flexible approach: each node is treated as a data point in a regression model that incorporates both nodal attributes and features derived from its position or role within the network.

In recent years, there has been tremendous progress in the modeling of network-linked data.
Methods now span from classical regression to modern machine learning techniques such as graph neural networks (\citealp{scarselli2008graph}), which have achieved substantial performance gains by leveraging structural information encoded in the network. \citet{li2019prediction} propose a predictive model for nodal outcomes with a cohesion penalty that adaptively groups node-specific intercepts, boosting predictive accuracy by capturing local similarity. More recently, conformal prediction has been extended to network-linked regression and graph neural networks (\citealp{lunde2023conformal}, \citealp{huang2023uncertainty}), where notions of symmetry are used to construct valid prediction sets. While these approaches offer uncertainty quantification for predictions, they do not provide formal inference guarantees on model parameters, which remain essential for interpreting the relationship between covariates and outcomes in the presence of network dependence.

The development of statistical models with valid inference guarantees has only recently begun to receive attention, largely due to the non-standard and complex dependence structures in network-linked data. A notable contribution by \citet{le2022linear} extend the cohesion-based approach into a semi-parametric model that incorporates network effects as interaction terms projected onto a low-rank spectral subspace, allowing for inference under certain structural assumptions. Similarly, \citet{zhu2017network, zhu2020grouped} study network autoregressive models with time-lagged responses and neighborhood effects, using least squares estimation under stationarity. \citet{cai2021network} develop a unified covariate-supervised estimation framework for centrality measures with inferential guarantees when the model is correct. Earlier foundational work in social sciences and economics focused on peer effects and the distinction between endogenous and exogenous influence (e.g., \citealp{manski1993identification}, \citealp{bramoulle2009identification}, \citealp{goldsmith2013social}), typically under the assumption of a fixed but unobserved network influencing the response surface.

More recent efforts have incorporated randomness into the network itself, particularly in causal inference. For example, \citet{li2022random} introduce a graphon-based model for estimating causal effects under network interventions, and \citet{hayes2022estimating} use spectral embeddings of random networks in a mediation analysis framework. Despite leveraging random graph models, these approaches still rely on linearity or structured assumptions to define and estimate network effects. In contrast, we consider a fully random-design setting, where the network, covariates, and outcomes are jointly sampled, and the network is generated from a graphon model with latent positions. Our approach adopts an assumption-lean inference framework: instead of requiring model correctness, we define network effects through population-level linear projections, allowing for robust estimation even under model misspecification.  With network models, it is common to assume that latent variables govern connection probabilities in the network. The presence of these latent variables makes the assumption of correct model specification even less tenable than standard regression problems, necessitating inferential methods that are robust to misspecification. 

Our framework not only captures meaningful network effects in complex or sparse regimes but also enables valid inference via a resampling scheme. In this framework, we are also able to incorporate cohesion effects, which we define as covariates that involve averages over a node's $k$-hop neighborhood for some prespecified $k$. To the best of our knowledge, our work provides one of the most general settings with inference guarantees for network-linked regression accounting for model specification. A comparison with existing approaches is provided in the following table.

\begin{table}[ht]
\centering
\small
\caption{Comparison with Existing Network Regression Approaches}
\begin{tabular}{|l|c|c|c|}
\hline
Method & Random Design & Assumption-lean & Cohesion Effect \\
\hline
Subspace Projection (\cite{le2022linear})      & $\times$ & $\times$ & \checkmark \\
\hline
Causal Network Mediation (\cite{hayes2022estimating})     & \checkmark & $\times$ & $\times$ \\
\hline
Our Method                          & \checkmark & \checkmark & \checkmark \\
\hline
\end{tabular}
\label{tab:comparison}
\end{table}

\subsection{Related Work}
Our work draws on several lines of literature in network analysis, addressing inference and estimation in the presence of network structure as well as the assumption-lean inference framework for regression. 

A foundational line of work in network modeling has focused on latent variable models known as graphons, which emerge from the theory of graph limits and array exchangeability \citep{Lovacz-Graph-Limits}.  Sparse graphon models, which have been proposed in various contexts by \citet{Bollobás_Riordan_2009}, \citet{borgs-lp-part-one} and \citet{Bickel-Chen-on-modularity}, provide a flexible class of network models that also allows sequences of graphs to be sparse, which is not possible with standard graphon models. Estimation and inference under sparse graphon models have been a central topic of research in the statistical network analysis literature. Among the early and influential contributions, \citet{bickel2011method} develop asymptotic theory for subgraph frequencies under the sparse graphon model.

The random dot product graph (RDPG) model, introduced by \citet{young-schneiderman-rdpg}, is another widely used network model that may be viewed as a low-rank graphon model \citep[see][]{lei2021network}. Regarding statistical inference for this model, \citet{athreya2018statistical} provide a comprehensive theoretical treatment of latent position recovery under RDPG via spectral embedding, establishing both consistency and asymptotic normality. Many of these results are generalized to generalized random dot product graph models (GRDPGs) in \citet{generalized-rdpg}.   These results support our later use of latent position–based covariates in regression settings. Despite their flexibility, GRDPGs suffer from certain identifiability issues, complicating statistical inference. Recent works have explored these challenges in hypothesis testing problems; see, for example, \citet{agterberg2020two} and \citet{du2023hypothesis}.

Our approach aligns with a growing interest in robust inference under model misspecification, particularly in settings where covariates are random and the regression function may be nonlinear or otherwise unknown. The assumption-lean framework formalized by \citet{berk2021assumption} promote the idea of treating linear models as approximations rather than true data-generating mechanisms, shifting the inferential target to projection-based summaries.  Statistical inference for such projection parameters was first studied by \citet{e7a4e83b-9158-3e76-82d2-3186d14b4bbd}.  \citet{vansteelandt2022assumption} further extend this framework to 
generalized linear models, emphasizing the use of influence functions and robust variance estimation. \citet{chang2023inference} develop Berry–Esseen bounds for the projection parameter in the i.i.d. setting, where the dimension is allowed to grow with the sample size. These works motivate our use of linear projections as meaningful inferential targets, even when the true data generating process is complex or partially unobservable due to latent network structure.

Recent works have proposed adaptations of resampling methods tailored to the dependence structure of graphs. \citet{levin2019bootstrapping} consider a bootstrap procedure for statistics under the graphon model that exhibits U-statistic structure. \citet{lin2020trading} develop both a fast computational linear bootstrap procedure and a higher-order correct bootstrap for count functionals under the graphon model, \citet{zhang2022edgeworth} also explore the higher-order correct property of the bootstrap procedure for the count functional. Moreover, \citet{lin2020theoretical} explore the network jackknife estimator, \citet{lunde2023subsampling} examine different subsampling procedures for the graphon under the minimal requirement of the weak convergence of the functional, and \citet{bhattacharyya2015subsampling} develop sub-sampling approaches that enable valid inference for the count functionals under the graphon model.  Collectively,  these contributions inform the development of our own resampling-based inference procedures.

\subsection{Our Contribution and Organization}
Our main contributions are summarized as follows. We propose a general regression framework for network-linked data leveraging both random graph theory and an assumption-lean inference approach. We introduce two possible projection parameters and develop inferential procedures that yield consistent inference for these targets for a wide range of network covariates.  To the best of our knowledge, our work is the first to provide valid statistical inference for regression effects in network-linked data under potential model misspecification. Compared to existing methods, our approach captures meaningful network information while remaining more robust and broadly applicable.

To this end, we develop inferential tools tailored to subgraph counts and spectral embeddings, which represent important subclasses of network covariates. We also develop more broadly applicable inferential tools that provide guarantees for network statistics whose properties are not as well understood. In the case of subgraph-based network covariates, we propose a novel bias-corrected estimator that yields valid inference for a natural target parameter under weaker sparsity conditions than the OLS estimator, and establish the consistency of a multiplier bootstrap procedure for inference. For problems involving spectral embeddings, we study inferential problems for which certain identifiability issues can be circumvented, and develop efficient bootstrap-based tools for inference.

\section{Problem Setup}
\label{sec:prob-setup}
\subsection{Jointly Exchangeable Regression Arrays}
We consider a model inspired by a notion of exchangeability for regression arrays.  Let $Y_1, \ldots, Y_n \in \mathbb{R}$ be the response of interest, $X_1, \ldots, X_n \in \mathbb{R}^p $ be covariates, and $\mathrm{A} \in \{0,1\}^{n \times n}$ be an adjacency matrix.  Furthermore, define the $n \times n$ regression array $V$, where $V_{ij} = (Y_i,X_i, Y_j, X_j, A_{ij})$ and let $V^\sigma$ be an array in which $V_{ij}^\sigma = V_{\sigma(i) \sigma(j)}$. In \citet{lunde2023conformal}, the notion of a jointly exchangeable regression array was introduced; such an array satisfies:
\begin{align*}
(V_{ij}^\sigma)_{1 \leq i,j \leq n} = (V_{ij})_{1 \leq i,j \leq n} \quad\text{ in distribution},
\end{align*}
for any permutation $\sigma: [n] \mapsto [n]$.  In the above work, it was shown that many natural data generating mechanisms fall under this framework.  When the array $V$ is infinite-dimensional with indices in $\mathbb{N}^2$, a representation theorem for such arrays in the spirit of de Finetti and \citet{aldous1981representations} and \citet{hoover1979relations} suggests a class of models that we consider in the present paper.  We state this result below.  We consider an adjacency matrix $\mathrm{A}$ with no self-loops, for which it makes sense to restrict our attention to indices of the form $\{i \neq j \ | \ i,j \in \mathbb{N}^2\}$.

\begin{theorem}\label{main:theorem1} (Aldous-Hoover Representation for Jointly Exchangeable Regression Arrays) Suppose that $(V_{ij})_{i \neq j}$ is a jointly exchangeable regression array.  Then, there exist mutually independent collections of i.i.d. $\mathrm{Uniform}[0,1]$ random variables $(\xi_i)_{i \in \mathbb{N}}, (\eta_{ij})_{i < j}, \alpha$, and Borel measurable $h,g$ such that:
\begin{align*}
(V_{ij})_{i \neq j} = \left(h(\alpha,\xi_i), h(\alpha, \xi_j), g(\alpha, \xi_i, \xi_j, \eta_{ij}) \right)_{i \neq j} \ \quad\text{in distribution},
\end{align*}
where  $h(\alpha,\xi_i) = (Y_i,X_i)$ in distribution and for $j > i$, $\eta_{ji} = \eta_{ij}$.  
\end{theorem}

To the best of our knowledge, this representation theorem for regression arrays was not previously known.  The Aldous-Hoover theorem also applies to regression arrays of this type, but the main question is whether nodal random variables can be expressed as simpler functions involving fewer random variables. A proof of the above result is provided in the Supplement.  For modeling purposes, it is natural to condition on the mixing variable $\alpha$, which then implies that $(Y_i, X_i, \xi_i)_{i \in \mathbb{N}}$ are (conditionally) i.i.d. Furthermore, when $A$ is binary and symmetric with no self-loops, we can consider a measurable function $w(x,y)$ where $w(x,x) = 0$, $w(x,y) = w(y,x)$ such that, conditional on $\alpha$, 
\begin{align*}
(A_{ij})_{i< j} = (\mathbbm{1}(\eta_{ij} \leq w(\xi_i, \xi_j))_{i < j} \quad  \text{in distribution}. 
\end{align*}

See, for example, Corollary III.6 of \citet{6847223}.  The function $w$ has come to be known as a graphon and also arises in graph limit theory; see, for instance, \citet{lovasz2006limits} and \citet{diaconis-janson-exchangeable-graph}. 
In what follows, we consider a sparse graphon model, where a sparsity parameter $\rho_n \rightarrow 0$ governs the sparsity level of the graph.  Based on the above considerations, we consider the following model:

\begin{assumption}[Independent Triples and Sparse Graphon]\label{assumption1:graphon}
Suppose that \((Y_i,X_i,\xi_i)_{i=1}^n \sim P\), and let 
\(A^{(n)}\in\{0,1\}^{n\times n}\) be generated by the sparse graphon model, with
\(A_{ij}^{(n)}=A_{ji}^{(n)}\) for \(i,j\in[n]\), \(A_{ii}^{(n)}=0\) for all \(i\in[n]\), 
and for \(i<j\),
\begin{equation}\label{eq-sparse-graphon}
A_{ij}^{(n)}=\mathbbm{1}\!\left\{\eta_{ij}\le \rho_n\,w(\xi_i,\xi_j)\wedge 1\right\}.
\end{equation}
Here \(\{\eta_{ij}\}_{1\le i<j\le n}\) are independent \(\mathrm{Unif}(0,1)\) random variables.
Assume \(\rho_n\to0\), and also that
\(\int_{[0,1]^2} w(x,y)\,dx\,dy=1\), \(w(x,y)\le C<\infty\).
\end{assumption}

Several comments are in order. First, it should be noted that the graphon model subsumes many other commonly studied models, including stochastic block models and (generalized) random dot product graphs, among others. Second, as is typical in the literature, we consider bounded graphons; however, our results can be extended to the unbounded case \citep{borgs-lp-part-one}, albeit at the expense of longer proofs and additional technical conditions.   Moreover, while the independent triple and graphon setup is natural for jointly exchangeable regression arrays, it should be noted that, strictly speaking, this representation theorem only holds for infinite arrays. Finite jointly exchangeable arrays may not admit such a representation but may often be approximated by such arrays (see \citet{VOLFOVSKY201654}).

Furthermore, while $\xi_1, \ldots, \xi_n$ are of natural interest in regression problems, they are unobserved and not identifiable.  Instead, we consider regression problems for which network information is captured by certain identifiable graph features of the form $Z_i =  f(\xi_i)$. We discuss this point in more detail in the next subsection.

\subsection{Assumption-Lean Inference for Network-Linked Regression}\label{assumption-lean}
 In what follows, let $\widehat{Z}_i$ be a nodal summary statistic computed from $\mathrm{A}$ that estimates the quantity $Z_i$. For instance, in the case of degrees, we may view $\widehat{Z}_i = \sum_{j \neq i} A_{ij}/(n-1) \widehat{\rho}_n$ as a noisy realization of $Z_i = E(w(\xi_i, \xi_j) \ | \ \xi_i)$, where $\widehat{\rho}_n$ is an estimator for $\rho_n$.  Degrees are one of many network statistics that capture the notion of a node's centrality in the network.  While it is often natural to posit that a node's centrality affects the response variable, it is nearly impossible to argue that a linear model involving degrees is correctly specified given the multitude of other choices for centrality measures. From another perspective, the latent positions themselves are abstract entities derived from an Aldous-Hoover representation, making it unlikely that one could choose a functional $f(\xi_i)$ that ensures classical linear regression assumptions are satisfied.  

For i.i.d. data, \citet{berk2021assumption} propose an assumption-lean inference framework for linear regression, where inference targets are well defined even under model misspecification.  This viewpoint is particularly well-suited for network-linked regression problems, for which parametric assumptions are often untenable for reasons described above.  However, defining these projection parameters is more subtle in the network setting.  In what follows, for $i \in [n]$, let  \(\widehat{L}_i = (X_i, \widehat{Z}_i)\)  and \(L_i = (X_i, Z_i)\), where $ X_{i,1} =1$. For notational convenience, we define the augmented design matrix as \(\widehat{\mathrm{L}} = [\mathrm{X}\ \;\widehat{\mathrm{Z}}]\), where $\mathrm{X} \in \mathbb{R}^{n\times p}$ is the covariate matrix and \(\widehat{\mathrm{Z}} \in \mathbb{R}^{n \times d}\) denotes the additional graph-derived features (e.g. local subgraph count statistics). Our first proposal for a target parameter is the following:
\begin{align}\label{noisy target}
\widetilde{\beta} = \arg \min_{\beta} E(Y-\beta^\T\widehat{L})^2 = \widetilde{\Lambda}^{-1} \widetilde{\gamma},
\end{align}
where $\widetilde{\Lambda} = E(\widehat{L}\widehat{L}^\T)$ and $\widetilde{\gamma}=E(Y\widehat{L})$. This parameter will typically be easier to conduct inference for compared to other possible targets.  However, this target is slightly unnatural given the independent triplet assumption since $\widehat{Z}_1, \ldots, \widehat{Z}_n$ are typically not i.i.d. and are often noisy versions of the network statistics of interest.  Instead, let  $\Lambda = E(LL^\T)$, $\gamma = E(YL)$  and define the following alternative parameter:
\begin{align}\label{clean beta}
\beta^* = \arg \min_{\beta} E(Y-\beta^\T L)^2 = \Lambda^{-1} \gamma.
\end{align}
In the case where $Y$ is a linear function that involves $Z$ and $X$, the above target would coincide with the true regression coefficients. However, even when network noise is asymptotically negligible, it will often be the case that inference for $\beta^*$ is more difficult and requires additional conditions compared to inference for $\widetilde{\beta}$. 

When the linear model is misspecified, these target parameters remain a meaningful object rather than arbitrary targets. As shown in \cite{berk2021assumption}, projection parameters correspond to weighted averages of adjusted case-wise slopes through the origin or pairwise slopes between observations, where the weights depend on the distribution of the covariates. In fact, this interpretation has a long history, appearing in the literature on ancillary statistics \cite{stigler2001ancillary}, the jackknife and bootstrap \cite{wu1986jackknife}, and model diagnostics \cite{gelman2009splitting}. Crucially, this perspective highlights that the projection parameter serves as a robust metric of the overall association between the outcome and covariates, without requiring correct specification of the functional form. Our simulations will revisit this point, demonstrating that the projection parameter can often capture the intuitive direction and strength of association, even under complex, nonlinear data-generating processes.

For both target parameters, a natural estimator is given by the standard OLS estimate:   
\begin{align*}
\widehat{\beta} = \arg \min_{\beta} \sum_{i=1}^n(Y_i-\beta^\T\widehat{L}_i)^2 = \widehat{\Lambda}^{-1} \widehat{\gamma},
\end{align*}
where $\widehat{\Lambda} = n^{-1} \sum_{i=1}^n \widehat{L}_i\widehat{L}_i^\T$ and $\widehat{\gamma} = n^{-1}\sum_{i=1}^n Y_i \widehat{L}_i$.   In certain instances, we will be able to construct modifications of the OLS estimate that correct for bias arising from noise in network covariates, leading to better inferential properties when targeting $\beta^*$.  Bias corrections are discussed in Section \ref{subsec-local-subgraph}. In the following sections, we establish asymptotic normality and bootstrap consistency results when the network covariates involve local subgraph counts or spectral embeddings, which are widely used in practice. We then present a resampling scheme with down-sampling that allows consistent inference even when the behavior of the network statistics is not as well understood.         

\section{Main Results}
\subsection{Inference with Local Subgraph Counts}
\label{subsec-local-subgraph}
Subgraph counts are widely recognized as critical descriptors of network structure, capturing unique meso-scale patterns. In biological networks, such as protein–protein interaction (PPI) graphs, the frequency of specific small subgraphs—often called network motifs—is linked to biological function. For instance, triangles and cliques often indicate protein complexes, where multiple proteins jointly interact to carry out a cellular function \citep{Milo2002, ShenOrr2002}. In social networks, local subgraph frequencies also encode interpretable social mechanisms. The number of triangles around a person is strongly associated with triadic closure, a tendency for friends of friends to become friends themselves, and is used to infer social cohesion \citep{Granovetter1973, Watts1998}. Star-like motifs (e.g., 1-to-many connections) can signal influencers or hub individuals in information diffusion, while the presence of squares may reflect indirect but tight-knit group structures \citep{kitsak2010identification, benson2016higher}. These subgraph patterns are often leveraged in predictive models for behavioral outcomes, such as political alignment \citep{conover2011political}, information sharing \citep{bakshy2011information}, or health behavior in online social networks \citep{centola2010spread}. Given these diverse roles across domains, incorporating subgraph counts as covariates in regression models enables associating local structural patterns with node-level responses, thereby facilitating, for example, more interpretable and biologically or socially grounded inference.

We now define the relevant subgraph frequencies. We associate a subgraph of interest $\mathcal{R}$ with  its edge set $\mathcal{E}(\mathcal{R}) \subseteq \{\{i,j\}  \ | \  i,j \in [n]\}$  and vertex set $\mathcal{V}(\mathcal{R}) = \{ i \ | \  \{i,j\} \in \mathcal{E}(\mathcal{R}) \text{ for some j} \}$.  Let $r = |\mathcal{V}(\mathcal{R})|$; for notational simplicity, we consider subgraphs for which $\mathcal{V}(\mathcal{R}) = \{1, \ldots, r\}$.  A natural parameter related to $\mathcal{R}$ is given by:
\begin{align*}
Q(\mathcal{R}) 
=E\left( \prod_{\{i,j\} \in \mathcal{E}(\mathcal{R})} w(\xi_i,\xi_j)\right).
\end{align*}

Since the parameter is an expectation involving  i.i.d. random variables, it is natural to consider estimators related to U-statistics. In what follows, we say that two subgraphs $\mathcal{R}_1$ and $\mathcal{R}_2$ are isomorphic, denoted $\mathcal{R}_1 \cong \mathcal{R}_2$, if there exists a bijection $ f: \mathcal{V}(\mathcal{R}_1) \mapsto \mathcal{V}(\mathcal{R}_2)$ such that $\{i,j\} \in \mathcal{E}(\mathcal{R}_1)$ if and only if $\{f(i), f(j)\} \in \mathcal{E}(\mathcal{R}_2)$.  Let $|\mathrm{Iso}(\mathcal{R})|$ denote the number of distinct graph isomorphisms that map into $\{1, \ldots, r\}$.  An estimator of $Q(\mathcal{R})$ is given by:
\begin{align*}
\widehat{Q}(\mathcal{R}) =  \frac{1}{{n \choose r} \hat{\rho}_n^s} \sum_{1 \leq i_1 <  \cdots < i_r \leq n} \frac{1}{|\mathrm{Iso}(\mathcal{R})|}
\sum_{\substack{\mathcal{S}: \mathcal{S} \cong \mathcal{R},\\
\mathcal{V}(\mathcal{S}) = \{i_1, \ldots, i_r\}}} \prod_{\{i,j\} \in \mathcal{E}(\mathcal{S}) } A_{ij},
\end{align*}
where $\hat{\rho}_n = \frac{1}{{n \choose 2}} \sum_{1 \leq i < j \leq n} A_{ij}$ is a consistent estimator for $\rho_n$ under mild conditions (see \citet{bickel2011method}).  
Given this definition of a global subgraph count, one may define a local subgraph count for node $i$ as the number of subgraphs node $i$ is involved in:
\begin{align}
\label{eq-local-subgraph}
\widehat{Q}_i(\mathcal{R}) =  \frac{1}{{n-1 \choose r-1} \hat{\rho}_n^s} \sum_{\substack{1 \leq i_1 <  \cdots < i_{r-1} \leq n, \\ i_1, \ldots, i_{r-1} \neq i}} \frac{1}{|\mathrm{Iso}(\mathcal{R})|}
\sum_{\substack{\mathcal{S}: \mathcal{S} \cong \mathcal{R},\\
\mathcal{V}(\mathcal{S}) = \{i,i_1, \ldots, i_{r-1}\}}} \prod_{\{j,k\} \in \mathcal{E}(\mathcal{S}) } A_{jk}.
\end{align}

In this case, it is natural to define $\beta^*$ with respect to the following noiseless network covariate:
\begin{align*}
Z_i = E[\widetilde Q_i(\mathcal{R}) \ | \ \xi_i], 
\end{align*}
where
\begin{align}\label{latent_node_stat}
\widetilde Q_i(\mathcal{R}) =  \frac{1}{{n-1 \choose r-1}} \sum_{\substack{1 \leq i_1 <  \cdots < i_{r-1} \leq n, \\ i_1, \ldots, i_{r-1} \neq i}} \frac{1}{|\mathrm{Iso}(\mathcal{R})|}
\sum_{\substack{\mathcal{S}: \mathcal{S} \cong \mathcal{R},\\
\mathcal{V}(\mathcal{S}) = \{i,i_1, \ldots, i_{r-1}\}}} \prod_{\{j,k\} \in \mathcal{E}(\mathcal{S}) } w(\xi_j,\xi_k).    
\end{align}
It can readily be verified that $\widehat{Q}(\mathcal{R}) = n^{-1}\sum_{i=1}^n \widehat{Q}_i(\mathcal{R})$. Asymptotic theory for $\widehat{Q}(\mathcal{R})$ was first established by \citet{bickel2011method}. It is often also of interest to consider notions of local subgraph counts in which the node $i$ occupies a fixed position in the subgraph.  Suppose that $\mathcal{R}$ corresponds to a $k$-star, which involves a root node $v$ and $k$ edges of the form $\{u, v\}$.  We also consider the following quantity associated with these rooted subgraphs:    
\begin{align}
\label{eq-local-subgraph-rooted}
\widehat{Q}_i^{\mathrm{rooted}}(\mathcal{R}) = \frac{1}{{n-1 \choose r-1} \hat{\rho}_n^s} \sum_{\substack{1 \leq i_1 <  \cdots < i_{r-1} \leq n, \\ i_1, \ldots, i_{r-1} \neq i}} \ \prod_{j = i_1}^{i_{r-1}} A_{ij}.  
\end{align}

We consider three classes of subgraphs.  Acyclic subgraphs are those that contain no cycles. Simple cycles refer to trails with the same starting and ending vertex.  General cycles refer to cycles that are not simple, such as cliques.  The following result establishes conditions under which inference based on the normal approximation is possible for our target parameters $\beta^*$ and $\widetilde{\beta}$. In what follows, let $\lambda_n = n \rho_n$, which captures the average degree. 

For the following theorem statements, let $R = \max_{k \in [d]} r_k$.
\begin{theorem}[Asymptotic Normality of OLS Estimator for $\beta^*$]\label{theorem 1}
Suppose that $\Lambda = E(LL^\T)$ is invertible, $E(Y^4) < \infty$, and $E(\|X\|^4) < \infty$. Further, suppose that one of the following conditions holds:
\begin{enumerate}
\item[(a)] All subgraphs are acyclic and $\lambda_n =  \omega(n^{1/2})$,
\item[(b)] All subgraphs are acyclic or simple cycles, and $\lambda_n^{r_j} = \omega(n^{3/2})$ for all $j$ corresponding to simple cycles.
\item[(c)] The sparsity condition in (b) holds for simple cycles, and 
$\lambda_n = \omega\!\left(
   n^{\,1-\tfrac{2R-3}{R(R-1)}} \vee 
   n^{\,1-\tfrac{2}{2R-1}}
\right)$ holds for subgraphs that are neither acyclic nor simple cycles. Then,
\begin{align*}
n^{1/2}(\widehat{\beta}-\beta^*) \rightarrow N(0,\Sigma_\beta) \quad\text{ in distribution,}
\end{align*}
where $\Sigma_\beta$ is defined in Supplement~\ref{sec-B.3}, equation~\eqref{CLT_covaraince_beta}.
\end{enumerate}
\end{theorem}
We now state a theorem for $\widetilde{\beta}$, which holds under weaker sparsity conditions than the corresponding theorem for $\beta^*$.  
\begin{theorem}[Asymptotic Normality of OLS Estimator for $\widetilde{\beta}$]\label{theorem 2}
Suppose that $\Lambda = E(LL^\T)$ is invertible, $E(Y^4) < \infty$, and $E(\|X\|^4) < \infty$. Further, suppose that one of the following conditions holds:
\begin{enumerate}
\item[(a)] All subgraphs are acyclic and $\lambda_n =  \omega(1)$.
\item[(b)] All subgraphs are acyclic or simple cycles, and $\lambda_n^{r_j} = \omega(n)$ for all $j$ corresponding to simple cycles.
\item[(c)] The sparsity condition in (b) holds for simple cycles and $\lambda_n = \omega(n^{\,1 - 2/\{\,2R-1\}})$ holds for subgraphs that are neither acyclic nor simple cycles. 
Then,
\begin{align*}
n^{1/2}(\widehat{\beta}-\widetilde{\beta}) \rightarrow N(0,\Sigma_\beta) \quad\text{ in distribution.}
\end{align*}
\end{enumerate}
\end{theorem}
An important step in our proofs of the above theorems is deriving limiting distributions for quantities of the form:
\begin{align*}
n^{1/2}\begin{pmatrix}
\widehat{\Lambda} - \widetilde{\Lambda} \\ 
\widehat{\gamma} - \widetilde{\gamma}
\end{pmatrix}, 
\quad \textrm{and}\quad n^{1/2}\begin{pmatrix}
\widehat{\Lambda} - \Lambda \\ 
\widehat{\gamma} - \gamma
\end{pmatrix}.
\end{align*}

While the asymptotic behavior of $\widehat{Q}(\mathcal{R})$ is well understood, central limit theorems for terms related to $n^{-1}\sum_{i=1}^n Y_i \widehat{Z}_i$ and $n^{-1}\sum_{i=1}^n \widehat{Z}_i\widehat{Z}_i^\T$ have not been studied previously in the literature. \citet{10.1093/biomet/asad080} considers Lipschitz functions of local subgraph frequencies; since the functions of interest are not Lipschitz, these results do not apply to our setting.  Moreover, our results hold under weaker sparsity conditions and for a broader class of subgraph frequencies for which a root node need not be fixed. Our study of these local subgraph frequencies and corresponding regression estimators reveals an interesting bias phenomenon that does not arise in the study of global subgraph frequencies. The above theorem has important implications for practitioners since it suggests that inference for $\beta^*$, which coincides with the true regression coefficients when the model is linear, requires much stronger sparsity conditions. We believe that these conditions for $\beta^*$ are sharp; in fact, our simulation results suggest that graphs need to be substantially denser than the conditions above when the sample size is in the thousands.  Intuitively, $\widehat{\beta}$ is sensitive to network noise since small perturbations in the covariates can lead to substantial changes in an appropriate matrix inverse.  In fact, it turns out that terms of the form $n^{-1} \sum_{i=1}^n \widehat{Z}_{ij}^2$, $n^{-1} \sum_{i=1}^n \widehat{Z}_{ij}\widehat{Z}_{ik}$ contribute disproportionately to this problem.  The next proposition sheds light on the structure of this term and suggests a potential remedy.    

\begin{proposition}\label{main:prop1}
Let $\widehat{Z}_{i,j}$ and $\widehat{Z}_{i,k}$ be local subgraph frequencies associated with subgraphs $\mathcal{R}_j$ and $\mathcal{R}_k$, with node-edge parameters $(r_j, s_j)$ and $(r_k, s_k)$, respectively; that is, $\widehat{Z}_{i,j} = \widehat{Q}_i(\mathcal{R}_j)$ and $\widehat{Z}_{i,k} = \widehat{Q}_i(\mathcal{R}_k)$. 
Then, we have the following decomposition:
\begingroup
\setlength{\abovedisplayskip}{3pt}
\setlength{\belowdisplayskip}{3pt}
\begin{align*}
n^{-1} \sum_{i=1}^{n} \widehat{Z}_{i,j} \widehat{Z}_{i,k} = \underbrace{\sum_{\mathcal{M} \in \mathfrak{M}^{\mathcal{R}_j, \mathcal{R}_k}_{1,0}} C^{\mathcal{M}}_{\mathcal{R}_j,\mathcal{R}_k,1}\cdot \widehat{Q}(\mathcal{M})}_{\text{Leading term, } S_{n,jk}} + \underbrace{\sum_{\substack{2 \leq c \leq \min(r_j, r_k) \\ 0 \leq d \leq \min(s_j, s_k)}} O(\hat{\rho}_n^{-d} n^{1 - c}) \sum_{\mathcal{M} \in \mathfrak{M}^{\mathcal{R}_j, \mathcal{R}_k}_{c,d}} C^{\mathcal{M}}_{\mathcal{R}_j,\mathcal{R}_k,c} \cdot \widehat{Q}(\mathcal{M})}_{\text{Remainder term, } R_{n,jk}}.
\end{align*}
\endgroup
where $\mathfrak{M}^{\mathcal{R}_j, \mathcal{R}_k}_{c,d}$ is a set of subgraphs containing one representative from each isomorphism class of subgraphs satisfying $\mathcal{S} = \mathcal{S}_j \cup \mathcal{S}_k$, where $\mathcal{S}_j \cong \mathcal{R}_j$, $\mathcal{S}_k \cong \mathcal{R}_k$, $|\mathcal{V}(\mathcal{S}_j) \cap \mathcal{V}(\mathcal{S}_k)| = c$, and $|\mathcal{E}(\mathcal{S}_j) \cap \mathcal{E}(\mathcal{S}_k)|=d$. Note that the isomorphism class         \;$\mathfrak{M}^{\mathcal{R}_j, \mathcal{R}_k}_{c,d}$ can be empty for certain values of $d$ between 0 and $\min(s_j,s_k)$ when no such subgraph exists. Above, $C^{\mathcal{M}}_{\mathcal{R}_j,\mathcal{R}_k,c}$ is a non-negative constant associated with each isomorphism class for each $c$.
\end{proposition}

This representation arises from multiplying two locally normalized subgraph statistics and reorganizing the resulting product into global statistics. Even after this expansion, the analysis is non-trivial since these U-statistics are not centered naturally when targeting $\beta^*$. In fact, the stronger sparsity conditions needed for $\beta^*$ may be viewed as a consequence of a mismatch in the centering associated with these quadratic terms.  In particular, additional conditions are needed for $R_n$ to be asymptotically negligible with improper centering.  To rectify this issue, we propose a modified regression estimator in which these quadratic terms are replaced with the leading term in the above representation. This leads to the following bias-corrected estimator: 
\begin{align}\label{modifiedOLS}
\widehat{\beta}_{\mathrm{mod}} = 
\widehat{\Lambda}_{\mathrm{mod}}^{-1} \widehat{\gamma}, \qquad\text{ with }\quad \widehat{\Lambda}_{\mathrm{mod}} =
\begin{pmatrix}
\mathrm{X}^{\T}\mathrm{X}/n \;& \mathrm{X}^{\T} \widehat{\mathrm{Z}}/n \\
\widehat{\mathrm{Z}}^\T \mathrm{X}/n & \mathrm{S}_n
\end{pmatrix},
\end{align}
where the modified block matrix $S_n$ contains the associated leading terms from the representation in Proposition \ref{main:prop1}. The following result establishes that our bias-corrected estimator substantially enhances the range of sparsity levels for which inference is possible for $\beta^*$.      

\begin{theorem}[Asymptotic Normality of the Bias-corrected OLS Estimator for $\beta^*$]\label{theorem 3} Suppose that $\Lambda = E(LL^\T)$ is invertible, $E(Y^4) < \infty$, and $E(\|X\|^4) < \infty$. Further, suppose that one of the following conditions holds:
\begin{enumerate}
\item[(a)] All subgraphs are acyclic and $\lambda_n =  \omega(1)$.
\item[(b)] All subgraphs are acyclic or simple cycles, and $\lambda_n^{r_j} = \omega(n)$ for all $j$ corresponding to simple cycles.
\item[(c)] The sparsity condition in (b) holds for simple cycles and $\lambda_n = \omega(n^{\,1 - 2/\{\,2R-1\}})$ holds for subgraphs that are neither acyclic nor simple cycles. 
Then,
\begin{align*}
n^{1/2}(\widehat{\beta}_{\mathrm{mod}}-\beta^*) \rightarrow N(0,\Sigma_\beta) \quad \text{ in distribution.}
\end{align*}
\end{enumerate}
\end{theorem}

Therefore, in sparse regimes, one may achieve $n^{1/2}$-consistent inference either by targeting a less interpretable estimand $\widetilde{\beta}$ or by implementing a bias correction. See Figure \ref{fig:beta_estimates} in Section \ref{simulation: sec1} for an empirical comparison of $\widehat{\beta}$ and $\widehat{\beta}_{mod}$. Note that $\Sigma_\beta$ is unknown and needs to be estimated for inferential purposes.  In the next section, we consider statistical inference for regression parameters with local subgraph frequencies using the bootstrap.  The bootstrap offers a user-friendly method that allows practitioners to implicitly estimate the covariance matrix via simulation. 

\subsection{Bootstrap Inference with Local Count Statistics}\label{sec-3.2:bootstrap_consistency}

For bootstrap inference, we jointly resample the OLS statistics and the estimated sparsity parameter. To make things precise, we define
\(\widehat{\Psi}_n=(\operatorname{vec}(\widehat{\Lambda}),\widehat{\gamma})\) and \(\widehat{\Psi}^{\mathrm{mod}}_n=(\operatorname{vec}(\widehat{\Lambda}_{\mathrm{mod}}),\widehat{\gamma})\) and adapt the linear multiplier bootstrap of \citet{lin2020trading} with multipliers \(\{W_1, W_2, \cdots, W_n\}\) that are i.i.d. standard normal random variables. First, consider the following bootstrap analogs of $\widehat{\Psi}_n$ and $\widehat{\Psi}^{\mathrm{mod}}_n$:
\begin{align}\label{subgraph:bootstrap_OLS-stage1}
\widehat{\Psi}_{n}^\sharp:=(\operatorname{vec}(\widehat{\Lambda}^\sharp),\widehat{\gamma}^\sharp)=\widehat{\Psi}_n+\frac{1}{n}\sum_{i=1}^{n}(W_i-1)\,D_k\widehat{g}_{1,\Psi}(i), \\
\widehat{\Psi}_{n}^{\mathrm{mod},\sharp}:=(\operatorname{vec}(\widehat{\Lambda}^\sharp_{\mathrm{mod}}),\widehat{\gamma}^\sharp)=\widehat{\Psi}^{\mathrm{mod}}_n+\frac{1}{n}\sum_{i=1}^{n}(W_i-1)\,D_k\widehat{g}_{1,\Psi}(i),    
\end{align}
where \(D_k=\mathrm{diag}(k_1,\ldots,k_q)\) collects the order of the corresponding underlying U-statistic and \(\widehat{g}_{1,\Psi}(i)\) approximates the first-order projection term (see Supplement \ref{appB-sec:bootstrap}), and $q$ is the dimension 
of $\widehat{\Psi}_n$.
These bootstrapped quantities do not account for fluctuations due to the estimated sparsity parameter $\hat{\rho}_n/\rho_n$. Consider the following bootstrap analog of the sparsity level:
\begin{align*}
\frac{\hat{\rho}_n^\flat}{\hat{\rho}_n}
=1+\frac{2}{n}\sum_{i=1}^n(W_i-1)\Bigl(\tfrac{1}{(n-1)\hat{\rho}_n}\sum_{j\neq i}A_{ij}-1\Bigr). 
\end{align*}

To account for fluctuations from estimating $\rho_n$, we can define new bootstrap quantities $\widehat{\Psi}_{n}^\flat$ and $\widehat{\Psi}_{n}^{\mathrm{mod},\flat}$ that incorporate this bootstrapped sparsity. Formally, we can express these quantities using the following functions: 
\begin{align*}
\widehat{\Psi}_{n}^\flat 
   &:= (\widehat{\Lambda}^\flat,\,\widehat{\gamma}^\flat)
     = g\!\left(\operatorname{vec}(\widehat{\Lambda}^\sharp),\,\widehat{\gamma}^\sharp,\,\hat{\rho}_n^\flat/\hat{\rho}_n\right), 
&
\widehat{\Psi}_{n}^{\mathrm{mod},\flat} 
   &:= (\widehat{\Lambda}^{\mathrm{mod},\flat},\,\widehat{\gamma}^\flat)
     = g\!\left(\operatorname{vec}(\widehat{\Lambda}^\sharp_{\mathrm{mod}}),\,\widehat{\gamma}^\sharp,\,\hat{\rho}_n^\flat/\hat{\rho}_n\right),
\end{align*}
where $g:(m,b,r)\mapsto (\mathrm{M}^{(r)},b^{(r)})$ with 
$m^{(r)} = (\,m_\ell r^{\alpha_{m,\ell}}\,)_{\ell=1}^{q^2},\quad
b^{(r)} = (\,b_t r^{\alpha_{b,t}}\,)_{t=1}^{q},\quad
\mathrm{M}^{(r)}=\operatorname{mat}_q(m^{(r)}),$
for fixed $\alpha=(\alpha_m,\alpha_b)$ with $\alpha_m\in\mathbb{R}^{q^2}$ and $\alpha_b\in\mathbb{R}^q$. 
Here $\operatorname{mat}_q:\mathbb{R}^{q^2}\to\mathbb{R}^{q\times q}$ is the inverse of the $\operatorname{vec}$ operator. 
For a linear subgraph statistic $n^{-1}\sum_{i=1}^n Y_i \widehat{Z}_{ij}$, we have $|\mathcal{E}(\mathcal{R}_j)|=s_j$ 
and hence $\alpha_{b,t} = s_j$.  
For a quadratic subgraph statistic $n^{-1}\sum_{i=1}^n \widehat{Z}_{ij}\widehat{Z}_{ik}$ with 
$|\mathcal{E}(\mathcal{R}_j)|=s_j$ and $|\mathcal{E}(\mathcal{R}_k)|=s_k$, 
it suffices to set $\alpha_{m,\ell} = s_j+s_k$ since we only bootstrap the signal component of the quadratic term. Finally, the bootstrapped regression coefficients are given by:
\begin{align}\label{subgraph:bootstrap_beta}
\widehat{\beta}^\flat=f(\widehat{\Psi}_{n}^\flat), 
\qquad 
\widehat{\beta}^\flat_{\mathrm{mod}}=f(\widehat{\Psi}_{n}^{\mathrm{mod},\flat}),    
\end{align}
where $f(\mathrm{M},b)= \mathrm{M}^{-1}b$. We have the following result:

\begin{theorem}[Validity of Linear Multiplier Bootstrap]\label{graphon:bootstrap}
Under the same conditions as in Theorem \ref{theorem 2}, 
\begin{enumerate}
\item[(a)] The bootstrap consistency holds for the bias-corrected OLS estimator \(\widehat{\beta}_{\mathrm{mod}}^\flat\), i.e., 
\begin{align*}
    \sup_{u \in \mathbb{R}^{p+d}} \Big| 
    \operatorname{pr}^{\flat}\left( n^{1/2}(\widehat{\beta}^{\flat}_{\mathrm{mod}} - \widehat{\beta}_{\mathrm{mod}}) \leq u \, \right) 
    - \operatorname{pr}\left( n^{1/2}(\widehat{\beta}_{\mathrm{mod}} - \beta^{*}) \leq u \right) 
    \Big| \rightarrow 0 \quad \text{in probability}.
\end{align*}
\item[(b)] The bootstrap consistency also holds for the OLS estimator \(\widehat{\beta}^\flat\), i.e.,
\begin{align*}
    \sup_{u \in \mathbb{R}^{p+d}} \left| 
    \operatorname{pr}^{\flat}\left(  n^{1/2}(\widehat{\beta}^{\flat} - \widehat{\beta}) \leq u \, \right) 
    - \operatorname{pr}\left( n^{1/2}(\widehat{\beta} - \widetilde{\beta}) \leq u \right) 
    \right| \rightarrow  0 \quad \text{in probability}.
\end{align*}
\end{enumerate}
\end{theorem}

For bootstrap consistency with Gaussian multipliers, the main step is establishing the consistency of an appropriate bootstrap covariance matrix. 
This includes handling complex higher-order (fourth-moment) noisy terms. We tackle this using the representation in Proposition \ref{main:prop1}, which effectively reduces these higher-order terms to an equivalent global subgraph problem. This enables us to apply techniques similar to those in \citet{lin2020trading} and \citet{zhang2022edgeworth} to establish variance consistency. The proof to establish bootstrap consistency is algebraically involved and has been deferred to Supplement \ref{appB-sec:bootstrap}.  Although it is not the focus of this paper, it is also possible to consider fast approximate bootstrap procedures as in \citet{lin2020trading}.

\subsection{Inference with Adjacency Spectral Embeddings Under the Generalized Random Dot Product Graph Model}\label{subsec:rdpg-reg}
We now consider statistical inference for projection parameters when the network covariates are spectral embeddings related to GRDPG models.  Spectral embeddings are estimates of a node's latent positions for a subclass of graphons that are low rank.  We provide definitions of GRDPGs and spectral embeddings below. While there are still identifiability issues associated with the network covariates under this sub-model, certain inferential tasks for projection parameters are possible when these embeddings are used as covariates.  Incorporating these embeddings into a regression model can potentially offer a more expressive summary of each node's global structural role compared to local motifs and serve as more informative control variables. We now define the sparse generalized random dot product graphs below.  The following representation may be viewed as a special case of the graph root representation of \citet{lei2021network}.

\begin{definition}(Sparse Generalized Random Dot Product Graph Model)\label{def:rdpg}
Suppose that the graphon $w(x,y)$ admits a finite-dimensional spectral decomposition of the form:
\begin{align*}
w(x,y) &= \sum_{r=1}^d \lambda_r \phi_r(x) \phi_r(y) = \sum_{i=1}^{r^{+}} \lambda_i^+ \phi_i^+(x) \phi_i^+(y) - \sum_{j=1}^{r^{-}} |\lambda_j^-| \phi_j^-(x) \phi_j^-(y),
\end{align*}
where $r^{+}+r^{-} =d$, $|\lambda_1| \geq \cdots \geq |\lambda_d| > 0$, $\lambda_i^+$ corresponds to the $ith$ largest positive eigenvalue and $\phi_i^+$ denotes its corresponding eigenfunction, and  $\lambda_j^-$ corresponds to the $j$th largest negative eigenvalue in magnitude and $\phi_j^-$ denotes its corresponding eigenfunction.  Under this assumption, if we let $\chi_{ij} = (\lambda_{j}^+)^{1/2} \phi_j^+(\xi_i)$ for $j \in \{1,\ldots r^{+}\}$ and $\zeta_{ik} = |\lambda_k^-|^{1/2} \phi_k^-(\xi_i)$ for $k \in \{1,\ldots r^{-}\}$, the sparse graphon in (\ref{eq-sparse-graphon}) can be parameterized as:
\begin{align}
\label{eq-grdpg-eq}
A_{ij}^{(n)} \sim \mathrm{Bernoulli}( \rho_n (\langle \chi_i, \chi_j  \rangle - \langle \zeta_i, \zeta_j  \rangle) \wedge 1). 
\end{align}

We refer to a data generating process for which the graphon admits a finite-dimensional spectral decomposition as a generalized random dot product graph (GRDPG) model. GRDPGs with $r^{-}=0$ will be referred to as RDPGs.
\end{definition}

With the above parametrization, it is of interest to estimate embeddings that, in a certain sense, approximate the latent positions.  Let $Z_i = ( \sigma_1^{1/2} \phi_1(\xi_i), \ldots, \sigma_d^{1/2} \phi_d(\xi_i))$, where $\sigma_j$ is the $j$th largest singular value associated with the graphon $w$.  This vector contains elements of $(\chi_i, \zeta_i)$ rearranged by the magnitude of their corresponding eigenvalues.  While $(Z_1, \ldots, Z_n)$ are not identifiable, they can be estimated up to an unknown orthogonal rotation under mild sparsity conditions by an adjacency spectral embedding, which we define below.        

\begin{definition}(Adjacency Spectral Embedding (ASE))\label{def:ase}
Let $|\mathrm{A}| = (\mathrm{A}^\T \mathrm{A})^{1/2}$ and $|\mathrm{A}| = \widehat{\mathrm{U}} \widehat{\mathrm{S}} \widehat{\mathrm{U}}^\T$ denote the spectral decomposition of $|\mathrm{A}|$.  Let $\widehat{\mathrm{S}}^{(d)} \in \mathbb{R}^{d \times d}$ be a diagonal matrix containing the $d$ largest eigenvalues of $|\mathrm{A}|$ and $\widehat{\mathrm{U}}^{(d)} \in \mathbb{R}^{n \times d}$ be a matrix in which the columns contain eigenvectors corresponding to the top $d$ eigenvalues. The adjacency spectral embedding is given by:   
\[
\widehat{\mathrm{Z}} = \hat{\rho}_n^{-1/2}\widehat{\mathrm{U}}^{(d)} |\widehat{\mathrm{S}}^{(d)}|^{1/2} = [\widehat{Z}_1^{\T} \ \widehat{Z}_2^{\T}  \ \cdots \ \widehat{Z}_n^{\T} ]^\T \in \mathbb{R}^{n \times d},
\]
where  $\widehat{Z}_i \in \mathbb{R}^d$ is the $d$-dimensional embedding for node $i$ and $\hat{\rho}_n$ is the same estimator considered in the definition of $\widehat{Q}(\mathcal{R})$.
\end{definition}
\begin{remark}
In the GRDPG literature, it is more common to target  $(\chi_1, \lambda_1), \ldots, (\chi_n, \lambda_n)$.  However, arguments to establish node-wise central limit theorems for related embeddings often hold only up to an indefinite orthogonal rotation (see, for example \citet{rubin2022statistical}), which complicates subsequent inference.  
\end{remark}

In the following subsections, we will consider a central limit theorem for projection parameters associated with linear regression models that take as input ASEs that hold up to an unknown linear transformation, and subsequently consider inference tasks for which the identifiability issue can be circumvented.    

\subsection{Asymptotics of the OLS Estimator}
Before stating a central limit theorem for projection parameters, we prepare some notations.  Suppose that we have $p-1$ conventional covariates and a $d$-dimensional GRDPG embedding.  Define the matrix:  
\[
\mathrm{M}_n =
\begin{bmatrix}
\mathrm{I}_p & 0 \\
0 & \mathrm{Q}_n
\end{bmatrix}
\in \mathbb{R}^{(p+d) \times (p+d)},
\]
where \( \mathrm{I}_p \) is the \( p \times p \) identity matrix and $\mathrm{Q}_n\in \mathbb{O}(d)$ is an appropriate rotation matrix.  We have the following result: 

\begin{theorem}\label{theorem 5}
 Suppose that $\Lambda = E(LL^\T)$ is invertible, $E(Y^4) < \infty$, and $E(\|X\|^4) < \infty$. Further suppose that \( \lambda_n = \omega(n^{1/2}\log^{2c} n) \), for some universal constant $c >1$.   Then, there exists a sequence of orthogonal  rotation matrices $(\mathrm{Q}_n)_{n \geq 1}$ and a  corresponding sequence of linear transformations $(\mathrm{M}_n)_{n \geq 1}$ such that:  
\[
n^{1/2}  (\mathrm{M}_n\widehat{\beta} - \beta^*) \rightarrow N ( 0, \Sigma_{\beta} ) \quad \text{in distribution}.
\]
\end{theorem}  
\begin{remark}
In the above theorem, we consider the target parameter $\beta^*$, but not $\widetilde{\beta}$.  The estimated spectral embeddings are not uniquely defined, making $\widetilde{\beta}$ an ambiguous target.       
\end{remark}

\begin{remark}
Our theorem above implies that the identifiability issues associated with the ASE lead to difficulties for estimation and inference of regression coefficients corresponding to spectral embeddings. However, coefficients corresponding to conventional covariates are unaffected by this rotation, and both estimation and inference are possible.   
\end{remark}

\begin{remark}
In the assumption-lean framework, a natural question is whether inferences are valid under misspecification of the network model.  If the underlying GRDPG is of higher rank than the embedding dimension, then the above theorem still holds.  If the underlying model is infinite rank, recent results by \citet{tang2025eigenvectorfluctuationslimitresults} may be invoked in our proofs to derive an analogous result under additional technical conditions.   
\end{remark}

As an intermediate step of our proof, we establish asymptotic normality in particular for 
$n^{-1/2}\sum_{i=1}^n\{\mathrm{Q}_n\widehat{Z}_i - E[Z_i]\}$ 
for some sequence of orthogonal matrices $(\mathrm{Q}_n)_{n \geq 1}$. To achieve this, we adapt arguments for node-wise central limit theorems for $\mathrm{Q}_n\widehat{Z}_i - Z_i$ developed by \citet{athreya2018statistical} and \citet{rubin2022statistical}, but exploit the independence of appropriate signal terms when averaging over nodes. Central limit theorems for means of functions of spectral embeddings have previously been studied by \citet{90ab183a-a1e8-320b-b810-5ecbf33e96cb}. Although these authors are able to prove a functional central limit theorem, these results are restricted to the dense case and also require stronger conditions on the underlying GRDPG model.  For terms of the form $n^{-1/2}\sum_{i=1}^n\{\mathrm{Q}_nY_i\widehat{Z}_i- E(Y_iZ_i)\}$ and $n^{-1/2}\sum_{i=1}^n\{\mathrm{Q}_n\widehat{Z}_i\widehat{Z}_i^\T\mathrm{Q}_n^\T  - E(Z_iZ_i^\T)\}$, we are able to derive central limit theorems under weaker regularity conditions.    

While a delta method argument is invoked to derive asymptotic normality of the OLS estimator, it should be noted that we make use of properties of the OLS functional beyond differentiability. OLS estimators are invariant under invertible linear transformations in the sense that if $\widehat{\beta} = (\mathrm{X}^\T \mathrm{X})^{-1} \mathrm{X}^\T\mathrm{Y}$, and one defines $\widetilde{\mathrm{X}} = \mathrm{X}\mathrm{M}$ for some invertible $\mathrm{M}$, then $(\mathrm{\widetilde{X}}^\T \mathrm{\widetilde{X}})^{-1} \mathrm{\widetilde{X}}^\T\mathrm{Y} = \mathrm{M}\widehat{\beta}$.  This property, which does not hold for general differentiable functionals, allows us to move $\mathrm{M}$ outside the OLS functional and construct test statistics for the presence of network effects.

\subsection{Bootstrap Inference with Spectral Embeddings}\label{bootstrap:rdpg}
Our theory suggests that spectral embeddings have a different dependence structure compared to local subgraph statistics.  In particular, these embeddings are asymptotically independent, motivating the development of multiplier bootstraps for independent data in this setting. Let \(W_1, W_2, ..., W_n\) be i.i.d. Gaussian multiplier random variables satisfying \(E(W_i)=0\) and \(\text{var}(W_i)=1\) for all \(i\). Consider the following bootstrap quantities:  
\begin{align*}
(n^{-1}\widehat{\mathrm{L}}^{\T}\widehat{\mathrm{L}})^{\flat}=\left(\frac{\hat{\rho}_n}{\hat{\rho}_n^\flat}\right)\left\{n^{-1}\widehat{\mathrm{L}}^{\T}\widehat{\mathrm{L}}+n^{-1}\sum_{i=1}^{n} W_i \cdot (\widehat{L}_i\widehat{L}_i^{\T}-n^{-1}\widehat{\mathrm{L}}^{\T}\widehat{\mathrm{L}})\right\},\\ 
(n^{-1}\widehat{\mathrm{L}}^{\T}Y)^{\flat}=\left(\frac{\hat{\rho}_n}{\hat{\rho}_n^\flat}\right)^{1/2}\left\{n^{-1}\widehat{\mathrm{L}}^{\T}Y+n^{-1}\sum_{i=1}^{n}W_i\cdot (\widehat{L}_i Y_i-n^{-1}\widehat{\mathrm{L}}^{\T}Y)\right\}.
\end{align*}
The bootstrapped version of the regression coefficient is then defined as:
\begin{align*}
\widehat{\beta}^{\flat}=\big\{(n^{-1}\widehat{\mathrm{L}}^{\T}\widehat{\mathrm{L}})^{\flat}\big\}^{-1} (n^{-1}\widehat{\mathrm{L}}^{\T}Y)^{\flat}.
\end{align*}

The bootstrapped estimator \( \widehat{\beta}^{\flat} \), constructed from the observed data, inherits identifiability issues from the original estimator.  However, it will turn out that the same linear transformation $\mathrm{M}_n$ can be used to simultaneously align both $\widehat{\beta}^\flat$ and $\beta^*$.  This property will be crucial for developing bootstrap-based tests for the presence of network effects. We have the following result:

\begin{theorem}\label{rdpg:theorem6}(Validity of the Rotation-aligned Bootstrap Estimator) Under the same conditions as Theorem \ref{theorem 5}, there exists a sequence of linear transformations $(\mathrm{M}_n)_{n \geq 1}$ depending only on the data such that: 
 \begin{align*}
   \sup_{u \in \mathbb{R}^{p+d}} \left|  \operatorname{pr}^\flat\left(n^{1/2}\mathrm{M}_n(\widehat{\beta}^{\flat}-\widehat{\beta}) \leq u\right) - \operatorname{pr}\left( \Sigma_\beta^{1/2} Z \leq u \right)\right| \rightarrow 0 \quad \text{in probability},  
 \end{align*}
 where $Z \sim N(0, I_{p+d})$.
\end{theorem}

\begin{remark}
There may be cases in which both local subgraph frequencies and spectral embeddings are desired as network covariates. Since our theory justifies the use of an independent bootstrap for spectral embeddings, they may be bootstrapped similarly to the conventional covariates in the bootstrap procedure introduced in Section \ref{sec-3.2:bootstrap_consistency}.  
\end{remark}

It should be noted that we cannot leverage global representations as we did for subgraphs, so our analysis here is different from the proofs of results in Theorem \ref{graphon:bootstrap}.  

We now consider hypothesis testing for assessing the existence of network effects. In what follows, let $\beta_{\mathrm{z}}^*$ denote the vector of projection parameters of the form (\ref{clean beta}) associated with the spectral embeddings.  Suppose that we are interested in testing $H_0: \beta_{\mathrm{z}}^* =0$ against $H_1: \beta_{\mathrm{z}}^* \neq 0$.     
Consider the test statistic:
\begin{align}\label{rdpg:observed testing statistic}
\widehat{T}_{test}(\beta_{\mathrm{z}}^*) = n\|\widehat{\beta}_{\mathrm{z}}- \beta_{\mathrm{z}}^*\|_2^2.    
\end{align}

To find a critical value for the above test statistic under $H_0$, we will use the following bootstrapped statistic:
\begin{align}\label{rdpg:bootstrapped testing statistic}
\widehat{T}^{\flat}_{test}(\widehat{\beta}_{\mathrm{z}}) = n\|\widehat{\beta}^{\flat}_{\mathrm{z}} - \widehat{\beta}_{\mathrm{z}}\|_2^2.    
\end{align}

Our result below establishes that, under the null hypothesis of no network effect, the bootstrap distribution of \( \widehat{T}^{\flat}_{test}(\widehat{\beta}_{\mathrm{z}}) \) provides a valid approximation of the sampling distribution of \( \widehat{T}_{test}(0) \), thereby serving as a reference distribution for inference. In what follows, let $c_{1-\alpha}^\flat$ denote the $(1-\alpha)$-quantile of $\widehat{T}^{\flat}_{test}$. We have the following corollary: 

\begin{corollary}[Validity of the Bootstrap Approximation for Testing Network Effect]\label{hypothesis testing network}
    Suppose that \( \lambda_n = \omega( n^{1/2}\log^{2c} n ) \), for some universal constant $c$.  Then under $H_0: \beta_{\mathrm{z}}^* = 0$, 
    \[
    \sup_{u \in \mathbb{R}} \left| \operatorname{pr}^\flat\left(\widehat{T}^{\flat}_{test}(\widehat{\beta}_{\mathrm{z}}) \leq u \right) - \operatorname{pr}\left(\widehat{T}_{test}(\beta_{\mathrm{z}}^*) \leq u \right) \right| \to 0 \quad \text{in probability}.
    \]
     Moreover, for any \( \alpha \in [0,1] \),
    \[
    \operatorname{pr}\left( \widehat{T}_{test}( \beta^{*}_{\mathrm{z}}) > c_{1-\alpha}^\flat \right) \to \alpha.
    \]
\end{corollary}
\begin{remark}
Our test statistic leverages the fact that for any $\mathrm{Q} \in \mathbb{O}(d)$, $\|\widehat{\beta}_{\mathrm{z}} \|_2^2 = \|\mathrm{Q}(\widehat{\beta}_{\mathrm{z}} - \beta_{\mathrm{z}}^*) \|_2^2$ holds when $\beta_{\mathrm{z}}^* = 0$.  It is also possible to consider the test statistic $\|\widehat{\beta}_{\mathrm{z}}\|^2_{2}-\|\beta_{\mathrm{z}}^*\|^2_{2}$, but it turns out that a scaled version of this test statistic is degenerate when $\beta_{\mathrm{z}}^* = 0$, which can be shown to lead to bootstrap inconsistency.       
\end{remark}

\section{Down-Sampling for Handling More Flexible Covariate Types and Challenging Sparsity Regimes} \label{sec:down-sampling}

In previous sections, we established $n^{1/2}$-consistency results in cases where the network covariates involved local subgraph frequencies or GRDPG embeddings.  While these statistics are widely used, many other covariate types are often considered for network-assisted regression (see for example, \citet{https://doi.org/10.1002/sam.11486}) and their dependency structures are not always well understood.  Even for spectral embeddings, our asymptotic normality results require \( \lambda_n = \omega(n^{1/2}\log^{2c} n) \).  It is natural to ask whether statistical inference is possible in sparser regimes if one is willing to accept a slower rate of convergence.       

To address these difficulties, we propose a down-sampling method that enables statistical inference based on the normal approximation under milder regularity conditions.  When $n^{1/2}$-consistency is achievable, the tradeoff is a loss of efficiency, but this approach allows inference in cases where such guarantees are not available.  We now define the down-sampled OLS estimator.  Let $m \ll n$ denote the down-sample size, and $(Y_1, \widehat{L}_1), \ldots (Y_m, \widehat{L}_m)$ denote the down-sample of interest.  The corresponding OLS is given by:   
\begin{align}\label{def:down-sampling estimator}
\widehat{\beta}^{(m)} = \argmin_{\beta \in \mathbb{R}^{p+d}} \frac{1}{m} \sum_{i=1}^m \left( Y_i - \widehat{L}_i^\T \beta \right)^2.
\end{align}

Intuitively, the network covariate $\widehat{Z}_i$ is still estimated using the entire dataset; if $m$ is chosen appropriately, then the fluctuations of $\widehat{Z}_i - Z_i$ would be negligible when the OLS estimator is centered and scaled by $m^{1/2}$.  We now state mild, simple sufficient conditions for our down-sample estimator.  

\begin{condition}[Choice of Down-Sample Size]\label{down-sampling:condition2}
Suppose there exist i.i.d. random variables $Z_1, \ldots, Z_n$ and a sequence $a_n \rightarrow 0$ such that: 
\begin{align}
\max_{1 \le i \le n} \|\widehat{Z}_i - Z_i\| = O_P(a_n). 
\end{align}
Moreover, $m_n$ is chosen so that $m_n \rightarrow \infty$, $m_n=o(n)$, while satisfying $a_n=o(m_n^{-1/2})$.
\end{condition}

We are now ready to state the $m^{1/2}$-scaled central limit theorem for the down-sampled OLS estimator:

\begin{theorem}[Asymptotic Normality of the Down-Sampled Estimator Targeting $\beta^*$]\label{thm:down-sampling}
Suppose Condition 
\ref{down-sampling:condition2} holds. Further,  assume that $\Lambda = E(LL^\T)$ is invertible, $E(Y^4) < \infty$, and $E(\|X\|^4) < \infty$. Then,
\begin{align}\label{eq-down-sample-clt}
m^{1/2} \left( \widehat{\beta}^{(m)} - \beta^* \right) \rightarrow N(0, \Sigma^{(m)}_\beta) \quad \text{in distribution},    
\end{align}
where  \( \Sigma^{(m)}_\beta \) is defined in Section \ref{subsec-appedix-down-sampling}.
\end{theorem}

Since the limiting distribution is not affected by network noise, multiplier bootstraps for independent data may be used to conduct inference for down-sampled regression estimators. The following proposition establishes alternative sufficient conditions under which the estimator admits a central limit theorem. In this case, $\beta^*$ and $\widetilde{\beta}$ are asymptotically equivalent if one can choose the down-sample size to sufficiently control an appropriate mean squared error.  
\begin{proposition}[CLT via Moment Control and Target Equivalence]\label{coro:2}
Assume that $\Lambda = E(LL^\T)$ is invertible, $E(Y^4) < \infty$, and $E(\|X\|^4) < \infty$. Suppose, in addition, that the network noise satisfies $
E\|\widehat{Z} - Z\|^2 = o(1/m)$.
Then the following holds:
\[
m^{1/2} \left( \widehat{\beta}^{(m)} - \beta^* \right) \rightarrow N(0, \Sigma^{(m)}_\beta) \quad \text{ in distribution}, \quad m^{1/2}(\widetilde{\beta} - \beta^*) \to 0.
\]
\end{proposition}
\begin{remark}
With an appropriate choice of down-sample size, it is also possible to mitigate the bias issue associated with local subgraph frequencies.  However, this comes at the cost of a reduced rate of convergence.  In additional experiments not presented here, we also see the bias-corrected estimator  $\widehat{\beta}_{\mathrm{mod}}$ reduces the bias much more effectively.
\end{remark}

We next introduce two examples that demonstrate how down-sampling can be used to attain valid inference in challenging regimes.  In both cases, the theory provides guidance for the choice of down-sample size in terms of $\lambda_n$, which can be estimated from the data.  The first example involves nonlinear functions of local subgraph counts.  

\begin{example}[Composite Local Subgraph Frequency Measures under Graphon Models] 
\label{ex:graphon_ratio}
Consider local composite measures of the following form, defined for the $j$th coordinate of $\widehat{Z}_{ij}$ at node $i$:
\begin{align}\label{def:general ratio subgraph}
\widehat{Z}_{ij} =  \widehat{U}_i(\mathcal{R}_1)/ \widehat{U}_i(\mathcal{R}_2),    
\end{align}
where both $\widehat{U}_i(\mathcal{R}_1)$ and $\widehat{U}_i(\mathcal{R}_2)$ are simple (possibly covariate-weighted) normalized subgraph frequencies for node $i$ containing $(r_1, r_2)$ number of nodes and $(s_1, s_2)$ number of edges respectively for the numerator and denominator statistic; for each, see definition \eqref{eq-local-subgraph}. Key examples include:
\begin{itemize}
    \item[(a)] Local Transitivity (\(s_1=3, s_2=2\)):
    \[
    \widehat{Z}_{ij} = \frac{\left(\binom{n-1}2\hat{\rho}_n^3\right)^{-1}3\sum_{j<k} A_{ij}A_{jk}A_{ki}}{\left(\binom{n-1}{2}\hat{\rho}_n^2\right)^{-1}\sum_{j<k} (A_{ij}A_{ik}+A_{ji}A_{jk}+A_{ki}A_{kj})}.
    \]
    \item[(b)]  Local Neighborhood Average (\(s_1=1, s_2=1\), weighted by covariate $X_j$ among neighborhood):
    \[
    \widehat{Z}_{ij} = \frac{\sum_{j \neq i} A_{ij}X_j}{\sum_{j \neq i} A_{ij}}.
    \]
\end{itemize}
While such statistics are of substantial interest, they do not fall under the framework considered in Section \ref{subsec-local-subgraph}. Define the corresponding noiseless network covariate as:
\begin{align*}
Z_{ij} = \frac{E[U_i(\mathcal{R}_1;\xi)\mid \xi_i]}{E[U_i(\mathcal{R}_2;\xi)\mid \xi_i]},    
\end{align*}
where $U_i(\mathcal{R}_1;\xi)$ and $U_i(\mathcal{R}_2;\xi)$ admit the form of $\widetilde Q_i(\mathcal{R})$, introduced by definition \eqref{latent_node_stat}, possibly weighted by the associated covariate.

The next proposition provides conditions under which the conclusion of Theorem \ref{thm:down-sampling} holds for Example \ref{ex:graphon_ratio}.  In what follows, we assume a lower- and upper-bounded graphon, but these conditions can be relaxed at the cost of longer proofs and more conservative conditions on the down-sample size. 
\begin{proposition}[Uniform Noise Control for Local Composite Covariate] 
\label{prop:composite_control}
Suppose that \( 0 < c \leq w(\xi_i,\xi_j) \leq C <\infty \) almost surely for some $C \geq c >0$. Furthermore, suppose that one of the following conditions holds:
\begin{enumerate}[label=(\roman*)]
\item[(a)] For transitivity, the down-sample size satisfies $m^2=o(\lambda_n \wedge \lambda_n^3/n).$ 
\item[(b)] For neighborhood average, the down-sample size is chosen as $m^2=o(\lambda_n)$. 
\end{enumerate}
Then, the conditions of Theorem \ref{thm:down-sampling} are satisfied and (\ref{eq-down-sample-clt}) holds.
\end{proposition}
\end{example}
The next example demonstrates that down-sampling can considerably relax the sparsity conditions required for regression with spectral embeddings at the cost of a reduced rate of convergence. 

\begin{example}[Adjacency Spectral Embedding in Sparser GRDPG]\label{ds-example:rdpg}
Consider a generalized random dot product graph (GRDPG) with $\rho_n = o(\log^{2c} n/n^{1/2})$, where the $n^{1/2}$-consistency conditions of Theorem~\ref{theorem 5} fail. Down-sampling still yields asymptotic normality if $m$ is chosen appropriately and $\rho_n = \omega(\log^c n / n)$ for some $c > 1$. From the $2\to\infty$ norm bounds in \citet{generalized-rdpg}, if
\[
m = o\!\left(\lambda_n/\log^{2c} n\right),
\]
then Theorem~\ref{thm:down-sampling} applies and \eqref{eq-down-sample-clt} holds. While Theorem~\ref{thm:down-sampling} provides a widely applicable bound, exploiting additional knowledge about the network covariate’s  structure (in this case, adapting arguments used to prove the CLT for spectral embeddings) yields a sharper bound and the wider regime: 
\begin{align}
m = o\!\left((\lambda_n^2/\log^{4c} n) \wedge n\right).    
\end{align}
See Supplement~\ref{appC-section:down-sampling_grdpg} for further discussion.
\end{example}

\section{Simulation Studies}
This section investigates the finite-sample behavior of the proposed inference framework across a variety of network-linked data settings. In Section~\ref{simulation: sec1}, we begin by evaluating the regression model with local subgraph statistics under sparse graphon models, where we find a notable instability for the uncorrected sample OLS estimator even for relatively dense graphs. In Section~\ref{simulation: sec2}, we analyze the regression model where the latent positions from the generalized random dot product graph (GRDPG) model  contribute to the regression model. Despite the non-identifiability of individual latent positions, valid inference remains feasible for detecting overall network effects and for obtaining coverage of the conventional covariates. Finally, Section~\ref{simulation: sec3} explores inference under model misspecification, where the outcome model involves nonlinear interactions between network and conventional covariates. While the linear projection parameters no longer coincide with structural parameters, we demonstrate that they still provide meaningful estimates and yield improved coverage for the corrected estimator across all sparsity regimes.

\subsection{Regression with Local Subgraph Estimator under Sparse Graphon Model}\label{simulation: sec1}
We consider the following random graph model: a graphon with kernel function $w(u,v) =(u+v)^{-1}$. The outcome variable $Y_i$ follows a linear relationship with the covariates $X_i$ through the equation $Y_i = 1 + 20 Z_i + 3 X_{i_1} + 2 X_{i_2}+\varepsilon_i$, where $X_{i}$ follows a bivariate normal distribution with mean $\mu = (1,3)$ and covariance matrix $\Sigma = (1, 0.6, 0.6, 4)$, and $
\varepsilon_i \overset{\text{i.i.d.}}{\sim} N(0, 1)$. For simulation purposes, we use a rooted two-star as the local subgraph of interest, where the corresponding $Z_i$ is given by $Z_i= E((w(\xi_i, \xi_j) w(\xi_i, \xi_k) \ | \ \xi_i)$. The sparsity levels for the sparse graphon are set as \( \rho_n = n^{\delta} \), with \( \delta \in \{-0.75, -0.4, -0.25\} \), corresponding to very sparse, moderately sparse, and weakly sparse regimes, respectively. We generate a random sample of size $n = 4000$ for each of these random graph models as described above. Figure \ref{fig:beta_estimates}. illustrates the sampling distribution from 200 Monte Carlo iterations for the estimated coefficient associated with the network covariate. The shaded sampling distribution provides a comparison between bias-corrected OLS estimates with its non-corrected counterpart. 
\begin{figure}[htbp]
    \centering
    \includegraphics[width=0.8\textwidth]{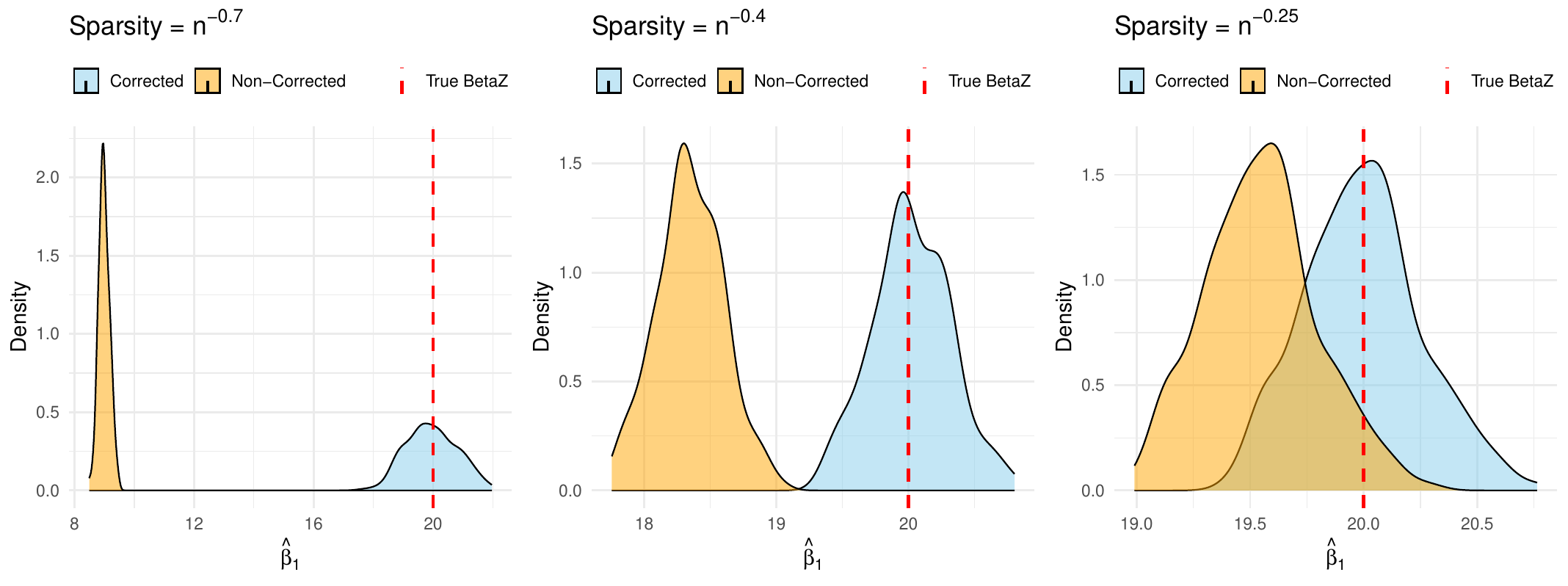}
    \captionsetup{font=scriptsize}
    \caption{Sampling distributions of the second regression coefficient for corrected and non-corrected estimators across different sparsity levels. The red dashed line indicates the true value $\beta_{\mathrm{z}} = 20$.}
    \label{fig:beta_estimates}
\end{figure}

Figure \ref{fig:beta_estimates}  demonstrates that the regression coefficient associated with the network covariate is accurately approximated by the corrected least squares estimator $\widehat{\beta}$ across all levels of sparsity. In contrast, the uncorrected estimator deviates substantially from the true value $\beta_{\mathrm{z}} = 20$, particularly in sparser regimes.  In fact, the overlap between the two sampling distributions only occurs in the tails even for moderately dense regimes for $n=4000$, highlighting the importance of bias corrections for targeting $\beta^*$.

\subsection{Bootstrap Inference under the GRDPG  Model}\label{simulation: sec2}
Next, we investigate the finite sample performance of the regression estimates under a  generalized random dot product graph (GRDPG) model.
In this setup, each node is assigned a latent position according to a stochastic block model having three communities with proportions \(\pi = (0.65, 0.25, 0.10)\). The community connectivity matrix $B$ is given as follows:
\[
B = 
\begin{bmatrix}
0.80 & 0.20 & 0.10 \\
0.20 & 0.70 & 0.15 \\
0.10 & 0.15 & 0.90
\end{bmatrix}.
\]

The latent position for this SBM is constructed as \( Z = \rho_n^{-1/2}\theta V D^{1/2} \), where \( \theta \in \mathbb{R}^{n \times d} \) is the block assignment matrix whose rows \( \theta_i \) are independently drawn from a multinomial distribution over $K$ blocks, i.e., \( \theta_i \sim \mathrm{Multinomial}(1; \pi_1, \dots, \pi_K) \). The matrices \( V \) and \( D \) are obtained via the spectral decomposition of the block connectivity matrix \( B \), with \( B = V D V^{-1} \). In addition to the network structure, each node is associated with two conventional covariates \(X_1\) and \(X_2\), independently sampled from a standard bivariate normal distribution with correlation coefficient $\rho = 0.3$. The observed outcome for each node is generated from a linear model having both conventional and network covariates as follows,
\[
Y_i = \beta_1 X_{1i} + \beta_2 X_{2i} + \beta_3 Z_{i_1} + \beta_4 Z_{i_2} + \beta_5 Z_{i_3} + \varepsilon_i,
\]
where \(Z_{i_1}\), \(Z_{i_2}\), and \(Z_{i_3}\) represent the three components from the scaled latent spectral embedding, and \(\varepsilon_i\)'s are independent standard normal errors. The regression coefficients \(\beta_1\) and \(\beta_2\) correspond to the conventional covariates, while the remaining three coefficients capture the network effect through the spectral embedding components. The intercept term is excluded to avoid issues related to identifiability 
To evaluate the inferential validity of the regression coefficients, we consider two different levels of network signal strengths. For the strong signal case, the regression coefficients are set as \((\beta_1,\beta_2, \beta_3, \beta_4, \beta_5) = (1, 2, 1, 2, 1)\), with substantial contributions from the network part in the last three coordinates, whereas in the weak signal case, the last three coefficients are set as \((\beta_3, \beta_4, \beta_5)=(0, 0.01, 0)\), corresponding to a negligible network effect.

To assess the robustness and sensitivity of our inferential procedure for detecting the presence of network effects, we perform the independent multiplier bootstrap procedure as mentioned in Section \ref{bootstrap:rdpg} with the choice of multipliers same as in Section \ref{simulation: sec1}. We first compute the empirical coverage of bootstrap confidence intervals for the individual conventional covariates \(\beta_1\) and \(\beta_2\), which are not directly affected by network variability. Second, we conduct a simultaneous bootstrap-based hypothesis testing defined by \eqref{rdpg:bootstrapped testing statistic} and \eqref{rdpg:observed testing statistic} to detect the presence of any network effect, that is we test for $H_{0}$ : \(\beta_3 = \beta_4 = \beta_5 = 0\). 

We conduct 200 Monte-Carlo iterations, and the results are summarized in Table~\ref{tab:rdpg_sim_results}. Under the strong network effect setting, the bootstrap confidence intervals perform well in maintaining the nominal coverage for \(\beta_1\) and \(\beta_2\) along with detecting the network effect with a high empirical power of 1. Under the weak network effect setting, although the empirical coverage still remains satisfactory, the test exhibits low empirical power performance. The latter confirms that our method avoids over-rejection by failing to reject when the true network signal is almost negligible. These findings support our theoretical claims that the proposed bootstrap procedure has robust performance across different levels of network signal strength.
\begin{table}[h!]
\centering
\captionsetup{font=scriptsize}
\caption{\textit{Bootstrap Coverage and Power under Strong and Weak Network Signal Settings}}
\label{tab:rdpg_sim_results}
\scriptsize
\renewcommand{\arraystretch}{1.2}
\begin{tabular}{lccc}
\toprule
Setting & Coverage (\(\beta_1\)) & Coverage (\(\beta_2\)) & Power (Network Effect) \\
\midrule
Strong Network Effect & 0.970 & 0.885 & 1.000 \\
Weak Network Effect & 0.965 & 0.925 & 0.065 \\
\bottomrule
\end{tabular}
\end{table}

\subsection{Inference for Linear Projection Parameter under Model Misspecification}\label{simulation: sec3}
In this section, we study the robustness of the regression estimators under non-linear models. In these cases, our approach approximates the population projection parameter, which is not directly observable from the data-generating process. To illustrate this feature of the assumption-lean framework, we consider the following nonlinear data-generating process.

For each observation \( i \), the outcome \( Y_i \) is generated according to the following model: 
\[
Y_i \;=\; \log\bigl(1 + 5 Z_i |X_{i_1}|\bigr) 
\;+\; (5Z_i)^{1/2}\,\sin\bigl(0.5\,X_{i_2}\bigr) 
\;+\; \varepsilon_i,
\]
and we again consider the sparse graphon model and rooted two-star for regression, the same as the choice for Section \ref{simulation: sec1}. Specifically, we use $Z_i= E((w(\xi_i, \xi_j) w(\xi_i, \xi_k)\mid \xi_i)$ for the above data generating process and $\widehat{Z}_i$ as the local rooted two-star frequency estimator. \( (X_{i_1}, X_{i_2}) \) are conventional covariates drawn from a standard bivariate Gaussian distribution with correlation coefficient $\rho = 0.3$ and  \( \varepsilon_i \) is a standard normal random variable drawn independently from everything else.  This construction introduces a non-separable interaction between the conventional and network covariates through both logarithmic and trigonometric transformations. As a result, the overall relationship between the predictors and the outcome cannot be easily linearized through standard variable transformations. Instead, it reflects the intrinsic complexity of real-world data-generating mechanisms, where underlying effects often interact in tangled and nonlinear ways beyond the scope of classical parametric models. 
\begin{figure}[htbp]
    \centering
    \includegraphics[width=0.57\textwidth]{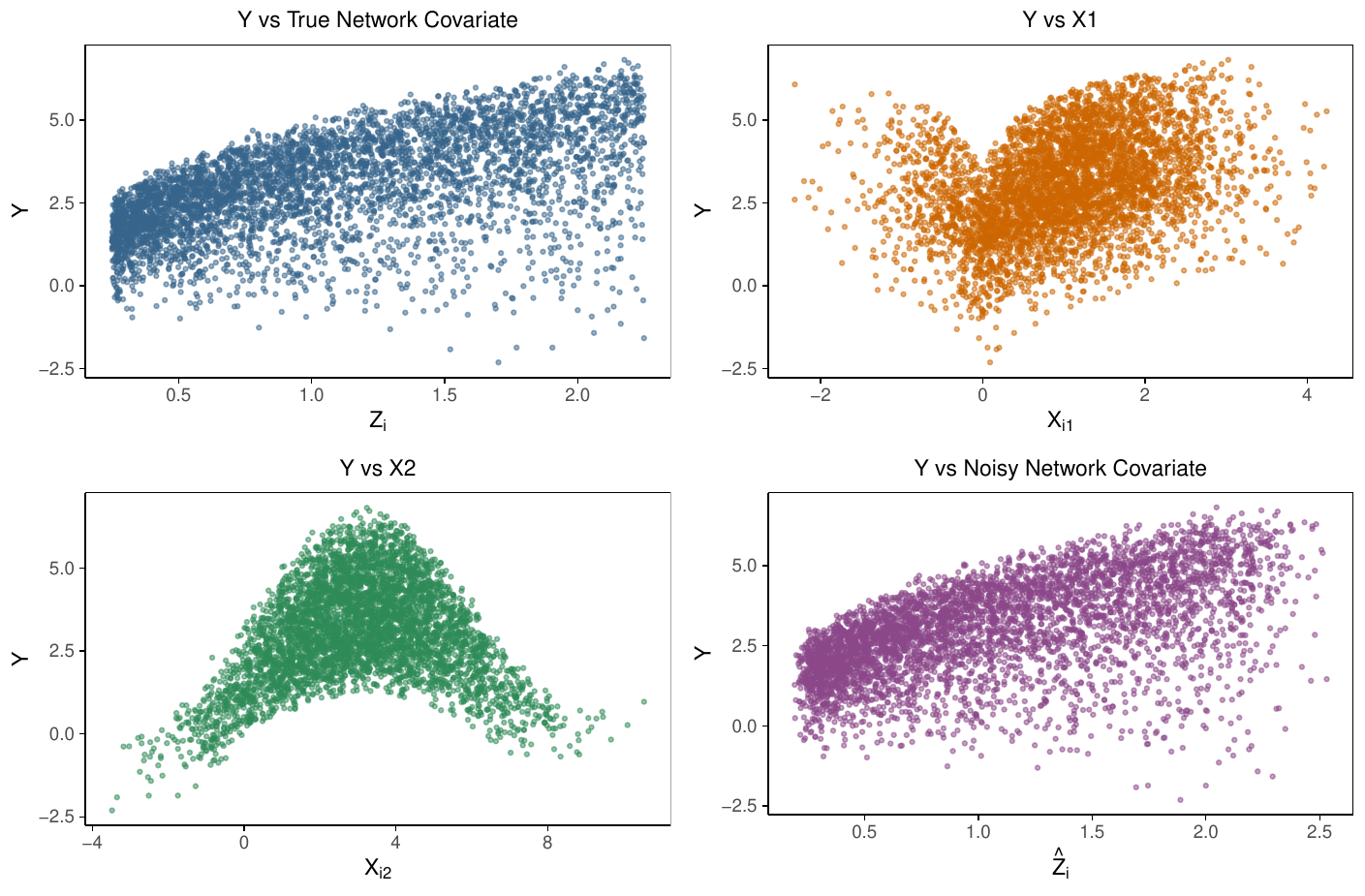}
    \captionsetup{font=scriptsize}
    \caption{\textit{Scatter plots between the response \( Y \) and each covariate. 
    For the network covariate, we plot both the noisy observed version \(\widehat{Z}\) and the latent version \(Z\). 
    The response \( Y \) is generated from the deterministic function 
    \(\mu(X,Z) = \log\bigl(1 + 5 Z_i |X_{i_1}|\bigr) 
    \;+\; (5Z_i)^{1/2}\,\sin\bigl(0.5\,X_{i_2}\bigr) \).}} \label{fig:associateion scatter}
\end{figure}
In practice, we fit a regression model targeting the best linear projection of \(Y_i\) onto \(Z_i\) or \(\widehat{Z}_i\), along with conventional covariates. Figure~\ref{fig:associateion scatter} shows one realization of the sampled network data. To roughly illustrate the trend, we plot the marginal association between the noiseless response and each of the covariates. To emphasize the meaningfulness of projecting the response onto the noisy network covariate, we also include a scatter plot of this association under the weakly sparse case, i.e., \(\rho_n=n^{-0.25}\), as the last one in Figure \ref{fig:associateion scatter}. 

To approximate population projection parameters $\beta^*$ and $\widetilde{\beta}$, we conduct a large-scale Monte Carlo simulation with \( N = 10^5 \) independent samples.  Note that $\widetilde{\beta}$ depends on both the sparsity level and $n$, while $\beta^*$ does not. The resulting coefficients are provided in the first block of the following Table \ref{tab:combined_projection_parameter_matrix_nonlinear}.
\begin{table}[htbp]
\centering
\scriptsize
\captionsetup{font=scriptsize}
\caption{Comparisons of Monte Carlo Approximations of Projection Parameters}
\label{tab:combined_projection_parameter_matrix_nonlinear}
\renewcommand{\arraystretch}{1.2}
\begin{tabular}{l|c|ccc}
\toprule
& $\beta^*$ & \multicolumn{3}{c}{$\widetilde{\beta}$} \\
\cmidrule(lr){2-2} \cmidrule(lr){3-5}
\text{Coefficient} & --- & $\delta = -0.70$ & $\delta = -0.40$ & $\delta = -0.25$ \\
\midrule
Intercept         & 0.8512 & 1.6785 & 0.9737 & 0.8368 \\
$\widehat{Z}$  & 1.4341 & 0.6483 & 1.3034 & 1.4060 \\
$X_1$             & 0.5113 & 0.5052 & 0.5115 & 0.5147 \\
$X_2$             & 0.0420 & 0.0511 & 0.0449 & 0.0548 \\
\bottomrule
\end{tabular}
\end{table}

We see that both projection parameters consistently capture the dominant trend between \(Y\) and \(Z\) despite the model misspecification, highlighting the robustness of the assumption-lean framework. 

For moderate to weakly sparse regimes, where $\rho_n = \omega(n^{-0.5})$, the coefficients for $\beta^*$ and $\widetilde{\beta}$ appear to be close. In contrast, for $\rho_n = o(n^{-0.5})$, these quantities differ, which is even more magnified when $n^{1/2}$ scaling is applied. In both cases, coefficients for conventional covariates $X_1$ and $X_2$ are quite close across the sparsity regimes, which suggests that coefficients for conventional covariates typically do not depend strongly on the choice of target for network covariates.  In sparse regimes, it appears that the value of the intercept term and the coefficients for the network covariates are most affected.

Next, we assess the coverage properties of both the bias-corrected and the non-corrected OLS estimators via bootstrap confidence intervals. We focus on three sparsity regimes as in Section \ref{simulation: sec1}. For the bias-corrected estimator, we evaluate 95\% bootstrap coverage against $\beta^*$ obtained from $N = 10^5$ Monte Carlo simulation with the latent covariate, as seen in the first block of Table \ref{tab:combined_projection_parameter_matrix_nonlinear}. For the OLS estimator, we consider two separate benchmarks. First, we evaluate its coverage against $\beta^*$ to illustrate its limitations for targeting this estimand. Second, we assess its coverage relative to the data-adaptive projection targets $\widetilde{\beta}$.

Table~\ref{tab:coverage_combined_nonlinear} presents coverage for both estimators when targeting $\beta^*$. The bias-corrected estimator achieves near-nominal coverage across all coefficients and sparsity regimes, validating its consistency and robustness. In contrast, the non-corrected estimator shows substantial under-coverage under sparser regimes, especially for the intercept and network covariate, highlighting its failure to target $\beta^*$ without proper bias correction.
\begin{table}[htbp]
\centering
\scriptsize
\captionsetup{font=scriptsize}
\caption{Bootstrap Coverage of 95\% Confidence Intervals for Corrected and Non-corrected Estimators across Three Sparsity Levels}
\label{tab:coverage_combined_nonlinear}
\renewcommand{\arraystretch}{1.15}
\begin{tabular}{c|cccc|cccc|cccc}
\toprule
\multirow{2}{*}{$\delta$} 
& \multicolumn{4}{c|}{Corrected (Targeting $\beta^*$)} 
& \multicolumn{4}{c|}{Non-corrected (Targeting $\beta^*$)} 
& \multicolumn{4}{c}{Non-corrected (Targeting $\widetilde{\beta}(\delta)$)} \\
& Intercept & $Z$ & $X_1$ & $X_2$ 
& Intercept & $Z$ & $X_1$ & $X_2$
& Intercept & $Z$ & $X_1$ & $X_2$ \\
\midrule
$-0.70$ & 0.995 & 0.995 & 0.980 & 0.945 & 0.000 & 0.000 & 0.980 & 0.935 & 0.965 & 1.000 & 0.970 & 0.940 \\
$-0.40$ & 0.940 & 0.945 & 0.955 & 0.930 & 0.615 & 0.190 & 0.965 & 0.940 & 0.940 & 0.955 & 0.965 & 0.940 \\
$-0.25$ & 0.925 & 0.945 & 0.960 & 0.940 & 0.915 & 0.805 & 0.960 & 0.945 & 0.900 & 0.940 & 0.960 & 0.940 \\
\bottomrule
\end{tabular}
\end{table}

To complement this comparison, the last block of Table~\ref{tab:coverage_combined_nonlinear} reports the empirical coverage of the non-corrected estimator when targeting the data-adaptive projection parameters $\widetilde{\beta}(\delta)$. Across all regimes, the non-corrected OLS achieves valid coverage for this alternative estimand, indicating that even without bias correction, it may still yield reliable inference for a well-defined but different projection target. This distinction underscores the flexibility of the assumption-lean framework: while correction is often necessary for recovering $\beta^*$, valid inference can still be obtained for alternative targets that remain meaningful.

\section{Real Data Case Study}
We study a real-world dataset from a school climate intervention study conducted in 28 schools in New Jersey. The original study was introduced in \cite{paluck2016changing}, and \cite{le2022linear} also consider this dataset for their method. In the experiment, a group of seed students was stratified by demographic characteristics and randomly selected to participate in educational workshops focused on school conflict and anti-conflict norm promotion. A binary indicator was used to record whether each student was selected as a seed and this serves as the treatment variable. Surveys were conducted at both the beginning and the end of the academic year. These surveys asked the students to nominate friends, rate the prevalence of various conflict or anti-conflict behaviors, and provide demographic information. Using this dataset, we investigate peer effects and the diffusion of treatment through social networks.

We focus on a series of survey questions designed to assess the school atmosphere. The set consists of 13 questions that measure students' perceptions of friendliness and positive behavior within their school environment. Students were asked how frequently the described behaviors occurred and their responses were recorded on a Likert Scale ranging between 0--5, where `0' denotes “Almost nobody” and `5' “Almost everyone”. Some items captured anti-conflict behaviors, while others described negative and unfriendly conduct. Following the methodology in \cite{le2022linear}, we reverse-coded the unfriendly items and computed an average score across all 13 questions. This aggregate serves as an indirect measure of perceived friendliness.

Following \cite{le2022linear}, we include several demographic covariates in addition to the treatment indicator: gender, race (Black, White, Hispanic, Asian, or Other), grade level, and whether the student lives with a parent or guardian. 

Since the diffusion of anti-conflict norms may be gradual or implicit, we define the social network by taking the elementwise maximum of the friend nomination matrices from the two survey waves as
$
A = \max(A_1, A_2),
$
where \(A_1\) and \(A_2\) represent the adjacency matrices from wave 1 and wave 2, respectively. This construction captures the union of social connections and aims to preserve as many meaningful links as possible for further analysis.  Note that the element-wise maximum of two jointly exchangeable arrays is also jointly exchangeable.  

Finally, after removing cases with missing data and retaining only eligible students, those with at least one response score and non-missing covariates, we include a total of 8,851 students in the analysis.

In this dataset, each school constitutes its own social network, and we assume that networks from different schools are independent. It is natural to ask what type of graphon or GRDPG would give rise to this type of hierarchical structure.  We now provide an example of such a process that admits a GRDPG representation of the form (\ref{eq-grdpg-eq}).  

For ease of exposition, consider a GRDPG model of dimension $d \times S$, where $S$ is the number of schools.  Suppose that if $(\chi_i, \zeta_i)$ belongs to the $s$th school, then $d \times S -d$ terms in the vector are $0$ almost surely and let $\mathcal{I}_s$ and $\mathcal{J}_s$ denote the indices of $\chi_i$ and $\zeta_i$, respectively, which correspond to elements that are not $0$ almost surely. Furthermore, suppose that $\mathcal{I}_s \cap \mathcal{I}_t = \emptyset$ and $\mathcal{J}_s \cap \mathcal{J}_t = \emptyset$  for $s \neq t$.  Let $(\chi_i^{(s)}, \zeta_i^{(s)}) = ((\chi_i)_{i \in \mathcal{I}_s}, (\zeta_j)_{j \in \mathcal{J}_s})$.  Then, the corresponding GRDPG takes the form:
\begin{align*}
P(A_{ij} =1 \ | \ \chi_i, \zeta_i, \chi_j,\zeta_j) &= \langle \chi_i,\chi_j \rangle - \langle \zeta_i,\zeta_j \rangle
\\ &= \begin{cases}
\langle \chi_i^{(s)}, \chi_j^{(s)} \rangle -\langle \zeta_i^{(s)}, \zeta_j^{(s)} \rangle & i,j \text{ belong to school } s \\ 
 0 & \text{otherwise}.
\end{cases}
\end{align*}

Of course, this example can readily be generalized to allow the ranks of the GRDPG models for each school to vary at the cost of more burdensome notation.  Here, the school membership is itself also random, which we believe is natural since school sizes also typically vary from year to year.  In this model, a coordinate of this GRDPG model may be viewed as a school-specific latent variable governing interaction.  For such a structure, it is natural to use spectral embeddings as network covariates in the model. This choice also facilitates a meaningful comparison between our analysis and the one in \cite{le2022linear}.

Let $\widetilde{Z}_i^{(s)}$ denote a vector consisting of elements $(\gamma_i^{(s)}, \zeta_i^{(s)})$ sorted in descending order of magnitude in the spectral decomposition of $\langle \gamma_i^{(s)}, \gamma_j^{(s)} \rangle -\langle \zeta_i^{(s)}, \zeta_j^{(s)} \rangle$.  Since the sparsity level would presumably vary across schools, it is natural to let $\widetilde{Z}_i^{(s)} = \rho_{n,(s)} Z_{i}^{(s)}$ for a school-specific sparsity parameter $\rho_{n,(s)}$. To estimate latent positions from this model, we compute Adjacency Spectral Embeddings (ASE) within a zero-padded global embedding scheme. For students from the $s$-th school, we compute ASE normalized locally:
\begin{align}\label{realdata:local ASE}
\widehat{Z}^{(s)} = \hat{\rho}_{n,(s)}^{-1/2} \widehat{U}_{(s)} |\widehat{S}_{(s)}|^{1/2},
\end{align}
where $\widehat{U}_{(s)}$ and $\widehat{S}_{(s)}$ are zero-padded spectral eigenvector matrix and eigenvalue diagonal matrix with $0$'s being assigned to all the entries outside the student's school block. Computing the ASE in this manner allows the ranks of the GRDPG models for each school to be directly taken into consideration.

\subsection{Main Results}
In this dataset, there are many covariates, so some form of dimension reduction would aid both in the interpretation and estimation of parameters.  In particular, we consider variable selection to reduce the number of factor variables related to school indicators.  To carry out dimension reduction in a principled manner, we conduct sample splitting. We randomly partition the data into two subsets: 30\% for variable selection and the remaining 70\% for inference. For variable selection, we iteratively remove school indicators that are not significant at $\alpha=0.05$ three times.  From now on, we will denote the training set and validation set as $\mathcal{D}_{train}$ and $\mathcal{D}_{valid}$ respectively.    

We fit a sample OLS regression that includes the locally-normalized ASE estimator constructed with $\widehat{Z}^{(s)}$ in \eqref{realdata:local ASE} as one of the covariates and obtain OLS estimates. We choose $d=2$ to be the dimension of the network embedding for each school. The multiplier bootstrap procedure described in Session \ref{bootstrap:rdpg} is applied with independent standard Gaussian multipliers with 500 bootstrap replications for estimating the standard deviation for all the regression coefficients. 
The OLS estimates along with their bootstrap standard error are reported in Table \ref{tab:bootstrap-nonspectral}. The $p$-values obtained using $\mathcal{D}_{valid}$ show that both ``Grade" and ``Living with Parents" are highly significant variables. This conclusion concurs with the previous findings in \cite{li2019prediction}, for both the direction and the magnitude of the effects.   
\begin{table}[ht]
\centering
\small
\captionsetup{font=small}
\caption{Bootstrap Regression Effect Estimates Excluding Spectral Components}
\label{tab:bootstrap-nonspectral}
\begin{tabular}{|l|c|c|c|}
\hline
Variable & Estimate & Bootstrap SE & p-value \\
\hline
Intercept & 3.8935 & 0.1370 & $< 2 \times 10^{-16}$ \\
Grade & -0.1709 & 0.0150 & $< 2 \times 10^{-16}$ \\
Returning to school or not & -0.0436 & 0.0246 & 0.0762 \\
Treatment & -0.0083 & 0.0324 & 0.7973 \\
Lives with both parents & 0.0545 & 0.0209 & 0.0090 \\
Race-Asian & 0.0221 & 0.0529 & 0.6767 \\
Race-Black & -0.0016 & 0.0455 & 0.9715 \\
Race-Hispanic & 0.0326 & 0.0405 & 0.4209 \\
Race-White & 0.0516 & 0.0356 & 0.1470 \\
Gender & -0.0153 & 0.0175 & 0.3841 \\
\hline
Significant School Indicators & & & \\
\hline
School 35 & 0.6717 & 0.0919 & $2.75 \times 10^{-13}$ \\
Other schools & 0.1998 & 0.0682 & 0.0034 \\
\hline
\end{tabular}
\end{table}

More importantly, a key difference between our findings and those reported in prior studies using this dataset lies in the estimated treatment effect. 
In our analysis, the treatment effect is statistically insignificant, having a high $p$-value of 0.7973. This contrasts with the results in \citet{li2019prediction}, where a significant effect was reported with an estimate of -0.0839 and a $p$-value of 0.0442. One can argue from two aspects. First, it is important to note that the response measure of perceived friendliness involves inherent complexity in interpreting the intervention's effects. Attendance at a workshop may simultaneously increase awareness and recognition of conflict. Consequently, this may result in higher reporting and thus potentially reduce the actual conflict in behavior. In fact, these opposing mechanisms can counterbalance one another, potentially leading to little or no observable change in perceived friendliness. 
Secondly, the larger bootstrap standard deviation observed in our bootstrap may indicate the presence of non-linearity in the relationship between the response variable and the covariates. It is important to emphasize that our method accommodates potential model misspecification, making the resulting standard error estimates more robust and reliable than those obtained under the assumption of a correctly specified linear model.

Furthermore, when performing regression using the locally normalized global spectral embedding, strong network effects were detected through our bootstrap-based hypothesis testing procedure, as described in Section \ref{bootstrap:rdpg}. The evidence of network effects in students' evaluation of perceived friendliness is further supported by the network effect plot presented in Figure ~\ref{fig:unit-level-network-effects}. The figure shows that the estimated network effects across different schools, represented by fluctuations in $\widehat{Z}^{\T}\widehat{\beta}$, fluctuate significantly around the baseline value $0$, indicative of significant network effects, which was also found in \citet{li2019prediction}. 

\begin{figure}[htbp]
    \centering    \includegraphics[width=0.7\textwidth]{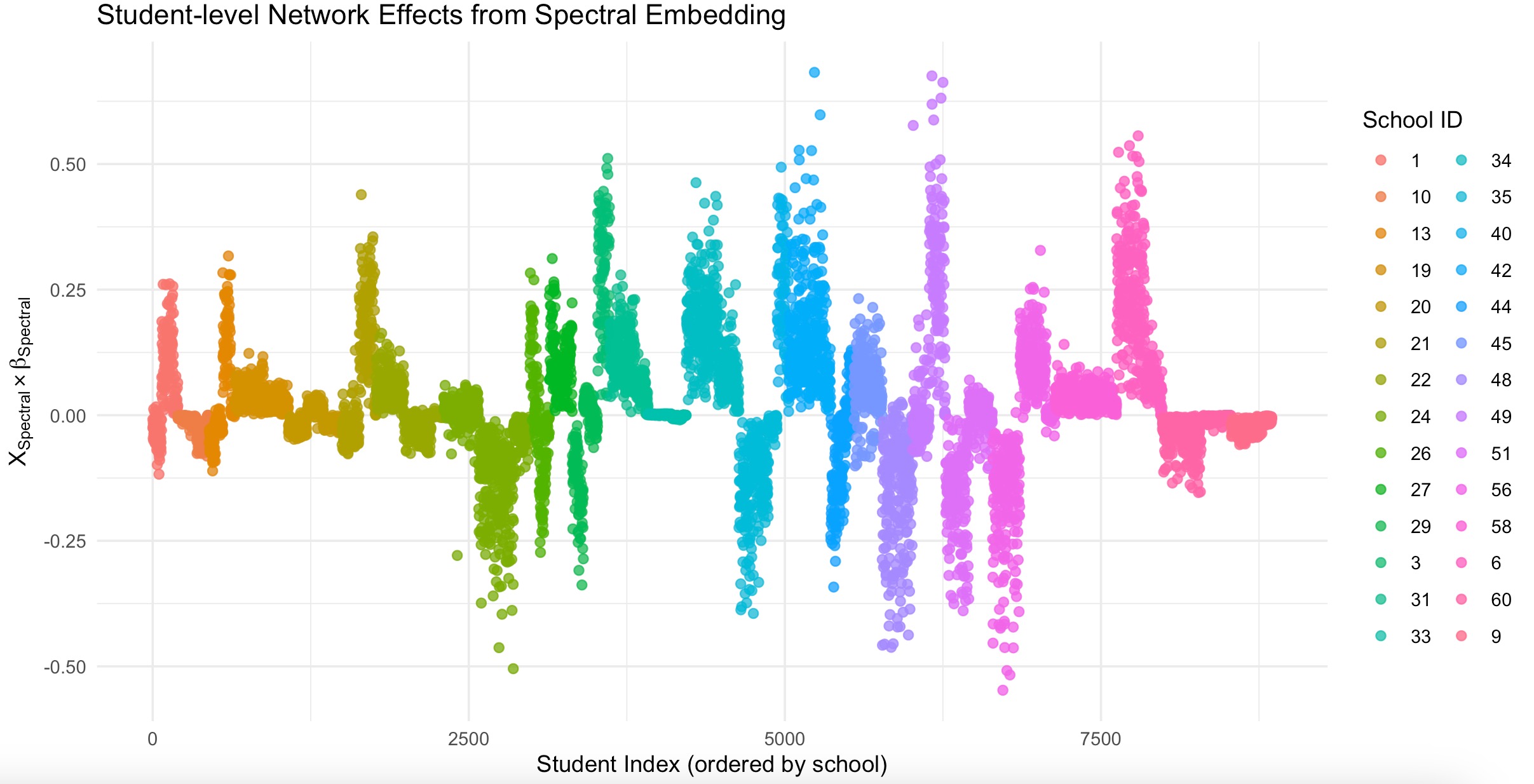}
    \caption{Unit‐level network effects on student outcomes, by school block.}
    \label{fig:unit-level-network-effects}
\end{figure}

It is important to emphasize that our inference for the so-called ``regression effect" is more robust than that of fixed-design regression, where inference is valid only under the assumption of correct model specification. The regression effect we estimate using a simple OLS estimator can be interpreted more broadly, as it does not rely on specific distributional assumptions or the correctness of the model.  As discussed previously, model misspecification is often more likely in settings involving network-linked data due to unobserved latent variables.  

\section{Discussion}
This paper introduces an assumption-lean linear regression framework for data linked through network structures. By reinterpreting regression coefficients as population-level linear projection effects, the approach enables valid inference under minimal assumptions, broadening the applicability of linear models beyond correct specification. We studied various network-derived covariates, including subgraph frequencies and latent position embeddings, and analyzed their asymptotic behavior under relevant models. Moreover, in sparse settings, we show that naive OLS can yield biased or unstable estimates, so we propose a bias-corrected OLS procedure that is consistent across a broad range of settings. When latent positions are unidentifiable, we demonstrate certain natural inference tasks remain well-defined. These results establish new inferential guarantees for regression coefficients under mild conditions.  

To broaden applicability, we introduce a down-sampling strategy that enables valid inference when full-sample methods may be unreliable, offering a practical solution for applied researchers working with sparse networks or complex network-derived covariates.  

It would be of interest to extend this framework to generalized linear models, which require iterative estimation. It is also worthwhile to consider inference under alternative network dependence assumptions. While the Aldous–Hoover representation suggests our model is natural and quite general, it relies on infinite exchangeability. Finitely jointly exchangeable processes need not admit such a representation, so studying such cases further would be valuable.

\section*{Acknowledgements}
Robert Lunde acknowledges support from NSF grant DMS-2515923. The authors would also like to thank Tianxi Li for sharing the school conflict data used in our data analysis, and Elizaveta Levina and Peter Bickel for raising questions related to representation theorems for network-assisted regression problems.

\bibliographystyle{apalike}
\bibliography{references}

\clearpage
\setcounter{section}{0}
\renewcommand{\thesection}{S\arabic{section}}
\renewcommand{\thesubsection}{S\arabic{section}.\arabic{subsection}}
\begingroup
\setstretch{1.15} 
\begin{center}
  {\LARGE\bfseries Supplementary Material\par}
  \vspace{1.2\baselineskip}
  {\normalsize\bfseries Abstract\par}
\end{center}
\endgroup

The supplementary material provides detailed proofs of the main theorems, 
together with computational algorithms and illustrative examples under the general results.  

Supplement~\ref{App-A} contains the proof of Theorem~\ref{main:theorem1} regarding the Aldous–Hoover representation, 
beginning with a reduction lemma and several auxiliary results.  
Supplement~\ref{App-B} provides proofs of the central limit theorems for OLS estimators with local subgraph 
frequencies as covariates (Section~\ref{sec-B.3}). Supporting results for these proofs, including 
structural representations and noise characterizations (Section~\ref{subgraph:lemmas}), are also given, 
along with computational details (Section~\ref{additional_details}). Proofs of bootstrap validity and 
pre-computation results appear in Sections~\ref{appB-sec:bootstrap}--\ref{bootstrap-precom}.  
Supplement~\ref{App-C} provides proofs involving spectral embedding under the GRDPG model. This section provides proofs of a CLT result, bootstrap consistency, and validity of a bootstrap-based test for network effects.  
Supplement~\ref{App-D} provides proofs for down-sampled OLS estimators, establishing root-$m$ consistency 
(Section~\ref{subsec-appedix-down-sampling}) and treating challenging cases under graphon and GRDPG 
models (Section~\ref{ds-example1}-\ref{appC-section:down-sampling_grdpg}).

\section{Proofs for Section \ref{sec:prob-setup}}\label{App-A}

The lemma below will play a key role in the proof of our variant of the Aldous-Hoover Representation.  Recall that a standard Borel space is a measurable space $(S,\mathcal{S})$ for which there exists a metric $d$ on $S$ that makes $\mathcal{S}$ the Borel $\sigma$-field (generated by open sets) and $S$ a complete, separable metric space (see, for example, \citet{preston-borel}).  To prove Theorem \ref{main:theorem1}, results for Borel $\sigma$-fields on $\mathbb{R}^p$ would suffice, but the generality of standard Borel spaces makes certain arguments more elegant. 

\begin{lemma}[Reduction of Functional Dependence]\label{appA-lem:functional_reduction}
Suppose that $U, Z_1, Z_2$ are mutually independent random variables taking values in standard Borel spaces. Further, suppose that $Z_1 = Z_2 \text{ in distribution}$. If $f(U, Z_1) = f(U,Z_2) \text{ almost surely}$ for some measurable map $f$ taking values in a standard Borel space, there exists a measurable map $h$ such that $f(U,Z_1) = h(U) \text{ almost surely}$.
\end{lemma}

\begin{proof}
In what follows let $(S, \mathcal{S})$ denote the measurable space associated with the image of $f$. Since $(S, \mathcal{S})$ is a standard Borel space, it admits a regular conditional probability (RCP) (see, for example, Theorem 8.36 of \citet{Klenke2007-KLEPTA}), and we have $ \operatorname{pr}( B \ | \ U=u) =  \operatorname{pr}( (f(u(\omega), Z_1), f(u(\omega), Z_2)) \in B)$ almost surely for any  $B \in \mathcal{S}$ by independence.  Moreover, since
\begin{align*}
1= \operatorname{pr}(f(U,Z_1) = f(U, Z_2)) = E[\operatorname{pr}(f(u(\omega),Z_1) = f(u(\omega),Z_2))],
\end{align*}
it follows that $f(u,Z_1) = f(u,Z_2)$ almost surely for $u \in A$, where $A$ is some set such that $\operatorname{pr}(U \in A) = 1$.  Since $f(u,Z_1)$ and $f(u,Z_2)$ are independent and are equal almost surely for $u \in A$, it follows that, for any $B \in \mathcal{S}$ and any $u \in A$, $\operatorname{pr}(f(u,Z_1)  \in B) = \operatorname{pr}(\{f(u,Z_1)  \in B\} \cap \{ f(u,Z_2) \in B\}) = \operatorname{pr}(f(u,Z_1)  \in B)^2$.  This in turn implies $\operatorname{pr}(f(u,Z_1) \in B) \in \{0,1\}$ for any $B \in \mathcal{S}$; consequently, $f(u,Z_1) = c(u)$ for some constant $c(u)$ with probability $1$ for any $u \in A$ by Lemma \ref{lemma-0-1} below.  

Since $\{U \in A\}$ is a measure $1$ set, the RCP corresponding to $f(u,Z_1) = c(u)$ almost surely for all $u \in A$ and $f(u,Z_1) = c(u) = 0$ for $u \not\in A$  is also a version of the RCP. Therefore, for any $B \in \mathcal{S}$, $\operatorname{pr}(f(U, Z_1) \in B) = E[\operatorname{pr}( c(u(\omega)) \in B)]$. By definition of RCP, $f_u: \omega \mapsto \operatorname{pr}( c(u(\omega)) \in B)$ is measurable for any $B \in \mathcal{S}$. We have that $f_u^{-1}[\{1\}] = \{ \omega \in \Omega \ | \ c(u(\omega)) \in B \}$; thus, measurability of $f_u$ for any $B \in \mathcal{S}$ implies that $c(U)$ is also measurable. Now, since $\operatorname{pr}(f(U, Z_1) \in B) = \operatorname{pr}(c(U) \in B)$ for any $B \in \mathcal{S}$, the claim follows.    
\end{proof}
We now prove an intermediate result that is needed in the proof of a lemma used in the above proof. In what follows, recall that an atom of a $\sigma$-finite measure $\mu$ is a set $A$ that satisfies $0 < \mu(A) < \infty$ and for any measurable subset $E \subseteq A$, $\mu(E) = 0$ or $\mu(E) = \mu(A)$.  
\begin{lemma}
\label{lemma-atom-singleton}
Suppose that $(S,\mathcal{S})$ is a standard Borel space equipped with a $\sigma$-finite measure $\mu$.  Any atom of $\mu$ must take the form $A = \{x\} \cup N$, where $x \in S$ and $\mu(N) = 0$.  
\end{lemma}
\begin{proof}
Since $(S, \mathcal{S})$ is standard Borel, there exists a metric $d$ such that $\mathcal{S}$ corresponds to the Borel $\sigma$-algebra on a separable metric space.  Let $A$ be an atom and let $D$ be a countable dense subset of $S$ with respect to this metric $d$.  For any $\epsilon >0$ and $x \in D$, define the open ball:
\begin{align*}
\mathcal{B}_\epsilon(x) = \{ s \in S  \ | \  d(x,s) < \epsilon \}.
\end{align*}
For an arbitrary ordering of $D$, let $x_i$ denote the $i$th element of $D$.  Let $\widetilde{\mathcal{B}}_\epsilon(x_1) = \mathcal{B}_{\epsilon}(x_1)$ and for any $k \geq 2$, $\widetilde{\mathcal{B}}_\epsilon(x_k) = \mathcal{B}_{\epsilon}(x_k) \setminus (\bigcup_{i=1}^{k-1} \widetilde{\mathcal{B}}_{\epsilon}(x_i))$.   By countable additivity, for any $\epsilon > 0$, 
\begin{align*}
\mu(A) =  \sum_{i=1}^\infty\mu(A  \ \cap \ \widetilde{\mathcal{B}}_{\epsilon}(x_i)).
\end{align*}
Therefore, exactly one $x \in D$ must satisfy $\mu(A  \ \cap \ \widetilde{\mathcal{B}}_{\epsilon}(x))  = \mu(A) > 0$.  Now for each $1/k$, $k \in \mathbb{N}$, choose $x^{(k)}$ such that $\mu( A  \cap \widetilde{\mathcal B}_{1/k}(x^{(k)})) > 0$.  Now, observe that:
\begin{align*}
\mu\left(A \setminus \bigcap_{k \in \mathbb{N}} \widetilde{\mathcal{B}}_{1/k}(x^{(k)}) \right) &= \mu\left( A \cap \left\{ \bigcup_{k \in \mathbb{N}} \widetilde{\mathcal{B}}_{1/k}^c(x^{(k)})\right\}
\right)
\\ & \leq \sum_{k \in \mathbb{N}} \mu(A \cap \widetilde{\mathcal{B}}_{1/k}^c(x^{(k)})) = 0.
\end{align*}
Thus, it must be the case that:
\begin{align*}
\mu\left(A \cap \left\{\cap_{k \in \mathbb{N}} \widetilde{\mathcal{B}}_{1/k}(x^{(k)}) \right\}  \right) = \mu(A).
\end{align*}
However, $\cap_{k \in \mathbb{N}} \widetilde{\mathcal{B}}_{1/k}(x^{(k)}) \subseteq \cap_{k \in \mathbb{N}} \mathcal{B}_{1/k}(x^{(k)})$ contains at most one element. Since the above intersection is non-empty, it contains exactly one element; denote this element $x$. It must be the case that $ \mu(A) = \mu(A \cap \{x\})$, and from above, we also have $\mu(A \setminus \{x\}) = 0$; the result follows.  
\end{proof}

The following result is used in the proof of Lemma \ref{appA-lem:functional_reduction}.  
\begin{lemma}
\label{lemma-0-1}
Suppose that $X: (\Omega, \mathcal{F}) \mapsto (S, \mathcal{S})$ is a measurable random variable and $(S, \mathcal{S})$ is a standard Borel space.  If $\operatorname{pr}(X \in  B) \in \{0,1\}$ for all $ B  \in \mathcal{S}$, then there exists some $c \in S$ such that $\operatorname{pr}(X = c) = 1$.     
\end{lemma}
\begin{proof}
We consider two cases, corresponding to the probability measure for $X$ being non-atomic or containing at least one atom.  We start with the former.  Suppose $\operatorname{pr}(X \in \cdot)$ is non-atomic (diffuse).  Let $B \in \mathcal{S}$ be an event such that $\operatorname{pr}(X \in B) = 1$; at least one such event exists ($S$, for example).  Since $\operatorname{pr}(X \in \cdot)$ is non-atomic, there must exist some event $A$ such that $A \subseteq B$ and $0 < \operatorname{pr}(A) < \operatorname{pr}(B) \leq 1$.  However, this contradicts the assumption in the lemma; therefore, $\operatorname{pr}(X \in \cdot)$ must contain at least one atom.  

Now, by Lemma \ref{lemma-atom-singleton}, an atom of such a distribution must be a singleton set, up to null sets.  For such $c$, we have $\operatorname{pr}(X =c) >0$, and by assumption, we must also have $\operatorname{pr}(X =c)=1$. The result follows.   
\end{proof}

We next provide the formal proof of Theorem \ref{main:theorem1}.
\begin{proof}
Let $S_V = \mathbb{R} \times \mathbb{R}^p \times \mathbb{R} \times \mathbb{R}^p \times \mathbb{R}$ be the space where $V_{ij}$ takes its values, and $\mathcal{S}_V$ be the corresponding Borel $\sigma$-field. Since $(V_{ij})_{i \neq j}$ is a jointly exchangeable array, by the Aldous-Hoover theorem (see for example, Theorem 28.1 of \citet{Kallenberg2002}, 
there exist mutually independent  collections of i.i.d. $\mathrm{Uniform}[0,1]$ random variables $\alpha$, $(\xi_i)_{i \in \mathbb{N}}$, and $(\eta_{ij})_{i<j}$ and a measurable function $f: [0,1]^4 \mapsto S_V$ such that the joint distribution of $(V_{ij})_{i \neq j}$ can be represented as:
\begin{align}
\label{eq-aldous-hoover}
(V_{ij})_{i \neq j} = \left(f(\alpha, \xi_i, \xi_j, \eta_{ij}) \right)_{i \neq j} \quad \text{in distribution}.   
\end{align}
The function $f$ can be decomposed into a tuple of measurable component functions corresponding to each component in the array as follows:
\begin{align*}
f(\alpha, u, v, w) = \left(f_{YX}( \alpha, u, v, w), f_{Y'X'}( \alpha, u, v, w),g( \alpha, u, v, w)\right). 
\end{align*}

Consider two  indices $k< l_1 < l_2$. Then,
\begin{align*}
(Y_k,X_k,Y_k,X_k) &= (f_{YX}(\alpha, \xi_k, \xi_{l_1}, \eta_{kl_1}) , f_{YX}(\alpha, \xi_k, \xi_{l_2}, \eta_{kl_2})) \quad \text{in distribution,}
\end{align*}
which in turn implies:
\begin{align*}
f_{YX}(\alpha, \xi_k, \xi_{l_1}, \eta_{kl_1}) = f_{YX}(\alpha, \xi_k, \xi_{l_2}, \eta_{kl_2}) \quad\text{almost surely.} 
\end{align*}
By Lemma \ref{appA-lem:functional_reduction}, it follows that there exists some measurable map $h$ such that:
\begin{align}
\label{eq-reduced-form-ah}
f_{YX}(\alpha, \xi_k, \xi_{l_1}, \eta_{k,l_1}) = h(\alpha, \xi_k)  \quad\text{almost surely.} 
\end{align}

Moreover, by examining the joint distribution of $(Y_{l_1}, Y_{l_1})$ induced by $(V_{k l_1}, V_{l_1 k})$, it is clear that:
\begin{align}
\label{eq-yx-xy-equivalence}
f_{Y'X'}( \alpha, \xi_k, \xi_{l_1}, \eta_{k l_1}) = f_{YX}( \alpha, \xi_{l_1}, \xi_{k}, \eta_{k l_1}) \quad\text{almost surely}.
\end{align}

Now, since \eqref{eq-reduced-form-ah} and \eqref{eq-yx-xy-equivalence} hold almost surely, they also hold almost surely conditional on $\xi_k, \xi_{l_1}, \alpha$ with probability $1$.  Therefore, on a probability $1$ set, we have:
\begin{align*}
& \mathcal{L}\left(f(\alpha, \xi_k, \xi_{l_1}, \eta_{kl_1})  \ \bigr\rvert \ \alpha = u_1, \xi_k = u_2, \xi_{l_1} = u_3 \right) \\  = \ & \mathcal{L}\left((h(\alpha, \xi_k), f_{YX}(\alpha ,\xi_{l_1},\xi_k, \eta_{kl_1}),g(\alpha,\xi_{k},\xi_{l_1}, \eta_{kl_1})) \ \bigr\rvert \ \alpha =u_1, \xi_k=u_2, \xi_{l_1}=u_3 \right) \\
= \ & \mathcal{L}((h(\alpha, \xi_{l_1}), f_{YX}(\alpha,\xi_{k},\xi_{l_1}, \eta_{kl_1}),g(\alpha,\xi_{l_1},\xi_{k}, \eta_{kl_1})) \ \bigr\rvert \ \alpha = u_1, \xi_{l_1} = u_2, \xi_{k} = u_3) \\
= \ &  \mathcal{L}\left((h(\alpha, \xi_{l_1}), h(\alpha, \xi_{k}),g(\alpha,\xi_{l_1},\xi_{k}, \eta_{kl_1})) \ \bigr\rvert \ \alpha=u_1, \xi_{l_1}=u_2, \xi_{k} = u_3 \right) \\ 
= \ & \mathcal{L}\left((h(\alpha,\xi_k), \, h(\alpha,\xi_{l_1}), g(\alpha,\xi_k,\xi_{l_1}, \eta_{kl_1})) \ \bigr\rvert \ \alpha=u_1, \xi_{k}=u_2, \xi_{l_1}=u_3 \right),
\end{align*}
where above $\mathcal{L}(\cdot)$ denotes distribution. 

Similarly, by permuting the latent variables, we have with probability $1$, marginally for any $i \neq j$, 
\begin{align}
\label{eq-conditional-equivalence}
\mathcal{L}\left(f(\alpha, \xi_i, \xi_j, \eta_{ij}) \ \bigr\rvert \ \alpha,\xi_i, \xi_j\right) = \mathcal{L}\left((h(\alpha, \xi_i), h(\alpha, \xi_j), g(\alpha, \xi_i, \xi_j, \eta_{ij}))  \ \bigr\rvert \ \alpha, \xi_i, \xi_j\right).
\end{align}

We will now show that the joint distribution has the claimed form.  To this end, it suffices to show that any finite-dimensional distribution is equal in distribution to any finite-dimensional distribution of the posited distribution.  Consider the set of distinct indices $(i_1,j_1), \ldots,(i_r,j_r)$, where $i_k \neq j_k$ for any $k \in [r]$ and $r \in \mathbb{N}$.  Observe that, for arbitrary measurable sets $A_1, \ldots, A_r$, by conditional independence,
\begin{align*}
\operatorname{pr}\left( \bigcap_{k=1}^r \bigl\{f(\alpha, \xi_{i_k}, \xi_{j_k}, \eta_{i_kj_k}) \in A_k \bigr\}  \right) = E\left[\prod_{k=1}^r \operatorname{pr}\left( f(\alpha, \xi_{i_k}, \xi_{j_k}, \eta_{i_kj_k}) \in A_k \ \bigr\rvert \ \alpha,\xi_{i_k},\xi_{j_k} \right)\right]. 
\end{align*}

Since our claimed representation also has the same conditional independence structure, to show equality in distribution, it suffices to show that each of the conditional distributions agree.  This is guaranteed by \eqref{eq-conditional-equivalence}.  Therefore, unconditionally we have
\begin{align*}
\mathcal{L}\left((V_{i_k j_k})_{1 \leq k \leq r} \right)=\mathcal{L}\left(\left(h(\alpha,\xi_{i_k}), h(\alpha, \xi_{j_k}), g(\alpha, \xi_{i_k}, \xi_{j_k}, \eta_{i_kj_k}) \right)_{1 \leq k \le r}\right).
\end{align*}
The result follows. 
\end{proof}

\section{Proofs for Section \ref{subsec-local-subgraph}}\label{App-B}
In this section, we provide proofs related to the properties of linear regression when local subgraph counts are used as covariates.  In our proofs, the following count frequency will often be of interest:
\begin{align}\label{eq-local-subgraph-proxy}
\widecheck{Q}_i(\mathcal{R}_j)=
\frac{1}{{n-1 \choose r-1} \rho_n^s} \sum_{\substack{1 \leq i_1 <  \cdots < i_{r-1} \leq n, \\ i_1, \ldots, i_{r-1} \neq i}} \frac{1}{|\mathrm{Iso}(\mathcal{R}_j)|}
\sum_{\substack{\mathcal{S}: \mathcal{S} \cong \mathcal{R}_j,\\
\mathcal{V}(\mathcal{S}) = \{i,i_1, \ldots, i_{r-1}\}}} \prod_{\{j,k\} \in \mathcal{E}(S) } A_{jk}.
\end{align}
Compared with the definition of $\widehat{Q}_i(\mathcal{R}_j)$, the only difference is that normalization is by the true sparsity level rather than the estimated one.  It should also be noted that since $\|w\|_\infty \leq C$, $\rho_n w \leq 1$ for $n$ large enough. Therefore, for notational simplicity, in the following derivations, we assume $n$ is large enough and truncation can be ignored. In what follows, let $\tau_i=(X_i,Y_i,\xi_i)$.

In what follows, conditioning on $\tau$ refers to conditioning on the vector $(\tau_1,\ldots,\tau_n)$, 
which is equivalent to conditioning on $(X_1,\ldots,X_n,\,Y_1,\ldots,Y_n,\,\xi_1,\ldots,\xi_n)$. 
When conditioning is written with respect to specific components, such as $\xi$, $X$, $Y$, or any combination thereof, 
it denotes conditioning on corresponding vectors; for example, 
$\xi = (\xi_1,\ldots,\xi_n)$. 
We also employ the conventional Bachmann-Landau notation and its probabilistic counterparts.

\subsection{Lemmas for Local Subgraph Frequencies}\label{subgraph:lemmas}
We begin by deriving a global representation for all local network subgraph statistics that are linear in $\widehat{Z}_i$.  We show that this admits a noisy U-statistic form, consisting of a U-statistic plus a negligible noise term. 

We first analyze the intermediate term \( \widecheck{Z}_{ij} = \widecheck{Q}_i(\mathcal{R}_j) \), assuming the sparsity level $\rho_n$ is known; the effect of estimating $\rho_n$ will be accounted for in a delta method argument in later proofs. The following Lemma \ref{appB:lemma1} gives a global representation for these statistics. 

\begin{lemma}[Noisy U-Statistic Representation for Linearly Weighted Network Statistics]
\label{appB:lemma1}
The following representation holds:
\begin{align*} 
n^{-1}\sum_{i=1}^nY_i\widecheck{Z}_{ij}&=\frac{1}{\binom{n}{r_j}}
\sum_{1 \leq k_1 < \cdots < k_{r_j} \leq n} h_{Y,\mathcal{R}_j}\left(\tau_{k_1},\ldots,\tau_{k_{r_j}}\right) +  \widehat{\delta}_{n,Y,\mathcal{R}_j},\\ 
n^{-1}\sum_{i=1}^n X_i\widecheck{Z}_{ij}&=\frac{1}{\binom{n}{r_j}}
\sum_{1 \leq k_1 < \cdots < k_{r_j} \leq n} h_{X,\mathcal{R}_j}\left(\tau_{k_1},\ldots,\tau_{k_{r_j}}\right) + \widehat{\delta}_{n,X,\mathcal{R}_j}.
\end{align*}
A special case immediately follows:
\begin{align*}
\frac{1}{n}\sum_{i=1}^n \widecheck{Z}_{ij}
  &= \frac{1}{\binom{n}{r_j}}
     \sum_{1 \leq k_1 < \cdots < k_{r_j} \leq n}  h_{\mathcal{R}_j}\!\left(\tau_{k_1},\ldots,\tau_{k_{r_j}}\right)
     + \widehat{\delta}_{n,\mathcal{R}_j}.
\end{align*}
Here $h_{Y,\mathcal{R}_j}(\cdot)$ and $h_{X,\mathcal{R}_j}(\cdot)$ are symmetric kernels of order $r_j$, associating $(Y_{k_1},\dots,Y_{k_{r_j}})$ or $(X_{k_1},\dots,X_{k_{r_j}})$ with the corresponding subgraph components. The remainder term represents the network noise of order $o_P(n^{-1/2})$ under appropriate sparsity conditions verified in Lemma \ref{appB:lemma2}.
\end{lemma}

\begin{proof}
Observe that:
\begin{align}
& n^{-1}\sum_{i=1}^nY_i\widecheck{Z}_{ij} = n^{-1}\sum_{i=1}^nY_i\widecheck{Q}_i(\mathcal{R}_j) \nonumber\\
&= n^{-1} \sum_{i=1}^n Y_i \left[ \frac{1}{\binom{n-1}{r_j-1} \rho_n^{s_j} |\mathrm{Iso}(\mathcal{R}_j)|} \sum_{\substack{1 \leq k_2 < \dots < k_{r_j} \leq n \\ k_2, \ldots, k_{r_j} \neq i}} \sum_{\substack{\mathcal{S}: \mathcal{S} \cong \mathcal{R}_j,\\ \mathcal{V}(\mathcal{S}) = \{i,k_2, \ldots, k_{r_j}\}}} \prod_{\{a,b\} \in \mathcal{E}(\mathcal{S})} A_{ab} \right] \nonumber\\
&= \frac{1}{n \binom{n-1}{r_j-1} \rho_n^{s_j} |\mathrm{Iso}(\mathcal{R}_j)|} \sum_{i=1}^n Y_i \sum_{\substack{1 \leq k_2 < \dots < k_{r_j} \leq n \\ k_2, \ldots, k_{r_j} \neq i}} \sum_{\substack{\mathcal{S}: \mathcal{S} \cong \mathcal{R}_j,\\ \mathcal{V}(\mathcal{S}) = \{i,k_2, \ldots, k_{r_j}\}}} \prod_{\{a,b\} \in \mathcal{E}(\mathcal{S})} A_{ab}. \label{appB-eq:noisy_linear_initial_sum}
\end{align}
The double summation $\sum_{i=1}^n \sum_{k_2 < \dots < k_{r_j},\; k_\ell \neq i}$ iterates over all distinct $r_j$-tuples of nodes $\{i, k_2, \dots, k_{r_j}\}$. For any given $r_j$-tuple $\{p_1, \dots, p_{r_j}\}$ (with $p_1 < \dots < p_{r_j}$), there are exactly $r_j$ ways that one of these nodes can serve as the root node $i$, with the remaining $r_j-1$ nodes forming the set $\{k_2, \dots, k_{r_j}\}$. We define the symmetric kernel  $h_{Y,\mathcal{R}_j}(\tau_{p_1}, \dots, \tau_{p_{r_j}})$ as:
\begin{align*}
h_{Y,\mathcal{R}_j}(\tau_{p_1}, \dots, \tau_{p_{r_j}}) = \left( \frac{1}{r_j} \sum_{\ell=1}^{r_j} Y_{p_\ell} \right) \left( \frac{1}{|\mathrm{Iso}(\mathcal{R}_j)|} \sum_{\substack{\mathcal{S}: \mathcal{S} \cong \mathcal{R}_j,\\ \mathcal{V}(\mathcal{S}) = \{p_1, \ldots, p_{r_j}\}}} \prod_{\{a,b\} \in \mathcal{E}(\mathcal{S})} w(\xi_a,\xi_b) \right).
\end{align*}

With this kernel, we can rewrite the initial sum \eqref{appB-eq:noisy_linear_initial_sum} as:
\begin{align*}
n^{-1}\sum_{i=1}^nY_i\widecheck{Z}_{ij} 
&= \frac{r_j}{n \binom{n-1}{r_j-1}} \sum_{1 \leq p_1 < \dots < p_{r_j} \leq n} h_{Y,\mathcal{R}_j}(\tau_{p_1}, \dots, \tau_{p_{r_j}})+\widehat{\delta}_{n,Y,\mathcal{R}_j} \\
&= \frac{1}{\binom{n}{r_j}} \sum_{1 \leq p_1 < \dots < p_{r_j} \leq n} h_{Y,\mathcal{R}_j}(\tau_{p_1}, \dots, \tau_{p_{r_j}})+\widehat{\delta}_{n,Y,\mathcal{R}_j},
\end{align*}
where the residual
\begin{align*}
\widehat{\delta}_{n,Y,\mathcal{R}_j}
&= \binom{n}{r_j}^{-1}\!
   \sum_{1 \le p_1<\cdots<p_{r_j}\le n}
   \frac{ r_j^{-1}\sum_{\ell=1}^{r_j}Y_{p_\ell} }{|\operatorname{Iso}(\mathcal R_j)|}\!
   \sum_{\substack{\mathcal S\cong\mathcal R_j\\[-0.25ex] \mathcal V(\mathcal S)=\{p_1,\dots,p_{r_j}\}}}
   \Bigg(
     \rho_n^{-s_j}\!
     \prod_{\substack{(a,b)\in \mathcal E(\mathcal{S})}}\! A_{ab}
     \!-\!
     \prod_{\substack{(a,b)\in \mathcal E(\mathcal{S})}}\! w(\xi_a,\xi_b)
   \Bigg),
\end{align*}
is considered in Lemma \ref{appB:lemma2}. This precisely matches the noisy U-statistic representation for $n^{-1}\sum_{i=1}^nY_i\widecheck{Z}_{ij}$ stated in the lemma. An identical derivation yields the representation for $n^{-1}\sum_{i=1}^nX_i\widecheck{Z}_{ij}$, under analogous expressions of $h_{X,\mathcal{R}_j}(\tau_{p_1}, \dots, \tau_{p_{r_j}})$ and $\widehat{\delta}_{n,X,\mathcal{R}_j}$.

For rooted subgraphs, where $\widecheck{Z}_{ij} = \widecheck{Q}_i^{\mathrm{rooted}}(\mathcal{R}_j)$ with $\widehat{Q}_i^{\mathrm{rooted}}(\mathcal{R}_j)$ defined in \eqref{eq-local-subgraph-rooted}, the kernel takes a simplified form:
\begin{align}\label{appB-eq:h_Y_rooted}
h_{Y,\mathcal{R}_j}^{\mathrm{rooted}}(\tau_{p_1}, \dots, \tau_{p_{r_j}}) = \frac{1}{r_j} \sum_{\ell=1}^{r_j} \left( Y_{p_\ell} \prod_{q \in \{p_1, \dots, p_{r_j}\} \setminus \{p_\ell\}} w(\xi_{p_\ell},\xi_q) \right) .
\end{align}
The U-statistic representation then follows the exact same combinatorial manipulation as in the general case. The claim follows.
\end{proof}

We next provide the proof of Proposition \ref{main:prop1}, which establishes a representation of the quadratic subgraph statistics as linear combinations of global noisy $U$-statistics with appropriate scaling factors.

\begin{proof}[Proof of Proposition \ref{main:prop1}.]
Multiplying expressions for subgraphs $\mathcal{R}_j$ and $\mathcal{R}_k$  and summing over \( i = 1, \ldots, n \), we obtain
\begin{align}\label{exp:raw squared expansion}
n^{-1} \sum_{i=1}^n \widehat{Z}_{i,j} \widehat{Z}_{i,k}  
& = \frac{1}{n \binom{n-1}{r_j - 1} \binom{n-1}{r_k - 1} \hat{\rho}_n^{s_j + s_k} 
|\mathrm{Iso}(\mathcal{R}_j)| |\mathrm{Iso}(\mathcal{R}_k)|} \nonumber \\
& \quad \times \sum_{i=1}^n \sum_{\substack{1 \leq i_1 < \cdots < i_{r_j - 1} \leq n \\ i_\ell \neq i}} 
\sum_{\substack{1 \leq j_1 < \cdots < j_{r_k - 1} \leq n \\ j_\ell \neq i}} \\
& \qquad \times \sum_{\substack{\mathcal{S}_j \cong \mathcal{R}_j \\ \mathcal{V}(\mathcal{S}_j) = \{i, i_1, \ldots, i_{r_j - 1}\}}} 
\sum_{\substack{\mathcal{S}_k \cong \mathcal{R}_k \\ \mathcal{V}(\mathcal{S}_k) = \{i, j_1, \ldots, j_{r_k - 1}\}}} \nonumber   
\prod_{\{i_a, i_b\} \in \mathcal{E}(\mathcal{S}_j)} A_{i_a i_b} 
\prod_{\{j_a, j_b\} \in \mathcal{E}(\mathcal{S}_k)} A_{j_a j_b}.
\end{align}

We now rearrange the expression \eqref{exp:raw squared expansion} by merging into overlapping motifs locally; this leads to the following equivalent characterization:
\begin{align}
\label{exp:local squared merging}
n^{-1}\sum_{i=1}^n \widehat{Z}_{i,j}\,\widehat{Z}_{i,k}
&=
\frac{1}
{\,n\,\binom{n-1}{r_j-1}\,\binom{n-1}{r_k-1}\,
\hat\rho_n^{\,s_j+s_k}\,
|\mathrm{Iso}(\mathcal{R}_j)|\,|\mathrm{Iso}(\mathcal{R}_k)|\,}
\nonumber\\
&\quad\times
\sum_{i=1}^n
\quad \sum_{c=1}^{\min(r_j,r_k)}
\sum_{\substack{L\subset [n]\setminus\{i\}\\|L|=c-1}}
\sum_{\substack{I\subset [n]\setminus(\{i\}\cup L)\\|I|=r_j-c}}
\sum_{\substack{J\subset [n]\setminus(\{i\}\cup L\cup I)\\|J|=r_k-c}}
\sum_{d=0}^{\min(s_j, s_k)} \sum_{\substack{\mathcal{M}^*\in\mathfrak{M}^{\mathcal{R}_j,\mathcal{R}_k}_{c,d}}}
\nonumber\\
&\quad\quad\times
\sum_{\substack{S_j\cong\mathcal{R}_j\\ \mathcal{V}(S_j)=\{i\}\cup L\cup I}}
\sum_{\substack{S_k\cong\mathcal{R}_k\\ \mathcal{V}(S_k)=\{i\}\cup L\cup J}}
\mathbbm{1}\{S_j\cup S_k\cong\mathcal{M}^*\}
\prod_{\{a,b\}\in \mathcal{E}(S_j\cup S_k)}
A_{a b}\,. 
\end{align}

Given a particular $c$ and a selection of nodes on the two subgraphs, suppose each overlap motif $\mathcal{M}^*$ has $K_{\mathfrak{M}^{\mathcal{R}_j,\mathcal{R}_k}_{c}}(\mathcal{M}^*)$ isomorphic cases when expanding 
\begin{align*}
\sum_{\substack{\mathcal{S}_j \cong \mathcal{R}_j \\ \mathcal{V}(\mathcal{S}_j) = \{i, i_1, \ldots, i_{r_j - 1}\}}} 
\sum_{\substack{\mathcal{S}_k \cong \mathcal{R}_k \\ \mathcal{V}(\mathcal{S}_k) = \{i, j_1, \ldots, j_{r_k - 1}\}}} \left(
\prod_{\{i_a, i_b\} \in \mathcal{E}(\mathcal{S}_j)} A_{i_a i_b} \times  
\prod_{\{j_a, j_b\} \in \mathcal{E}(\mathcal{S}_k)} A_{j_a j_b}\right),     
\end{align*}
the sum over these multiplicities then satisfies $\sum_{\substack{\mathcal{M}^*\in\mathfrak{M}^{\mathcal{R}_j,\mathcal{R}_k}_{c}}} K_{\mathfrak{M}^{\mathcal{R}_j,\mathcal{R}_k}_{c}}(\mathcal{M}^*)=|\mathrm{Iso}(\mathcal{R}_j)|\times|\mathrm{Iso}(\mathcal{R}_k)|$, possibly spanning different $d$, where
\[
\mathfrak{M}^{\mathcal{R}_j,\mathcal{R}_k}_{c} := \bigcup_{d=0}^{\min(s_j, s_k)} \mathfrak{M}^{\mathcal{R}_j,\mathcal{R}_k}_{c,d}.
\]

Re-expressing the sum \eqref{exp:local squared merging} yields
\begin{align}\label{exp:global_squared_merging}
&n^{-1}\sum_{i=1}^n \widehat{Z}_{i,j}\,\widehat{Z}_{i,k} \nonumber\\
&=
\frac{1}
{%
|\mathrm{Iso}(\mathcal{R}_j)|\,|\mathrm{Iso}(\mathcal{R}_k)|%
}
\sum_{c=1}^{\min(r_j,r_k)}
\sum_{d=0}^{\min(s_j, s_k)}
\frac{\hat\rho_n^{-d}\binom{n}{r_j+r_k-c}}
{%
n\,\binom{n-1}{r_j-1}\,\binom{n-1}{r_k-1}\,%
}
\nonumber\\
& \quad \times
\frac{1}{\binom{n}{r_j+r_k-c}\hat\rho_n^{\,s_j+s_k-d}} \sum_{\substack{L\subset[n]\\|L|=r_j+r_k-c}}
\left\{
\sum_{\substack{i_u\in L}}
\sum_{\substack{I\subset L\setminus\{i_u\}\\|I|=c-1}}
\sum_{\substack{J\subset L\setminus(\{i_u\}\cup I)\\|J|=r_j-c}}
\sum_{\mathcal M^*\in\mathfrak{M}^{\mathcal R_j,\mathcal R_k}_{c,d}}
\right.
\nonumber\\
& \quad \left.
\times \sum_{\substack{S_j\cong\mathcal R_j\\
\mathcal{V}(S_j)= \{i_u\}\cup I \cup J}}
\sum_{\substack{S_k\cong\mathcal R_k\\
\mathcal{V}(S_k)= (L \setminus \mathcal{V}(S_j)) \cup \{i_u\} \cup I}} \mathbbm{1}\{S_j\cup S_k\cong\mathcal M^*\} \prod_{\{a,b\}\in \mathcal{E}(S_j\cup S_k)}A_{ab}
\right\}.
\end{align}

We now separate this expression \eqref{exp:global_squared_merging} into the leading term that corresponds to $c=1, d=0$ and a remainder term as follows:
\begin{align*}
n^{-1}\sum_{i=1}^n \widehat{Z}_{i,j}\,\widehat{Z}_{i,k}
&= \frac{1}{|\mathrm{Iso}(\mathcal{R}_j)|\,|\mathrm{Iso}(\mathcal{R}_k)|}
   \sum_{\mathcal M^*\in
     \mathfrak{M}^{\mathcal R_j,\mathcal R_k}_{1,0}}
   \frac{1}{\binom{n}{r_j+r_k-1}\,\hat\rho_n^{\,s_j+s_k}} \\
&\quad \times 
   \sum_{\substack{L\subset[n]\\|L|=r_j+r_k-1}}
   \frac{(r_j-1)!\,(r_k-1)!}{(r_j+r_k-1)!}\,\bigl(1+o(1)\bigr) \\
&\quad \times 
   \sum_{\substack{i_u \in L}}
   \sum_{\substack{I\subset L\setminus\{i_u\}\\|I|=r_j-1}}
   \sum_{\substack{S_j\cong\mathcal R_j\\
                   \mathcal V(S_j)=\{i_u\}\cup I}}
   \sum_{\substack{S_k\cong\mathcal R_k\\
                   \mathcal V(S_k)=L\setminus I}}
   \mathbbm{1}\{S_j\cup S_k\cong \mathcal M^*\} \\
&\quad \times 
   \prod_{\{a,b\}\in \mathcal E(S_j\cup S_k)} A_{ab}
   \;+\; R_{n,jk}.
\end{align*}
where the $o(1)$ term is error associated with the following remainder:
\begin{align*}
r_n = \left|\frac{\binom{n}{r_j+r_k-1}}
{%
n\,\binom{n-1}{r_j-1}\,\binom{n-1}{r_k-1}\,%
}/\frac{(r_j-1)!(r_k-1)!}{(r_j+r_k-1)!} - 1 \right|. 
\end{align*}
By Stirling's approximation, this error is exceedingly small; for notational convenience, henceforth, we will suppress this error term. 

For the remainder term $R_{n,jk}$, we have
\begin{align*}
R_{n,jk} 
&= \frac{1}{|\mathrm{Iso}(\mathcal{R}_j)|\,|\mathrm{Iso}(\mathcal{R}_k)|}
\sum_{c=2}^{\min(r_j, r_k)} 
\sum_{d=0}^{\min(s_j, s_k)} O\!\left(n^{1-c}\hat{\rho}_n^{-d}\right)
\sum_{\mathcal M^*\in\mathfrak{M}^{\mathcal R_j,\mathcal R_k}_{c,d}}
\\
& \quad \times
\frac{1}{\binom{n}{r_j+r_k-c}\,\hat\rho_n^{\,s_j+s_k-d}}
\sum_{\substack{L\subset[n]\\|L|=r_j+r_k-c}}
\left\{\sum_{i_u\in L}
\sum_{\substack{I\subset L\setminus\{i_u\}\\|I|=c-1}}
\sum_{\substack{J\subset L\setminus(\{i_u\}\cup I)\\|J|=r_j-c}}
\right.
\\[2pt]
& \qquad \times
\left.
\sum_{\substack{S_j\cong\mathcal R_j\\
\mathcal{V}(S_j)= \{i_u\}\cup I \cup J}}
\sum_{\substack{S_k\cong\mathcal R_k\\
\mathcal{V}(S_k)= (L \setminus \mathcal{V}(S_j)) \cup \{i_u\} \cup I}}
\,
\mathbbm{1}\{S_j\cup S_k\cong\mathcal M^*\}
\prod_{\{a,b\}\in \mathcal{E}(S_j\cup S_k)} A_{ab}
\right\}.
\end{align*}
Now, for the leading term, we observe the combinatorial identity
\begin{align*}
\frac{(r_j - 1)! \, (r_k - 1)!}{(r_j + r_k - 1)!} 
= \frac{1}{(r_j + r_k - 1) \binom{r_j + r_k - 2}{r_j - 1}}.
\end{align*}

Furthermore, since each distinct isomorphism class in $\sum_{\substack{S_j\cong\mathcal R_j\\\mathcal{V}(S_j)= \{i_u\} \cup I}}
\sum_{\substack{S_k\cong\mathcal R_k\\\mathcal{V}(S_k)=L \setminus I}}
\mathbbm{1}\{S_j\cup S_k\cong\mathcal M^*\}$ corresponds to $K_{\mathfrak{M}^{\mathcal{R}_j,\mathcal{R}_k}_{1}}(\mathcal{M}^*)$ terms, in that case,
\begin{align*}
|\mathrm{Iso}(\mathcal{M}^*)| = (r_j + r_k - 1)\times \binom{r_j + r_k - 2}{r_j - 1}\times K_{\mathfrak{M}^{\mathcal{R}_j,\mathcal{R}_k}_{1}}(\mathcal{M}^*).
\end{align*}
Therefore, for the leading term $S_{n,jk}$, we obtain 
\begin{align*}
S_{n,jk} 
=&\sum_{\mathcal M^*\in\mathfrak{M}^{\mathcal R_j,\mathcal R_k}_{1,0}}
   \frac{K_{\mathfrak{M}^{\mathcal{R}_j,\mathcal{R}_k}_{1}}(\mathcal{M}^*)}
        {|\mathrm{Iso}(\mathcal{R}_j)|\,|\mathrm{Iso}(\mathcal{R}_k)|} \\[4pt]
& \qquad \times \Bigg(
   \frac{1}{\binom{n}{r_j+r_k-1}\,\hat\rho_n^{\,s_j+s_k}}
   \sum_{k_1<\cdots<k_{r_j+r_k-1}}
   \frac{1}{|\mathrm{Iso}(\mathcal{M}^*)|}
   \sum_{\substack{\mathcal{M}\cong \mathcal{M}^*\\
                   \mathcal{V}(\mathcal{M})=\{k_1,\ldots,k_{r_j+r_k-1}\}}}
   \prod_{\{a,b\}\in\mathcal{E}(\mathcal{M})} A_{ab}
   \Bigg) \\[6pt]
=&\sum_{\mathcal{M}^*\in\mathfrak{M}^{\mathcal{R}_j,\mathcal{R}_k}_{1,0}}
   C^{\mathcal{M}^*}_{\mathcal{R}_j,\mathcal{R}_k,1}\,\widehat{Q}(\mathcal{M}^*),
\end{align*}
where
\begin{align*}
C^{\mathcal{M}^*}_{\mathcal{R}_j,\mathcal{R}_k,1} = \frac{K_{\mathfrak{M}^{\mathcal{R}_j,\mathcal{R}_k}_{1}}(\mathcal{M}^*)}{|\mathrm{Iso}(\mathcal{R}_j)|\,|\mathrm{Iso}(\mathcal{R}_k)|}.
\end{align*}

Now for the remainder term, the overall number of isomorphic cases corresponding to the summation $\sum_{i_u\in L}
\sum_{\substack{I\subset L\setminus\{i_u\}\\|I|=c-1}}
\sum_{\substack{J\subset L\setminus(\{i_u\}\cup I)\\|J|=r_j-c}}$ is now
\[
C^c_{jk}=(r_j+r_k-c)\times\binom{r_j+r_k-c-1}{c-1} \times \binom{r_j+r_k-2c}{r_j-c}.
\]
Therefore, 
\[
|\mathrm{Iso}(\mathcal{M}^*)| = C^c_{jk}\times K_{\mathfrak{M}^{\mathcal{R}_j,\mathcal{R}_k}_{c}}(\mathcal{M}^*).
\]
By further rearranging the expression into a global subgraph frequency, we obtain
\begin{align*}
R_{n,jk}
&=\sum_{c=2}^{\min(r_j, r_k)} 
  \sum_{d=0}^{\min(s_j, s_k)} 
  O\!\left(n^{1-c}\hat{\rho}_n^{-d}\right) \,
  C^c_{jk} \sum_{\mathcal M^*\in\mathfrak{M}^{\mathcal R_j,\mathcal R_k}_{c,d}} \frac{K_{\mathfrak{M}^{\mathcal{R}_j,\mathcal{R}_k}_{c}}(\mathcal{M}^*)}
       {|\mathrm{Iso}(\mathcal{R}_j)|\,|\mathrm{Iso}(\mathcal{R}_k)|} \\
  &\quad \times
  \frac{1}{\binom{n}{r_j+r_k-c}\,\hat\rho_n^{\,s_j+s_k-d}} \sum_{k_1<\cdots< k_{r_j+r_k - c}}
  \frac{1}{|\text{Iso}(\mathcal{M}^*)|}
  \sum_{\substack{\mathcal{M} \cong \mathcal{M}^*\\
        \mathcal{V}(\mathcal{M})=\{k_1, \ldots, k_{r_j+r_k - c} \}}}
  \prod_{\{a, b\} \in \mathcal{E}(\mathcal{M})} A_{ab} \\[1ex]
&=\sum_{\substack{2 \leq c \leq \min(r_j, r_k) \\ 0 \leq d \leq \min(s_j, s_k)}} 
  O\!\left(\hat{\rho}_n^{-d} n^{1 - c}\right)\,
  \sum_{\mathcal{M}^* \in \mathfrak{M}^{\mathcal{R}_j, \mathcal{R}_k}_{c,d}} 
  C^{\mathcal{M}^*}_{\mathcal{R}_j,\mathcal{R}_k,c}\,
  \widehat{Q}(\mathcal{M}^*),
\end{align*}
where
\begin{align}\label{appB-eq:structral-constant}
C^{\mathcal{M}^*}_{\mathcal{R}_j,\mathcal{R}_k,c} = \frac{K_{\mathfrak{M}^{\mathcal{R}_j,\mathcal{R}_k}_{c}}(\mathcal{M}^*)}{|\mathrm{Iso}(\mathcal{R}_j)|\,|\mathrm{Iso}(\mathcal{R}_k)|}.
\end{align} 
Analogous derivations also hold for rooted subgraphs. The claim follows.
\end{proof}

\begin{remark}
The analogous representation also holds for $n^{-1}\sum_{i=1}^n \widecheck{Z}_{i,j}\,\widecheck{Z}_{i,k}$, with $\rho_n$ in place of $\hat{\rho}_n$ in each of the global subgraph estimates $\widehat{Q}(\mathcal{M}^*)$.
\end{remark}

For the following linearly weighted case, let \( n^{-1} \sum_{i=1}^n Y_i \widetilde{Z}_{ij} \) and \( n^{-1} \sum_{i=1}^n X_i \widetilde{Z}_{ij} \), where \( \widetilde{Z}_{ij} = \widetilde{Q}_i(\mathcal{R}_j) := E[\widecheck{Q}_i(\mathcal{R}_j) \mid \xi] \), which takes the following form:
\begin{align}
\label{eq-z-tilde-definition}
\widetilde{Z}_{ij} =  \frac{1}{{n-1 \choose r_j-1}} \sum_{\substack{1 \leq i_1 <  \cdots < i_{r_j-1} \leq n, \\ i_1, \ldots, i_{r_j-1} \neq i}} \frac{1}{|\mathrm{Iso}(\mathcal{R}_j)|}
\sum_{\substack{\mathcal{S}: \mathcal{S} \cong \mathcal{R}_j,\\
\mathcal{V}(\mathcal{S}) = \{i,i_1, \ldots, i_{r_j-1}\}}} \prod_{\{a,b\} \in \mathcal{E}(\mathcal{S}) } w(\xi_a,\xi_b).    
\end{align}

We next characterize and control the network noise. Specifically, we establish mild conditions under which this noise vanishes for all such linear quantities. Before doing so, we introduce the first set of sparsity conditions, referred to as the standard sparsity condition, which ensures a negligible rate of $o_P(n^{-1/2})$.

\begin{condition}[Standard Sparsity Conditions]\label{appB:sparsity-cond-weak}
Consider all subgraphs $\mathcal{R}_j$, $j \in [d]$, with $|\mathcal{V}(\mathcal{R}_j)|=r_j$, recall $\lambda_n = n\rho_n$. 
The sparsity requirement is that \emph{all} of the following hold: 
\begin{itemize}
    \item[(a)] For every acyclic subgraph, $\lambda_n = \omega(1)$;
    \item[(b)] For every $\mathcal{R}_j$ that is a simple cycle of length $r_j$, 
          $\lambda_n^{\,r_j} = \omega(n)$;
    \item[(c)] For every $\mathcal{R}_j$ that is neither acyclic nor a simple cycle, 
          $\lambda_n = \omega\!\left(n^{\,1 - \tfrac{2}{\,r_j + R-1}}\right)$, with $R=\max_{k\in [d]} r_k$.
\end{itemize}
\end{condition}

We have the following lemma:
\begin{lemma}[Order of Network Noise for Linear Terms]
\label{appB:lemma2}
Suppose the sparsity condition in Condition~\ref{appB:sparsity-cond-weak} holds and \(E\|X\|^4 < \infty\), \(E(Y^4) < \infty\). Then, for any \(j \in [d]\),
\begin{align*}
\frac{1}{n} \sum_{i=1}^n Y_i \widecheck{Z}_{ij} - \frac{1}{n} \sum_{i=1}^n Y_i \widetilde{Z}_{ij} &= o_P(n^{-1/2}), \\
\frac{1}{n} \sum_{i=1}^n X_i \widecheck{Z}_{ij} - \frac{1}{n} \sum_{i=1}^n X_i \widetilde{Z}_{ij} &= o_P(n^{-1/2}),\\
\frac{1}{n} \sum_{i=1}^n  \widecheck{Z}_{ij} - \frac{1}{n} \sum_{i=1}^n \widetilde{Z}_{ij} &= o_P(n^{-1/2}).
\end{align*}
\end{lemma}

\begin{proof}
The third line is a special case of the first and follows from \citet{bickel2011method} under Condition~\ref{appB:sparsity-cond-weak}; the second line follows analogously. Thus, it remains to establish only the first claim. The representation in Lemma~\ref{appB:lemma1} yields
\begin{align*}
\widehat{\delta}_{n,Y,\mathcal{R}_j}
&= \frac{1}{n}\sum_{i=1}^n Y_i\bigl(\widecheck Z_{ij}-\widetilde Z_{ij}\bigr) \nonumber\\[6pt]
&= \frac{1}{\binom{n}{r_j}}
   \sum_{1 \le k_1<\cdots<k_{r_j}\le n}
 \frac{(r_j^{-1}\sum_{\ell=1}^{r_j}Y_{k_\ell}
   )}{|\operatorname{Iso}(\mathcal R_j)|}
     \sum_{\substack{\mathcal S\cong\mathcal R_j\\\mathcal V(\mathcal S)=\{k_1,\dots,k_{r_j}\}}}
     \left(
       \rho_n^{-s_j}\prod_{(a,b)\in \mathcal{E}(\mathcal{S})}A_{ab}
       \!-\!
       \prod_{(a,b)\in \mathcal{E}(\mathcal{S})}  w(\xi_a,\xi_b)
     \right). 
\end{align*}
By the law of total variance, and since $E[\widehat{\delta}_{n,Y,\mathcal{R}_j}\mid \tau] = 0$, we obtain
\begin{align*}
\operatorname{var}(\widehat{\delta}_{n,Y,\mathcal{R}_j})
= E\!\left[\operatorname{var}\bigl(\widehat{\delta}_{n,Y,\mathcal{R}_j}\mid \tau\bigr)\right].
\end{align*}

We next calculate the conditional variance,
\(\operatorname{var}\!\left(n^{-1}\sum_{i=1}^n Y_i\widecheck{Z}_{ij}\,\middle|\, \tau\right)\),
following a similar reasoning as \citet{bickel2011method} for the global subgraph frequency estimator \(\widecheck{Q}(\mathcal{R}_j)\). 
For a subgraph \(\mathcal{S}\), define
\begin{align*}
H(\mathcal{S}) := \prod_{(a,b)\in \mathcal{E}(\mathcal{S})} A_{ab}.
\end{align*}

Let \(K = \{k_1, \ldots, k_{r_j}\}\) be an unordered \(r_j\)-tuple of distinct indices. For simplicity, define
\begin{align*}
\widecheck{h}_K
   &:= \left(\frac{1}{r_j}\sum_{\ell=1}^{r_j} Y_{k_\ell}\right)
      \frac{1}{|\operatorname{Iso}(\mathcal R_j)|}
      \sum_{\substack{\mathcal S \cong \mathcal R_j \\ \mathcal V(\mathcal S)=K}}
      \rho_n^{-s_j}\prod_{(a,b)\in \mathcal{E}(\mathcal S)} A_{ab}.
\end{align*}
The conditional variance can then be written as
\begin{align}\label{appB-eq:conditional_varaince_total}
\operatorname{var}\!\left(
   \binom{n}{r_j}^{-1} \sum_K \widecheck{h}_K \,\middle|\, \tau
\right)
= \binom{n}{r_j}^{-2}
  \sum_K \sum_{K'} \operatorname{cov}\!\left(\widecheck{h}_K, \widecheck{h}_{K'} \,\middle|\, \tau \right).
\end{align}

Let \(Y_K = r_j^{-1} \sum_{\ell=1}^{r_j} Y_{k_\ell}\). Then 
\(\widecheck{h}_K
= Y_K\, |\operatorname{Iso}(\mathcal R_j)|^{-1}\,\rho_n^{-s_j}
  \sum_{\mathcal S \in \mathfrak S_K} H(\mathcal S)\), 
where \(\mathfrak S_K\) is the set of all subgraphs on vertex set \(K\) that are isomorphic to \(\mathcal R_j\). The conditional covariance of \(\widecheck h_K\) and \(\widecheck h_{K'}\) is therefore
\begin{align}\label{appB-eq:noisy-kernel-covaraince}
\operatorname{cov}(\widecheck{h}_K, \widecheck{h}_{K'} \mid \tau) = (Y_KY_{K'})|\text{Iso}(\mathcal{R}_j)|^{-2}\rho_n^{-2s_j}  \sum_{\mathcal{S} \in \mathfrak{S}_K} \sum_{\mathcal{S}' \in \mathfrak{S}_{K'}} \operatorname{cov}(H(\mathcal{S}), H(\mathcal{S}') \mid \tau).
\end{align}

Now, we are ready to bound the individual covariance. For any two isomorphic terms $\mathcal{S}$ and $\mathcal{S}'$ in \eqref{appB-eq:noisy-kernel-covaraince}) that share $c$ nodes and $d$ edges, the covariance satisfies the upper bound
\begin{align*}
\operatorname{cov}(H(\mathcal{S}),H(\mathcal{S}')) 
&\leq \rho_n^{2s_j-d} \prod \{  w(\xi_{a}, \xi_{b}) : (a, b) \in \mathcal{E}(\mathcal{M}^{(\mathcal{R}_j)}(c,d)) \}
,
\end{align*}
 where $\mathcal{M}^{(\mathcal{R}_j)}(c,d) := \{ \mathcal{S} \cup \mathcal{S}' : |\mathcal{V}(\mathcal{S}) \cap \mathcal{V}(\mathcal{S}')| = c, |\mathcal{E}(\mathcal{S}) \cap \mathcal{E}(\mathcal{S}')| = d \}$. This covariance is non-zero only if $d \ge 1$.

Let $K_0'$ be a constant that absorbs the term $(|\mathrm{Iso}(\mathcal{R}_j)|)^{-2}$. The covariance term \eqref{appB-eq:noisy-kernel-covaraince} can be bounded by
\begin{align*}
\operatorname{cov}(\widecheck{h}_K, \widecheck{h}_{K'} \mid \tau ) 
&\leq \sum_{d} K_0' \cdot Y_K Y_{K'} \rho_n^{-2s_j} \cdot \left( \rho_n^{2s_j-d} \prod \{ w(\xi_{a}, \xi_{b}): (a, b) \in \mathcal{E}(\mathcal{M}^{(\mathcal{R}_j)}(c,d)) \} \right).
\end{align*}

Now, we consider the overall summation for the conditional variance. The summation $\sum_K \sum_{K'}$ can be rearranged indexed by $c$, the number of shared nodes. The number of pairs $(K,K')$ that share $c$ nodes and for which there exist $\mathcal S \in \mathfrak{S}_K$, $\mathcal S' \in \mathfrak{S}_{K'}$ such that $|\mathcal{E}(\mathcal S) \cap \mathcal{E}(\mathcal S')| = d$, is $N(c,d) = O(n^{2r_j-c})$. Substituting $N(c,d) = O(n^{2r_j-c})$ and recalling $\binom{n}{r_j}^{-2} = O(n^{-2r_j})$, and under the assumption that $w(x,y)\leq C$, the conditional variance \eqref{appB-eq:conditional_varaince_total} is thus bounded by
\begin{align}\label{appB-eq:linear_conditional_variance}
\operatorname{var}\!\bigl(n^{-1}\!\sum_{i=1}^n Y_i\widecheck{Z}_i \,\big|\, \tau\bigr)
&\le n^{-2r_j} \sum_{c,d} N(c,d)
   \sum_{\substack{K,K'\\|K \cap K'|=c}}
   \sum_{\substack{\mathcal S \in \mathfrak{S}_K,\, \mathcal S' \in \mathfrak{S}_{K'}\\
                   |\mathcal{E}(\mathcal S)\cap \mathcal{E}(\mathcal S')|=d}}
   K_0' \, Y_K Y_{K'} \, \rho_n^{-d}
   \prod_{(a,b)\in \mathcal{E}(\mathcal{M}^{(\mathcal{R}_j)}(c,d))} w(\xi_{a}, \xi_{b})
   \nonumber\\
&\le K_2 \sum_{c,d} n^{-c}\rho_n^{-d}
   \sum_{\substack{K,K'\\|K \cap K'|=c}}
   \sum_{\substack{\mathcal S \in \mathfrak{S}_K,\, \mathcal S' \in \mathfrak{S}_{K'}\\
                   |\mathcal{E}(\mathcal S)\cap \mathcal{E}(\mathcal S')|=d}}
   Y_K Y_{K'},
\end{align}
for some constant $K_2$.

To obtain the unconditional variance $\text{var}(\widehat{\delta}_{n,Y,\mathcal{R}_j}) = E[\text{var}(n^{-1}\sum_{i=1}^n Y_i\widecheck{Z}_i\mid \tau)]$, we take the expectation of both sides of \eqref{appB-eq:linear_conditional_variance} and obtain
\begin{align*}
\text{var}(\widehat{\delta}_{n,Y,\mathcal{R}_j}) \leq K_3 \sum_{c,d} n^{-c} \rho_n^{-d} E\left( \sum_{\substack{K,K' \\ |K \cap K'|=c 
}} Y_K Y_{K'} 
\right),   
\end{align*}
for some constant $K_3$.

Since $Y_K = r_j^{-1} \sum_{\ell=1}^{r_j} Y_{k_\ell}$, and similarly for $Y_{K'}$, the term involving $Y_K Y_{K'}$ satisfies $E(Y_K Y_{K'})< \infty$
under the assumption that the $Y_i$ are i.i.d.\ with $E(Y_i^2) < \infty$, which is implied by our assumed condition $E(Y_i^4)<\infty$. Hence,
\begin{align}\label{appB-eq:linear_total_variance_bound}
\operatorname{var}(\widehat{\delta}_{n,Y,\mathcal{R}_j}) \leq K_4 \sum_{c,d} n^{-c}\rho_n^{-d},    
\end{align}
for some constant $K_4$.

In the acyclic case (\(d\le c-1\)), the slowest-decaying term is \((c,d)=(2,1)\), which yields the variance of order
\[
n^{-2}\,\rho_n^{-1} \;=\;(n\rho_n)^{-1}\,n^{-1}=o(n^{-1}),
\]
when $\lambda_n=\omega(1)$. In the case of a simple cycle, the dominant term  is $(c,d)=(r_j,s_j)$ where $r_j=s_j$, yielding variance of order
\[
(n\rho_n)^{-r_j}=n^{-1}(n^{1-r_j}\rho_n^{-r_j})\;=o(n^{-1}),
\]
when $\lambda_n=\omega(n^{1/r_j})$. More generally, for any \(r_j\)-node pattern that contains a cycle, the worst overlap is \(c=r_j\), \(d=r_j(r_j-1)/2\), yielding variance of order
\[
n^{-r_j}\,\rho_n^{r_j(1-r_j)/2} = n^{-1} \left( n^{1-r_j}\,\rho_n^{r_j(1-r_j)/2} \right) = o(n^{-1}),
\]
when \(\lambda_n = \omega\left(n^{1-2/r_j}\right)\). Following the above sparsity regimes for different types of subgraphs, which hold under Condition \ref{appB:sparsity-cond-weak}, $\widehat{\delta}_{n,Y,\mathcal{R}_j}=o_P\bigl(n^{-1/2}\bigr),$ as required.

In the case $\widehat{Z}_{ij} = \widehat{Q}^{\text{rooted}}_i(\mathcal{R}_j)$, where the kernel reduces to \eqref{appB-eq:h_Y_rooted}. It is easy to check that the variance of the corresponding network‐noise term $\widehat{\delta}_{n,Y,\mathcal{R}_j}$ obeys the same bound derived in \eqref{appB-eq:linear_total_variance_bound}.  Hence, under the standing moment assumptions on \(Y_i\) and \(w(\xi_a,\xi_b)\), and the sparsity condition $\lambda_n \to \infty$, one again obtains $\widehat{\delta}_{n,Y,\mathcal{R}_j} = o_P\bigl(n^{-1/2}\bigr)$. 

The analysis for $n^{-1}\sum_{i=1}^nX_i\widecheck{Z}_{ij}- n^{-1}\sum_{i=1}^nX_i\widetilde{Z}_{ij}$ and $n^{-1}\sum_{i=1}^n\widecheck{Z}_{ij}- n^{-1}\sum_{i=1}^n\widetilde{Z}_{ij}$ follows an analogous derivation. The claim follows.
\end{proof}

Our next lemma provides an error rate for quadratic terms. 

\begin{lemma}[Order of Network Noise for Quadratic Terms]
\label{appB:lemma3}
Grant Condition \ref{appB:sparsity-cond-weak}. Then,
\begin{align*}
n^{-1}\sum_{i=1}^n\widecheck{Z}_{ij}\widecheck{Z}_{ik}- n^{-1}\sum_{i=1}^n E(\widecheck{Z}_{ij}\widecheck{Z}_{ik}\mid\xi)=o_P(n^{-1/2}). \\
\end{align*}
\end{lemma}
\begin{proof}
By Proposition \ref{main:prop1},  we have:
\begingroup
\setlength{\abovedisplayskip}{3pt}
\setlength{\belowdisplayskip}{3pt}
\begin{align*}
n^{-1}\sum_{i=1}^n\widecheck{Z}_{ij}\widecheck{Z}_{ik}
= \underbrace{\sum_{\mathcal{M} \in \mathfrak{M}^{\mathcal{R}_j, \mathcal{R}_k}_{1,0}} C^{\mathcal{M}}_{\mathcal{R}_j,\mathcal{R}_k,1}\cdot \widecheck{Q}(\mathcal{M})}_{\text{Leading term, } \widecheck{S}_{n,jk}} + \underbrace{\sum_{\substack{2 \leq c \leq \min(r_j, r_k) \\ 0 \leq d \leq \min(s_j, s_k)}} O(\rho_n^{-d} n^{1 - c}) \sum_{\mathcal{M} \in \mathfrak{M}^{\mathcal{R}_j, \mathcal{R}_k}_{c,d}} C^{\mathcal{M}}_{\mathcal{R}_j,\mathcal{R}_k,c} \cdot \widecheck{Q}(\mathcal{M})}_{\text{Remainder term, } \widecheck{R}_{n,jk}}.    
\end{align*}
\endgroup

We first consider the leading term. For any two subgraphs \(\mathcal{R}_j\) and \(\mathcal{R}_k\) that share exactly one vertex (and no edges), recall
\[
r_j = \bigl|\mathcal V(\mathcal R_j)\bigr|,\quad s_j = \bigl|\mathcal{E}(\mathcal R_j)\bigr|,\qquad
r_k = \bigl|\mathcal V(\mathcal R_k)\bigr|,\quad s_k = \bigl|\mathcal{E}(\mathcal R_k)\bigr|.
\]
The merged leading-term motif \(\mathcal{M}=\mathcal{R}_j \cup \mathcal{R}_k\) then satisfies 
\[
\bigl|\mathcal V(\mathcal M)\bigr| = r_j + r_k - 1,
\qquad 
\bigl|\mathcal E(\mathcal M)\bigr| = s_j + s_k.
\]

Denote $\widetilde{Q}(\mathcal{M}):=E( \widecheck{Q}(\mathcal{M})\mid \xi)$. We now analyze the order of $\widecheck Q(\mathcal{M}) - \widetilde Q(\mathcal{M})$ for $\mathcal{M} \in \mathfrak{M}^{\mathcal{R}_j, \mathcal{R}_k}_{1,0}$, following the variance-covariance calculation in \cite{bickel2011method}, we compare the few cutoffs of potential phase transitions in the leading order in the following. 

If both \(\mathcal R_j\) and \(\mathcal R_k\) are acyclic, then \(\mathcal{M}\) is still acyclic, and $\operatorname{var}(\widecheck Q(\mathcal{M}) - \widetilde Q(\mathcal{M}))$ is of order
\[
O\bigl(n^{-1}(n\,\rho_n)^{-1}\bigr),
\]
which decays as $o(n^{-1})$ when $\lambda_n \to \infty$. 

If \(\mathcal R_j\) is acyclic (\(r_j \geq s_j+1\)) and \(\mathcal R_k\) is a simple cycle \(r_k=s_k\), $\operatorname{var}(\widecheck Q(\mathcal{M}) - \widetilde Q(\mathcal{M}))$ is of order
\[
O\bigl(n^{-1}(n\,\rho_n)^{-1} \vee (n\,\rho_n)^{-r_k}\bigr),
\]
which decays as $o(n^{-1})$ when $\lambda_n=\omega(n^{1/r_k})$.

If both \(\mathcal R_j\) and \(\mathcal R_k\) are simple cycles, then \(s_j+s_k=r_j+r_k\) and $\operatorname{var}(\widecheck Q(\mathcal{M}) - \widetilde Q(\mathcal{M}))$ is of order
\[
 O \bigl( n^{-1}(n\,\rho_n)^{-1} \vee (n\rho_n)^{-\min(r_j,r_s)} \vee n\cdot(n\rho_n)^{-(r_j+r_s)}\bigr),
\]
which decays as $o(n^{-1})$ when $\lambda_n=\omega(n^{1/\min(r_j,r_k)}\vee n^{2/(r_j+r_k)})$. 
Since $\min(r_j,r_k)\le (r_j+r_k)/2$, this condition reduces to 
$\lambda_n=\omega(n^{1/\min(r_j,r_k)})$.

If $\mathcal{R}_j$ or $\mathcal{R}_k$ does not belong to the class of acyclic graphs or simple cycles, we analyze the worst-case variance by treating $\mathcal{M}$ as a general cyclic subgraph. 
In this case, the worst-case decay of 
$\operatorname{var}\bigl(\widecheck Q(\mathcal{M}) - \widetilde Q(\mathcal{M})\bigr)$ 
is $o(n^{-1})$ whenever 
$\lambda_n = \omega\!\left(n^{\,1 - 2/(r_j + r_s - 1)}\right)$.

Therefore, under Condition \ref{appB:sparsity-cond-weak}, we have:
\begin{align*}
  \sum_{\mathcal{M} \in \mathfrak{M}^{\mathcal{R}_j, \mathcal{R}_k}_{1,0}} C^{\mathcal{M}}_{\mathcal{R}_j,\mathcal{R}_k,1}\cdot \left(\widecheck Q(\mathcal{M}) - \widetilde Q(\mathcal{M})\right)
  = o_P\bigl(n^{-1/2}\bigr).    
\end{align*}

For the remainder term, we have
\begin{align*}
\sum_{\substack{2 \leq c \leq \min(r_j, r_k) \\ 0 \leq d \leq \min(s_j, s_k)}} O(\rho_n^{-d} n^{1 - c}) \sum_{\mathcal{M} \in \mathfrak{M}^{\mathcal{R}_j, \mathcal{R}_k}_{c,d}} C^{\mathcal{M}}_{\mathcal{R}_j,\mathcal{R}_k,c} \cdot \left(\widecheck Q(\mathcal{M}) - \widetilde Q(\mathcal{M})\right)=o_P\bigl(n^{-1/2}),
\end{align*}
which follows a similar argument as for the leading term. The claim follows.
\end{proof}

Next, we prepare the proof of the central limit theorem established by Theorem \ref{main:theorem1}. Due to the presence of a bias term that will be described shortly, the following stronger sparsity conditions will be needed:

\begin{condition}[Sparsity Conditions for Non-Vanishing Bias]\label{appB:sparsity-cond-strong}
Consider all subgraphs $\mathcal{R}_j$, $j \in [d]$, with $|\mathcal{V}(\mathcal{R}_j)|=r_j$. 
The sparsity requirement is that \emph{all} of the following hold: 
\begin{itemize}
    \item[(a)] For every acyclic subgraph, $\lambda_n = \omega(n^{1/2})$;
    \item[(b)] For every $\mathcal{R}_j$ that is a simple cycle of length $r_j$, 
          $\lambda_n^{\,r_j} = \omega(n^{3/2})$; 
    \item[(c)] For every $\mathcal{R}_j$ that is neither acyclic nor a simple cycle, 
          $\lambda_n = \omega\left(n^{1-\left\{(2R-3)/R(R-1)\right\}}\right)$, with $R=\max_{k\in [d]} r_k$.
\end{itemize}
\end{condition}

\noindent In what follows, let $\Theta_{jk} = E[\widecheck{Z}_{ij}\widecheck{Z}_{ik}]$ and $\Xi_{jk} = E[Z_{ij} Z_{ik}]$  for $j,k \in [d]$, where $\widetilde{Z}_{ij}$ is defined in \eqref{eq-z-tilde-definition}. The next lemma establishes a mismatch between $\Theta_{jk}$ and $\Xi_{jk}$, which leads to a bias that must be taken into account in later proofs.     

\begin{lemma}[Characterization of Bias]
\label{lemma-bias-count}
For any $j,k \in [d]$,
\begin{align}
\label{eq:mismatch-expression}
\delta_{jk,n}:=\Theta_{jk} - \Xi_{jk} = \sum_{\substack{2\le c\le\min(r_j,r_k)\\0\le d\le\min(s_j,s_k)}}
   O\bigl(n^{1-c}\rho_n^{-d}\bigr)
   \sum_{\mathcal M\in\mathfrak M^{\mathcal R_j,\mathcal R_k}_{c,d}}
     C^{\mathcal M}_{\mathcal{R}_j,\mathcal{R}_k,c}\,Q(\mathcal{M}).
\end{align}
Consequently, if Condition \ref{appB:sparsity-cond-strong} is satisfied, $\Theta_{jk} - \Xi_{jk} = o(n^{-1/2})$.  
\end{lemma}

\begin{proof}
By the representation in Proposition \ref{main:prop1}, it follows that
\begin{align*}
\Theta_{jk}
&=\;E\Bigl(n^{-1}\sum_{i=1}^n\widecheck Z_{ij}\,\widecheck Z_{ik}\Bigr) \nonumber\\
&=
\sum_{\mathcal M\in\mathfrak M^{\mathcal R_j,\mathcal R_k}_{1,0}}
   C^{\mathcal M}_{\mathcal{R}_j,\mathcal{R}_k,1}\,Q(\mathcal{M})
\;+\;
\sum_{\substack{2\le c\le\min(r_j,r_k)\\0\le d\le\min(s_j,s_k)}}
   O\bigl(n^{1-c}\rho_n^{-d}\bigr)
   \sum_{\mathcal M\in\mathfrak M^{\mathcal R_j,\mathcal R_k}_{c,d}}
     C^{\mathcal M}_{\mathcal{R}_j,\mathcal{R}_k,c}\,Q(\mathcal{M}).
\end{align*}

We also have:
\begin{align*}
\Xi_{jk}&:=E\left(Z_{ij}Z_{ik} \right) \nonumber\\
&=E\left[
   \frac{1}{|\text{Iso}(\mathcal{R}_j)|}\sum_{\mathcal S \cong \mathcal{R}_j}E\Bigl(\prod_{\substack{(a,b)\in \mathcal{E}(\mathcal S)\\i\in\mathcal V(\mathcal S)}}w(\xi_a,\xi_b)\Bigm|\xi_i\Bigr)
   \;\times\;
   \frac{1}{|\text{Iso}(\mathcal{R}_k)|}\sum_{\mathcal S \cong \mathcal{R}_k}E\Bigl(\prod_{\substack{(a,b)\in \mathcal{E}(\mathcal S)\\i\in\mathcal V(\mathcal S)}}w(\xi_a,\xi_b)\Bigm|\xi_i\Bigr)
\right] \nonumber \\
&=
\sum_{\mathcal M\in\mathfrak M^{\mathcal R_j,\mathcal R_k}_{1,0}} C^{\mathcal M}_{\mathcal{R}_j,\mathcal{R}_k,1} E\left[\prod_{(a,b)\in \mathcal{E}(\mathcal M)} w(\xi_a,\xi_b)\right] \nonumber\\
&=\sum_{\mathcal M\in\mathfrak M^{\mathcal R_j,\mathcal R_k}_{1,0}}
   C^{\mathcal M}_{\mathcal{R}_j,\mathcal{R}_k,1}\,Q(\mathcal{M}).
\end{align*}

Now for the second claim, let $r_{\min}=\min(r_j,r_k)$ and $r_{\max}=\max(r_j,r_k)$. If at least one of the two subgraphs is acyclic, the dominant term on the RHS of \eqref{eq:mismatch-expression} is $O(\lambda_n^{-1})$, so that $\delta_{jk,n}=o(n^{-1/2})$ whenever $\lambda_n=\omega(n^{1/2})$. When both subgraphs are simple cycles, the leading term is $O\bigl(n\,\lambda_n^{-r_{\min}}\bigr)$ in the squared case (and $O(\lambda_n^{-1})$ for a cross-product), which yields $\delta_{jk,n}=o(n^{-1/2})$ as soon as $\lambda_n=\omega\bigl(n^{3/(2r_{\min})}\bigr)$. If at least one subgraph is a general (non-simple) cyclic graph, the worst-case bound requires 
\[
\lambda_n=\omega\!\left(n^{1-(2r_{\max}-3)/\left\{r_{\max}(r_{\max}-1)\right\}}\right).
\]
Thus, under Condition~\ref{appB:sparsity-cond-strong}, the claim follows.
\end{proof}

\subsection{Additional Details Regarding Computation of Combinatorial Constants in Proposition \ref{main:prop1}}\label{additional_details}
The representation in Proposition~\ref{main:prop1} involves a sum over 
$\mathcal{M}\in\mathfrak{M}^{\mathcal{R}_j,\mathcal{R}_k}_{c,d}$, 
where $(c,d)$ denotes the node–edge overlaps. 
Not all $(c,d)$ yield non-empty families; whether a class of subgraphs corresponding to a given pair is non-empty depends on the structure of $\mathcal{R}_j,\mathcal{R}_k$.  

For cross terms ($j\neq k$) involving acyclic subgraphs or simple cycles, the overlap must be acyclic, so $c \ge d+1$. 
If both $\mathcal{R}_j$ and $\mathcal{R}_k$ are acyclic, then $d \in \{0,\ldots,\min(s_j,s_k)\}$. 
If both are simple cycles, then $d \in \{0,\ldots,\min(s_j,s_k)-1\}$, excluding full embedding of the smaller cycle. 
If $\mathcal{R}_j$ is acyclic and $\mathcal{R}_k$ is a simple cycle, then 
$d \in \{0,\ldots,s_k-1\}$ when $s_j \ge s_k$, and $d \in \{0,\ldots,s_j\}$ when $s_j < s_k$. If $\mathcal{R}_j$ is a general cyclic subgraph, the admissible set of $d$ depends on its structure, with the generic bound $d \le c(c-1)/2$. 

For squared terms ($j=k$), if $\mathcal{R}_j$ is a simple cycle then $d \in \{0,\ldots,s_j-2\} \cup \{s_j\}$, where $d=s_j$ corresponds to full overlap, $d=s_j-1$ is excluded, and the node condition is $c \ge d+1$ for $d \le s_j-2$ and $c=d$ for $d=s_j$. 
If $\mathcal{R}_j$ is acyclic, then $d \in \{0,\ldots,s_j\}$ with $c \ge d+1$. 
If $\mathcal{R}_j$ is a general cyclic subgraph, then again the admissible $d$ depend on its structure, with the universal bound $d \le c(c-1)/2$.

To implement the bias correction, the term $C^{\mathcal{M}}_{\mathcal{R}_j,\mathcal{R}_k,1}$ (see \eqref{appB-eq:structral-constant}) must be computed.  This term depends on the subgraph of interest.  For simpler subgraphs, these constants can readily be worked out, and efficient subgraph counting algorithms may be used for each isomorphism class.  However, for more complicated subgraphs, working these cases out by hand becomes more challenging; even in these cases, however, a brute force algorithm can be used as a last resort to compute the appropriate terms.  In what follows, we first consider a relatively simple example involving two-stars, and then we present a more general algorithm that can be used when working out these constants analytically becomes challenging.  

\begin{example}
Consider the case where \( \widehat{Z}_{ij} \) denotes the local two-star statistic—where the root node can be any node within the two-star; that is:
\begin{align*}
\widehat{Z}_{ij}
= \frac{1}{\binom{n-1}{2}\,\hat{\rho}_n^{2}}
\sum_{\substack{k<\ell\\ k,\ell\neq i}}
\frac{1}{3}\Big( A_{ik}A_{i\ell} + A_{ik}A_{k\ell} + A_{i\ell}A_{k\ell} \Big).
\end{align*}

Consider the leading term of the squared statistic \( n^{-1} \sum_{i=1}^n \widehat{Z}_{ij}^2 \), which corresponds to the multiplication of two local two-stars attached to the same root node \( i \), assuming no shared edges. For a realization of two disjoint neighborhoods \( (i,j,k) \) and \( (i,l,m) \), the product of the two isomorphic collections at node \( i \) yields 9 motif instances, which can be grouped into 3 distinct motif types as illustrated in Figure~\ref{fig:leading_term_motifs}. 

The relative frequency of each motif type in this expansion also corresponds to its global weight contribution in the leading term. The collection of overlapping motifs for the case \( c = 1 \) includes: 4-node wheel, tree-type motifs, and length-4 walk motifs, with relative weights 
\[
C^{\mathcal{M}}_{\mathcal{R}_j,\mathcal{R}_j,1} = 
\begin{cases}
1/9 & \text{for size-4 wheel} \quad \mathcal{M}_1, \\
4/9 & \text{for ``Y"-type tree} \quad \mathcal{M}_2, \\
4/9 & \text{for length-4 walks} \quad \mathcal{M}_3.
\end{cases}
\]
for each \( \mathcal{M} \in \mathfrak{M}^{\mathcal{R}_j,\mathcal{R}_j}_{1} \), where \( \mathcal{R}_j \) denotes the two-star, which can also be seen from Figure~\ref{fig:leading_term_motifs}.

Therefore, by knowing these distinct motifs and their relative weights, the leading term can be exactly characterized as a weighted average of global subgraph frequency estimators:
\begin{align*}
    \frac{1}{n} \sum_{i=1}^n \widehat{Z}_{ij}^2 \approx \frac{1}{9} \widehat{Q}(\mathcal{M}_1) + \frac{4}{9} \widehat{Q}(\mathcal{M}_2) + \frac{4}{9} \widehat{Q}(\mathcal{M}_3),
\end{align*}
where each \( \widehat{Q}(\mathcal{M}_\ell) \) corresponds to the global frequency estimate of motif \( \mathcal{M}_\ell \) for \( \ell = 1,2,3 \). These estimators admit efficient computation, and the number of isomorphic embeddings can also be exactly determined:
\[
|\mathrm{Iso}(\mathcal{M}_1)| = 5, \quad |\mathrm{Iso}(\mathcal{M}_2)| = 60, \quad |\mathrm{Iso}(\mathcal{M}_3)| = 60.
\]
This structural decomposition makes the computation of the global leading term both fast and interpretable.
\end{example}

\begin{figure}[htbp]
    \centering
    \includegraphics[width=0.90\textwidth]{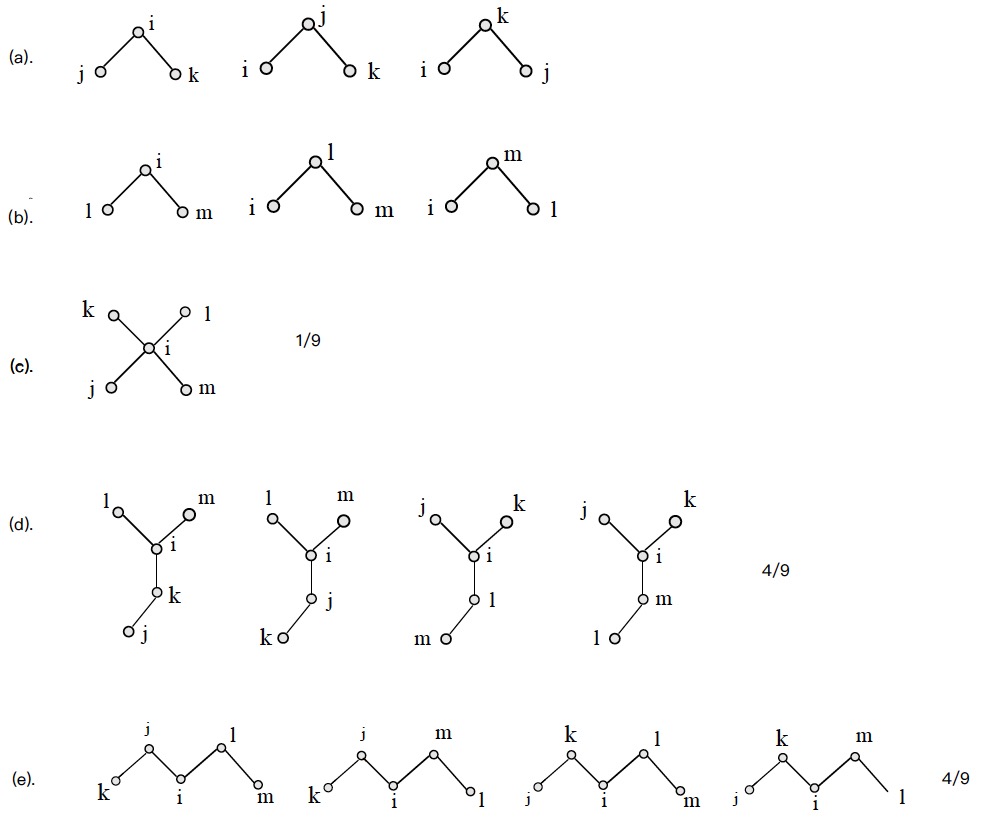} %
    \caption{
    Illustration of leading term motifs arising from the multiplication of two isomorphic two-star patterns rooted at the same node $i$, with no shared edges.  
    (a) Isomorphic two-star structure over nodes $(i,j,k)$.  
    (b) Isomorphic two-star structure over nodes $(i,l,m)$.  
    Multiplying these two configurations yields 9 distinct overlapping cases, which fall into three motif types:  
    (c) size-4 wheel motif, contributing a relative weight of $1/9$;  
    (d) Tree-type motifs, contributing a total relative weight of $4/9$;  
    (e) Length-4 walk motifs, contributing a total relative weight of $4/9$.
    }
    \label{fig:leading_term_motifs}
\end{figure}

The above example represents a simple case where all overlapping motifs can be enumerated explicitly. In more complex configurations, the overlapping motifs and their weights can instead be computed implicitly using the brute-force search procedure illustrated in Algorithm~\ref{appB-ALG}.

\begin{algorithm}[H]
\caption{Identification of Leading Term  for \( n^{-1} \sum_{i=1}^n \widehat{Z}_{ij}\widehat{Z}_{ik} \)}
\label{appB-ALG}
\KwIn{Graph $G = (V, E)$; subgraph patterns $\mathcal{R}_j$, $\mathcal{R}_k$ with node–edge parameters $(r_j,s_j)$ and $(r_k,s_k)$, respectively.}
\KwOut{Leading term approximation of \( n^{-1} \sum_{i=1}^n \widehat{Z}_{ij}\widehat{Z}_{ik} \).}

Estimate  \(\hat{\rho}_n\) by \(\frac{2}{n(n-1)} \sum_{1 \le i < j \le n} A_{ij}\)\;
\( C \gets 0 \)\;

\ForEach(\tcp*[f]{Loop over all nodes}){\( i \in V \)}{
  \( c_i\gets 0 \)\;

  \ForEach(\tcp*[f]{All choices of \( r_j{-}1 \) nodes excluding i}){\( V_j \subset V \setminus \{i\},\ |V_j| = r_j - 1 \)}{

    \ForEach{\( V_k \subset V \setminus (\{i\} \cup V_j),\ |V_k| = r_k - 1 \)\tcp*{All \( r_k{-}1 \) nodes excluding \( \{i\} \cup V_j \)}}{
      \( S_j \gets \{i\} \cup V_j,\quad S_k \gets \{i\} \cup V_k \)\;
      \If(\tcp*[f]{Both patterns match}){\( G[S_j] \cong \mathcal{R}_j \) and \( G[S_k] \cong \mathcal{R}_k \)}{
        \If(\tcp*[f]{Only shared node is \( i \)}){\( S_j \cap S_k = \{i\} \)}{
          \( c_i \gets c_i + 1 \)\;
        }
      }
    }
  }
  \( C \gets C + c_i \)\;
}
\Return \(\dfrac{C}{\binom{n}{r_j+r_k-1} \hat{\rho}_n^{\,s_j+s_k}}\)\;
\end{algorithm}

\subsection{Proofs of Main Results for Network Regression with Local Subgraph Frequencies}\label{sec-B.3}

\begin{proof}[Proof of Theorems \ref{theorem 1} and \ref{theorem 2}.]
To analyze the asymptotic distribution of the OLS estimator \( \widehat{\beta} \), 
we begin with the intermediate quantity \( \widecheck{L}_i := (X_i, \widecheck{Z}_i) \), 
where \( \widecheck{Z}_i \) uses network covariates normalized by the true sparsity level \( \rho_n \) 
(see definition~\eqref{eq-local-subgraph-proxy}). 
Define
\begin{align*}
\widecheck{\Psi}_n &:= 
\begin{pmatrix}
\mathrm{vec}(\widecheck{\Lambda}) \\
\widecheck{\gamma}
\end{pmatrix}, 
\quad \text{where} \quad 
\widecheck{\Lambda} = \frac{1}{n} \sum_{i=1}^n \widecheck{L}_i \widecheck{L}_i^\T, 
\quad
\widecheck{\gamma} = \frac{1}{n} \sum_{i=1}^n Y_i \widecheck{L}_i, \\[0.5em]
\Psi &:= 
\begin{pmatrix}
\mathrm{vec}(\Lambda) \\
\gamma
\end{pmatrix}.
\end{align*}

Furthermore, define
\begin{align*}
\widecheck{J}_n &:=
\begin{pmatrix}
    \widecheck{\Psi}_n\\[4pt]
    \widehat{\rho}_n/\rho_n
\end{pmatrix},
\qquad
\eta^* :=
\begin{pmatrix}
    \Psi \\[4pt]
    1
\end{pmatrix},
\qquad
\widetilde{\eta}_n :=
\begin{pmatrix}
    E(\widecheck{\Psi}_n) \\[4pt]
    1
\end{pmatrix}
= E(\widecheck{J}_n).
\end{align*}

For Theorem~\ref{theorem 1}, our goal is to establish
\begin{align*}
n^{1/2}(\widecheck{J}_n - \eta^*) 
&= n^{1/2}\,
\begin{pmatrix}
  \widecheck{\Psi}_n - \Psi \\[4pt]
  \widehat{\rho}_n / \rho_n - 1
\end{pmatrix}
\;\;\to\;\;
N\!\left(0,\Sigma_\Psi\right) \quad \text{in distribution},
\end{align*}
for some covariance matrix \( \Sigma_\Psi \). 
The desired result then follows from a delta method argument.

Consider the following decomposition of the centered and scaled joint estimator:
\begin{align}\label{appB-eq:full-OLS-joint-decomp}
n^{1/2}(\widecheck{J}_n - \eta^*)
&=
\underbrace{n^{1/2}\left( \widecheck{J}_n - E[\widecheck{J}_n \mid \tau] \right)}_{I}
+
\underbrace{n^{1/2}\left( E[\widecheck{J}_n \mid \tau] - E(\widecheck{J}_n)\right)}_{II} + \underbrace{n^{1/2} \left(E(\widecheck{J}_n)-\eta^* \right)}_{III}.
\end{align}

By Lemmas \ref{appB:lemma2} and \ref{appB:lemma3}, together with the stated conditions, it follows that
\begin{align*}
n^{-1}\sum_{i=1}^nY_i\widecheck{Z}_{ij}- n^{-1}\sum_{i=1}^nY_i\widetilde{Z}_{ij} =o_P(n^{-1/2}), \\
n^{-1}\sum_{i=1}^nX_i\widecheck{Z}_{ij}- n^{-1}\sum_{i=1}^nX_i\widetilde{Z}_{ij} =o_P(n^{-1/2}),\\
n^{-1}\sum_{i=1}^n\widecheck{Z}_{ij}\widecheck{Z}_{ik}- n^{-1}\sum_{i=1}^n E(\widecheck{Z}_{ij}\widecheck{Z}_{ik} \mid \xi)=o_P(n^{-1/2}).
\end{align*}
Moreover, by Theorem 1 of \cite{bickel2011method}, it also follows that, so long as $\lambda_n \rightarrow \infty$, 
\begin{align*}
\widehat\rho_n / \rho_n - E(\widehat\rho_n / \rho_n \mid \xi)= \widehat\rho_n / \rho_n-\frac{1}{\binom{n}{2}}\sum_{i<j}w(\xi_i,\xi_j)=o_P(n^{-1/2}).
\end{align*}
Therefore, $I \rightarrow 0 \text{ in probability.}$  By Lemma \ref{lemma-bias-count}, $III \rightarrow 0$ under the assumed conditions.  

It remains to establish a CLT involving $II$.  Coordinates corresponding to conventional covariates are U-statistics of order $1$, cross-terms of the form $n^{-1}\sum_{i=1}^nY_i\widetilde{Z}_{ij}$ or $n^{-1}\sum_{i=1}^nX_i\widetilde{Z}_{ij} $ for $j \in [d]$, are U-statistics of order $r_j$, and terms related to $n^{-1}\sum_{i=1}^n E(\widecheck{Z}_{ij}\widecheck{Z}_{ik} \mid \xi)$ for $j,k \in [d]$ can be well-approximated by U-statistics of order $r_j + r_k -1$. In particular, the argument used in the proof of Lemma \ref{appB:lemma3} implies that:
\begin{align*}
& n^{-1} \sum_{i=1}^n \left\{E[\widecheck{Z}_{ij}\widecheck{Z}_{ik} \mid \xi] - E[\widecheck{Z}_{ij}\widecheck{Z}_{ik}]\right\} \nonumber\\
&= \sum_{\mathcal{M} \in \mathfrak{M}^{\mathcal{R}_j,\mathcal{R}_k}_{1,0}}
   \frac{C^{\mathcal{M}}_{\mathcal{R}_j,\mathcal{R}_k,1}}{{n \choose r_j+r_k -1}}
   \sum_{1 \leq i_1,\ldots,i_{r_j+r_k-1} \leq n}
    h_{\mathcal{M}}(\xi_1,\ldots,\xi_{r_j+r_k-1})
+ o_P(n^{-1/2}),
\end{align*}
where
\begin{align*}
h_{\mathcal{M}}(\xi_1,\ldots,\xi_{r_j+r_k-1}) = \frac{1}{|\mathrm{Iso}(\mathcal{M})|}
\sum_{\substack{\mathcal{S}: \mathcal{S} \cong \mathcal{M},\\
\mathcal{V}(\mathcal{S}) = \{i_1, \ldots, i_{r_j+r_k-1}\}}} \prod_{\{i,j\} \in \mathcal{E}(\mathcal{S}) } w(\xi_i,\xi_j).    
\end{align*}

Now, for concreteness, suppose that the $j$th coordinate of $II$ is well-approximated by a U-statistic of order $k_j$ with kernel $h_j( \tau_1, \ldots, \tau_{k_j})$. 
 
A H\'ajek projection argument yields: 
\begin{align*}
II \rightarrow N(0,\Sigma_\Psi) \text{ in distribution},
\end{align*}
where $\Sigma_{\Psi,ij}= \operatorname{cov}(k_i E[h_i( \tau_1, \ldots, \tau_{k_i}) \ | \ \tau_1], k_j E[h_j( \tau_1, \ldots, \tau_{k_j}) \ | \ \tau_1])$.

Next, recall that the OLS estimator satisfies $\widehat{\beta} = f \circ g \left( \operatorname{vec}(\widecheck{\Lambda}) , \widecheck{\gamma}, \hat{\rho}_n/\rho_n \right)$. Since $f$ is appropriately differentiable at $(\Lambda, \gamma)$ by Lemma~\ref{appB:OLS_functional_property} and $g$ is differentiable at $(\Lambda, \gamma, 1)$, it follows by the Chain Rule that $f \circ g$ is differentiable at $(\Lambda, \gamma, 1)$. Therefore, by the Delta Method, 
\begin{align*}
n^{1/2}(\widehat{\beta} - \beta^*) = n^{1/2}(f(\widehat{\Psi}_n) -f(\Psi)) \to  N(0, \Sigma_\beta) \quad \textrm{in distribution},     
\end{align*}
where
\begin{align}\label{CLT_covaraince_beta}
\Sigma_\beta = D_{f \circ g}(\Lambda, \gamma, 1) \  \Sigma_{\Psi} \  D_{f \circ g}(\Lambda, \gamma, 1)^\T .
\end{align}
This completes the proof of Theorem \ref{theorem 1}. Theorem \ref{theorem 2} follows similar reasoning, but the term $III$ in \eqref{appB-eq:full-OLS-joint-decomp} does not appear in an analogous decomposition.  Then, Condition \ref{appB:sparsity-cond-weak} suffices, and the result follows. 
\end{proof}

\begin{proof}[Proof of Theorem \ref{theorem 3}.]
This proof largely follows the structure of Theorems \ref{theorem 1} and \ref{theorem 2}, but for completeness, we highlight several crucial steps. Define the quantities:
\begin{align*}
\widecheck{S}_{n,jk} 
&= \sum_{\mathcal{M} \in \mathfrak{M}^{\mathcal{R}_j, \mathcal{R}_k}_{1,0}} 
   C^{\mathcal{M}}_{\mathcal{R}_j,\mathcal{R}_k,1}\,\widecheck{Q}(\mathcal{M}), 
\qquad 
\widecheck{\Lambda}^{\mathrm{mod}}
= 
\begin{pmatrix}
  n^{-1}\!\sum_{i=1}^n X_i X_i^\top 
  & n^{-1}\!\sum_{i=1}^n X_i \widecheck{Z}_i^\top \\[1.2em]
  n^{-1}\!\sum_{i=1}^n \widecheck{Z}_i X_i^\top 
  & \widecheck{S}_n
\end{pmatrix}.
\end{align*}
Moreover, let:
\begin{align*}
\widecheck{\Psi}_n^{\text{mod}} &:= 
\begin{pmatrix}
\mathrm{vec}(\widecheck{\Lambda}^{\text{mod}}) \\
\widecheck{\gamma}
\end{pmatrix}, \qquad \widecheck{J}_n^{\text{mod}} :=
\begin{pmatrix}
    \widecheck{\Psi}_n^{\text{mod}}\\
    \widehat{\rho}_n/\rho_n
\end{pmatrix}.
\end{align*}

By an analogous representation to \eqref{appB-eq:full-OLS-joint-decomp}, 
\begin{align*}
n^{1/2}(\widecheck{J}_n^{\mathrm{mod}} - \eta^*)
&=
\underbrace{n^{1/2}\left( \widecheck{J}_n^{\mathrm{mod}} - E[\widecheck{J}_n^{\mathrm{mod}} \mid \tau] \right)}_{I}
+
\underbrace{n^{1/2}\left( E[\widecheck{J}_n^{\mathrm{mod}} \mid \tau] - E(\widecheck{J}_n^{\mathrm{mod}})\right)}_{II}\\
& \qquad + \underbrace{n^{1/2} \left(E(\widecheck{J}_n^{\mathrm{mod}})-\eta^* \right)}_{III}.
\end{align*}
The term $I \to 0$ in probability by the same reasoning as before.  
For $III$, note that the lower-order components are absent in $E(\widecheck{J}_n^{\mathrm{mod}})$, so this term vanishes identically.  
For $II$, the elements of the vector corresponding to $\widecheck{S}_n$ can be approximated by the same U-statistic considered in the proofs of Theorems~\ref{theorem 1} and~\ref{theorem 2}.  
The result then follows by the delta method. 
\end{proof}

\subsection{Proofs for Section \ref{sec-3.2:bootstrap_consistency}}\label{appB-sec:bootstrap}

\begin{proof}[Proof of Theorem \ref{graphon:bootstrap}.]
To facilitate the proof, we introduce
\begin{align*}
\widehat{J}_n &= (\widehat{\Psi}_n,\;1),
\qquad 
\widehat{J}^{\mathrm{mod}}_n = (\widehat{\Psi}^{\mathrm{mod}}_n,\;1),
\end{align*}
both of which are vectors in $\mathbb{R}^{q+1}$.

The corresponding bootstrapped quantities, including the estimated sparsity parameter, are expressed as:
\begin{align}\label{subgraph:bootstrap_OLS}
\widehat{J}_{n}^\flat
= \widehat{J}_n + \frac{1}{n}\sum_{i=1}^{n}(W_i-1)\,D_k \widehat{g}_1(i), 
\qquad
\widehat{J}_{n}^{\mathrm{mod},\flat}
= \widehat{J}^{\mathrm{mod}}_n + \frac{1}{n}\sum_{i=1}^{n}(W_i-1)\,D_k \widehat{g}_1(i),
\end{align}
that satisfy $\widehat{J}_{n}^\flat=(\widehat{\Psi}_{n}^\sharp,\;\hat{\rho}_n^\flat/\hat{\rho}_n)$ and $\widehat{J}_{n}^{\mathrm{mod},\flat}=(\widehat{\Psi}_{n}^{\mathrm{mod},\sharp},\;\hat{\rho}_n^\flat/\hat{\rho}_n)$, where $\widehat{\Psi}_{n}^\sharp$ and $\widehat{\Psi}_{n}^{\mathrm{mod},\sharp}$ are defined in \eqref{subgraph:bootstrap_OLS-stage1}.

By the construction of the bootstrap procedure in \eqref{subgraph:bootstrap_OLS},  $\operatorname{cov}^{\flat}(\widehat{J}^{\flat}_n)$ and 
$\operatorname{cov}^{\flat}(\widehat{J}_{n}^{\mathrm{mod},\flat})$ 
correspond to the same bootstrap covariance, which is the conditional covariance matrix of $n^{-1}\sum_{i=1}^{n}(W_i-1)\,D_k\widehat{g}_1(i)$ given the data, and it admits the following form:
\begin{align}\label{subgraph: bootstrap variance OLS}
\operatorname{cov}^{\flat}(\widehat{J}^{\flat}_n)=\operatorname{cov}^{\flat}(\widehat{J}_{n}^{\mathrm{mod},\flat})=\frac{1}{ n^2}\sum_{i=1}^nD_k\widehat{g}_1(i)\widehat{g}_1(i)^\T D_k^\T.    
\end{align}
 Let $
\widehat{\Sigma}^\flat := n \, \operatorname{cov}^\flat(\widehat{J}_n^\flat \mid \mathcal{D}_n)$. By Lemma \ref{lem:countbootcons} below, \(\widehat{\Sigma}^\flat \to \Sigma_\Psi\) in probability. 
Since the bootstrapped quantities are sequences of mean zero Gaussian random variables, this further implies that:
\begin{align*}
\sup_{u \in \mathbb{R}^{q+1}} \Big| 
\operatorname{pr}^{\flat} \!\left( n^{1/2}(\widehat{J}_n^{\flat} - \widehat{J}_n) \leq u \right)
-
\operatorname{pr} \!\left( n^{1/2}(\widecheck{J}_n - \widetilde{\eta}_n) \leq u \right)
\Big| &\;\to 0 \quad \text{in probability}, \\[0.75em]
\sup_{u \in \mathbb{R}^{q+1}} \Big| 
\operatorname{pr}^{\flat} \!\left( n^{1/2}(\widehat{J}_n^{\flat,\mathrm{mod}} - \widehat{J}^{\mathrm{mod}}_n) \leq u \right)
-
\operatorname{pr} \!\left( n^{1/2}(\widecheck{J}^{\mathrm{mod}}_n - \eta^*) \leq u \right)
\Big| &\;\to 0 \quad \text{in probability}.
\end{align*}

Further, for \(\widehat{\beta}^{\flat}\) and \(\widehat{\beta}^{\flat,\mathrm{mod}}\) defined in \eqref{subgraph:bootstrap_beta}, the delta method for the bootstrap applies (see Theorem 23.5 of \cite{vanderVaart1998})\footnote{The theorem stated in this reference involves convergence in distribution almost surely, but it also holds for sequences converging in distribution in probability since a sequence of random variables converges in probability if and only if it contains a further subsequence converging almost surely.}, provided that \(\|\widetilde{\eta}_n - \eta^*\| = o_P(1)\) and that \(f \circ g\) is continuously differentiable in a neighborhood of \(\eta^*\), as ensured by Lemma \ref{appB:OLS_functional_property}. The claim follows.
\end{proof}

\begin{lemma}\label{lem:countbootcons}
Under the conditions of Theorem \ref{graphon:bootstrap},
\begin{align}\label{appB-eq:bootstrap-covaraince-consistency}
\widehat{\Sigma}^\flat \to \Sigma_\Psi \quad \textrm{in probability}.
\end{align}
\end{lemma}

\begin{proof}
Before beginning the proof, we define a few quantities that appear in \eqref{subgraph: bootstrap variance OLS}. 
We first define the sample approximation of H\'ajek projection at node $i$, denoted by $\widehat{g}_1(i)$, as
\begin{align}\label{appB-eq:sample-projection-unit}
\widehat{g}_1(i) :=
\begin{pmatrix}
\operatorname{vec}(\widehat{\Lambda}^H(i)) \\
\widehat{\gamma}^H(i) \\
\frac{1}{(n - 1)\hat{\rho}_n} \sum_{j \ne i} A_{ij} -1
\end{pmatrix}.
\end{align}
where \(\widehat{g}_1(i) = \bigl(\widehat{g}_{1,\Psi}(i),\, (n-1)^{-1}\hat{\rho}_n^{-1}\sum_{j\neq i}A_{ij} - 1 \bigr)\), with $\widehat{g}_{1,\Psi}(i)$ introduced in \eqref{subgraph:bootstrap_OLS-stage1}.

Further, we define the node-wise projection terms in \eqref{appB-eq:sample-projection-unit} as
\begin{align}\label{appB-eq:block matrices}
\widehat{\Lambda}^H(i) :=
\begin{pmatrix}
\widehat{\Lambda}_{XX}(i) & \widehat{\Lambda}_{XZ}^H(i) \\[8pt]
(\widehat{\Lambda}_{XZ}^H(i))^\T & \widehat{\Lambda}_{ZZ}^H(i)
\end{pmatrix},
\qquad
\widehat{\gamma}^H(i) :=
\begin{pmatrix}
\widehat{\gamma}_{XY}(i) \\
\widehat{\gamma}_{YZ}^H(i)
\end{pmatrix},
\end{align}
where the individual block matrices in \eqref{appB-eq:block matrices} are specified as follows.  

For the independent components, we denote 
\begin{align*}
\widehat{\Lambda}_{XX}(i) 
&= X_i X_i^\T - n^{-1}\sum_{j=1}^n X_j X_j^\T,
\end{align*}
and $\widehat{\gamma}_{XY}(i)$ can be defined in a similar manner. Note that, the first column of $\mathrm{X}$ is a column of ones; therefore, $X_i-n^{-1}\sum_{j=1}^n X_j$ and $Y_i-n^{-1}\sum_{j=1}^n Y_j$ would appear in the above scenarios as special cases. 

Next, for the subgraph-based components, we define the terms 
$\widehat{\Lambda}_{XZ}^H(i)$, $\widehat{\gamma}_{YZ}^H(i)$ and $\widehat{\Lambda}_{ZZ}^H(i)$ in a following manner.
The $j$th entry of $\widehat{\gamma}_{YZ}^H(i)$ is given by
\begin{align*}
\widehat{\gamma}_{YZ,j}^H(i) &=
\frac{1}{\binom{n-1}{r_j-1}\hat{\rho}_n^{s_j}}
\sum_{\substack{\pi_i \subset [n]\setminus\{i\}\\ |\pi_i|=r_j-1}}
h_{Y,j}(A_{\{i\}\cup \pi_i};\,Y_{\{i\}\cup \pi_i})
- n^{-1}\sum_{l=1}^n Y_l \widehat{Z}_{lj}.
\end{align*}
An analogous expression also holds for $\widehat{\Lambda}_{XZ}^H(i)$, where a special case includes 
\begin{align*}
\frac{1}{\binom{n-1}{r_j-1}\hat{\rho}_n^{s_j}}
\sum_{\substack{\pi_i \subset [n]\setminus\{i\}\\ |\pi_i|=r_j-1}}
 h_{j}(A_{\{i\}\cup \pi_i})
- n^{-1}\sum_{l=1}^n \widehat{Z}_{lj}.  \end{align*}

For $\widehat{\Lambda}_{ZZ}^H(i)$, the $(j,k)$th entry is
\begin{align*}
\widehat{\Lambda}_{ZZ,jk}^H(i) &=
\sum_{\mathcal{M}\in\mathfrak{M}_{c,d}^{\mathcal{R}_j,\mathcal{R}_k}}
C_{\mathcal{R}_j,\mathcal{R}_k,c}^{\mathcal{M}}
\frac{1}{\binom{n-1}{r_j+r_k-2}\hat{\rho}_n^{s_j+s_k}}
\sum_{\substack{\pi_i \subset [n]\setminus\{i\}\\ |\pi_i|=r_j+r_k-2}}h_{\mathcal{M}}(A_{\{i\}\cup \pi_i})
- \widehat{S}_{n,jk},
\end{align*}
where $\widehat{S}_{n,jk}$ denotes the leading term of 
$n^{-1}\widehat{Z}_{ij}\widehat{Z}_{ik}$ 
(as defined in Proposition~\ref{main:prop1}).  
Throughout, we use the shorthand $h_{Y,j} := h_{Y,\mathcal{R}_j}$, $h_{X,j} := h_{X,\mathcal{R}_j}$, and $h_{j} := h_{\mathcal{R}_j}$ for kernel functions, following the global $U$-statistic notation in Lemma~\ref{appB:lemma1}.  
Similarly, $h_{\mathcal{M}}$ denotes the kernel associated with the leading-term $U$-statistic $\widetilde{Q}(\mathcal{M})$ for $\mathcal{M} \in \mathfrak{M}^{\mathcal{R}_j,\mathcal{R}_k}_{1,0}$, as given in Proposition~\ref{main:prop1}.

Now, we are in a position to begin the proof. In what follows, we rearrange $\widehat{\Sigma}^\flat$ with subscripts indicating the component type under consideration. 

Entries of \eqref{subgraph: bootstrap variance OLS} contains three types of terms, depending on whether the entries of \( \widehat{g}_1(i) \) involve independent statistics, subgraph statistics, or both. When both components correspond to independent terms such as \( \widehat{\Lambda}_{XX}(i) \) or \( \widehat{\gamma}_{XY}(i) \), we encounter empirical covariances. In this case, the weak law of large numbers guarantees convergence in probability to the population covariance, under the conditions that $E\|X\|^4<\infty$ and $E(Y^4)<\infty$. 

We next consider cross-terms between independent and subgraph-based components. 
Let $\widehat{u}_i$ denote the subgraph-based projection at node $i$, 
constructed over a subgraph with $r_u$ nodes and $s_u$ edges:
\begin{align}\label{appB:u_i_hat}
\widehat{u}_{i}
   := \frac{1}{\binom{n-1}{r_u - 1}\,\hat{\rho}_n^{s_u}} 
      \sum_{\substack{\pi_i \subset [n]\setminus\{i\} \\ |\pi_i| = r_u - 1}} 
      h_u\!\left(A_{\{i\}\cup\pi_i};\, X_{\{i\}\cup\pi_i}; Y_{\{i\}\cup\pi_i}\right),
 \end{align}
with global average $\widehat{U}_n := n^{-1}\sum_{i=1}^n \widehat{u}_i$. 
The corresponding bootstrap covariance of the this type is then denoted as
\begin{align*}
\widehat{\Sigma}^\flat_{u,X}
   := \frac{r_u}{n} \sum_{i=1}^n 
      \big(\widehat{u}_{i} - \widehat{U}_n\big)\,(X_i - \widebar{X}).
\end{align*}

Finally, for the terms involving two subgraph-based components, let \( \widehat{v}_i \) be another local subgraph projection based on a subgraph with \( r_v \) nodes and \( s_v \) edges:
\begin{align*}
\widehat{v}_{i}:= \frac{1}{\binom{n-1}{r_v - 1} \hat{\rho}_n^{s_v}} 
\sum_{\substack{\pi_i \subset [n] \setminus \{i\} \\ |\pi_i| = r_v - 1}} 
h_v(A_{\{i\} \cup \pi_i};\, X_{\{i\} \cup \pi_i}; Y_{\{i\} \cup \pi_i}),
\end{align*}
with corresponding mean \( \widehat{V}_n := n^{-1} \sum_{i=1}^n \widehat{v}_i \). The bootstrap covariance between \( \widehat{u}_i \) and \( \widehat{v}_i \) takes the form
\begin{align}\label{appB-eq:boot-cov-uv}
\widehat{\Sigma}^\flat_{u,v}:= \frac{r_u r_v}{n} \sum_{i=1}^n (\widehat{u}_i - \widehat{U}_n)(\widehat{v}_i - \widehat{V}_n).
\end{align}
It is worth noting that the term \( \widehat{\rho}_n / \rho_n \) can be regarded as the simplest subgraph estimator, corresponding to
\[
\widehat{u}_i \;=\; \frac{1}{(n - 1)\hat{\rho}_n} \sum_{j \ne i} A_{ij}.
\]

Hence, proving \eqref{appB-eq:bootstrap-covaraince-consistency} reduces to showing that
\begin{align}\label{appB-eq:bootstrap-covariance-2-cases}
\widehat{\Sigma}^\flat_{u,X} \;\to\; \Sigma_{u,X} := \Sigma_{\Psi,u,X}
\quad \text{in probability}, 
\qquad
\widehat{\Sigma}^\flat_{u,v} \;\to\; \Sigma_{u,v} := \Sigma_{\Psi,u,v}
\quad \text{in probability}.
\end{align}
where 
\begin{align}\label{limiting_covariances}
\Sigma_{u,X}=r_u\operatorname{cov}\! \big( h_{1,u}(\tau_i), X_i\big),
\qquad
\textrm{and}
\quad
\Sigma_{u,v}=r_ur_v\operatorname{cov}\!\big( h_{1,u}(\tau_i),  h_{1,v}(\tau_i)\big).
\end{align}
with $ h_{1,u}(\tau_i):=E(h_u(\cdot)|\tau_i)$ and $ h_{1,v}(\tau_i):=E(h_v(\cdot)|\tau_i)$. Moreover, $\operatorname{cov}(h_u(\tau_i),h_v(\tau_i)) 
= E[(h_{1,u}(\tau_i)-\mu_U)(h_{1,v}(\tau_i)-\mu_V)]$, 
where $\mu_U=E[h_{1,u}(\tau_i)]$ and $\mu_V=E[h_{1,v}(\tau_i)]$.
The first convergence in \eqref{appB-eq:bootstrap-covariance-2-cases} is a 
special case of the second, obtained by viewing $X_i$ as a
one-node subgraph statistic. The above claim follows from 
Lemma~\ref{appB-lem:bootstrap-lemma-3} and thus completes the proof.
\end{proof}

\begin{lemma}[Cross-Covariance of Two Subgraph Terms]\label{appB-lem:bootstrap-lemma-3}
Under the same sparsity condition as Condition \ref{appB:sparsity-cond-weak}. And further, assume $E(Y^4)<\infty$, $E\|X\|^4<\infty$. 
Then,
\begin{align*}
\widehat{\Sigma}^\flat_{u,v} \to \Sigma_{u,v}
\quad\textrm{in probability}. 
\end{align*}
with $\widehat{\Sigma}^\flat_{u,v}$ and $\Sigma_{u,v}$ defined in \eqref{appB-eq:boot-cov-uv} and \eqref{limiting_covariances}, respectively.
\end{lemma}
\begin{proof}
We define an intermediate quantity for the following use:
\begin{align*}
u_{i}:= \frac{1}{\binom{n-1}{r_u - 1}} 
\sum_{\substack{\pi_i \subset [n] \setminus \{i\} \\ |\pi_i| = r_u - 1}} 
h_u(w_{\{i\} \cup \pi_i};\, X_{\{i\} \cup \pi_i}, Y_{\{i\} \cup \pi_i}),
\end{align*}
that satisfies $u_i=E(\widehat{u}_i\mid \tau)$ and $U_n=n^{-1}\sum_{i=1}^nu_i$ is the corresponding underlying U-statistic. Similar definitions apply for $v_i$ and $V_n$.

We begin the proof by stating the following decomposition, similar in spirit to the one in \cite{zhang2022edgeworth}. We decompose 
$ \widehat{\Sigma}^\flat_{u,v}-\Sigma_{u,v}$ as 
\begin{align*}
\widehat{\Sigma}^\flat_{u,v}-\Sigma_{u,v}=\underbrace{\widehat{\Sigma}^\flat_{u,v}-\widehat{\Sigma}_{u,v}}_{\widehat{\Delta}_{n,u,v}}+\underbrace{\widehat{\Sigma}_{u,v}-\Sigma_{u,v}}_{\Delta_{n,u,v}},
\end{align*}
where $\widehat{\Sigma}_{u,v} = (r_ur_v/n)\sum_{i=1}^n (u_i - U_n)(v_i - V_n)$.  
We will leave the analysis for $\Delta_{n,u,v}$ in Lemma \ref{appB-lem:bootstrap-lemma-4} and here we mainly handle $\widehat{\Delta}_{n,u,v} \to 0$ in probability.

For convenience, it is equivalent to show $\widehat{\Delta}_{n,u,v}/(r_ur_v) \to 0$ in probability, normalized by the rank.  We observe that 
\begin{align*}
\frac{\widehat{\Sigma}^\flat_{u,v}}{r_u r_v} = 
\underbrace{
\frac{1}{n} \sum_{i=1}^n (u_i - U_n)(v_i - V_n)
}_{\substack{ \widehat{\Sigma}_{u,v}/(r_u r_v)}}
+ 
\delta_{n,u,v},  \quad \delta_{n,u,v}:=\frac{\widehat{\Delta}_{n,u,v}}{r_ur_v},
\end{align*}
where the residual $\delta_{n,u,v}$ further expands as 
\begin{align}\label{appB:boot_var_noise_rep}
\delta_{n,u,v}&=\underbrace{\frac{1}{n} \sum_{i=1}^n (\widehat{u}_i - u_i)(v_i - V_n)}_{\textrm{A}}
+ \underbrace{\frac{1}{n} \sum_{i=1}^n (u_i - U_n)(\widehat{v}_i - v_i)}_{\textrm{B}} \\
&\qquad + \underbrace{\frac{1}{n} \sum_{i=1}^n (\widehat{u}_i - u_i)(\widehat{v}_i - v_i)}_{\textrm{C}} 
+ \underbrace{(\widehat{U}_n - U_n)(\widehat{V}_n - V_n)}_{O_P(n^{-1})} \nonumber.
\end{align}

By Lemma \ref{appB:lemma2} and Lemma \ref{appB:lemma3}, it follows that  $(\widehat{U}_n - U_n)=O_P(n^{-1/2})$ and $(\widehat{V}_n - V_n)=O_P(n^{-1/2})$ under Condition \ref{appB:sparsity-cond-weak},  after appropriately accounting  for the plug-in error $(\rho_n/\hat{\rho}_n-1)$. Therefore, the last term of \eqref{appB:boot_var_noise_rep} vanishes as $O_P(n^{-1})$.

We handle the three terms individually. Since the first and second terms are of the same type, it suffices to analyze term \(A\). Accounting for the plug-in error from the estimated sparsity parameter, \(A\) decomposes as  
\begin{align}\label{appB:first_term_cross}
\frac{1}{n}\sum_{i=1}^n (\widehat{u}_i - u_i)(v_i - V_n)
= \frac{1}{n}\sum_{i=1}^n (\widehat{u}_i - \widecheck{u}_i)(v_i - V_n)
+ \underbrace{\frac{1}{n}\sum_{i=1}^n (\widecheck{u}_i - u_i)(v_i - V_n)}_{\text{III}}.
\end{align}
Here, \(\widecheck{u}_i\) is defined analogously to \eqref{appB:u_i_hat} but with the true normalization factor \(\rho_n\):  
\begin{align*}
\widecheck{u}_i
:= \frac{1}{\binom{n-1}{r_u-1}\rho_n^{s_u}}
   \sum_{\substack{\pi_i \subset [n]\setminus\{i\}\\ |\pi_i|=r_u-1}}
   h_u(A_{\{i\}\cup\pi_i};\, X_{\{i\}\cup\pi_i}; Y_{\{i\}\cup\pi_i}).
\end{align*}
The first term on the right-hand side of \eqref{appB:first_term_cross} further decomposes as 
\begin{align*}
\frac{1}{n}\sum_{i=1}^n (\widehat{u}_i - \widecheck{u}_i)(v_i - V_n)
&= \underbrace{\Biggl(\Biggl(\frac{\rho_n}{\hat\rho_n}\Biggr)^{s_u} - 1\Biggr)
     \left[\frac{1}{n}\sum_{i=1}^n (\widecheck{u}_i - u_i)(v_i - V_n)\right]}_{\text{I}} \\
&\quad + \underbrace{\Biggl(\Biggl(\frac{\rho_n}{\hat\rho_n}\Biggr)^{s_u} - 1\Biggr)
     \left[\frac{1}{n}\sum_{i=1}^n u_i(v_i - V_n)\right]}_{\text{II}}.
\end{align*}

We begin with term III in \eqref{appB:first_term_cross}, which can be expanded as   
\[
n^{-1}\sum_{i=1}^n v_i(\widecheck{u}_i - u_i) 
\;-\; V_n \cdot n^{-1}\sum_{i=1}^n (\widecheck{u}_i - u_i).
\]

Since $V_n = O_P(1)$ under the assumption $E\!\left[\,h_v^2(\cdot)\,\right] < \infty$, 
which is implied by $E(Y^4)<\infty$, $E\|X\|^4<\infty$, 
and since the network noise 
$n^{-1}\sum_{i=1}^n (\widecheck{u}_i - u_i) = o_P(n^{-1/2})$ 
by Lemma~\ref{appB:lemma2} and Lemma~\ref{appB:lemma3} under the assumed conditions, 
we obtain
\[
V_n \cdot n^{-1}\sum_{i=1}^n (\widecheck{u}_i - u_i) = o_P(n^{-1/2}).
\]
Furthermore, applying the Cauchy--Schwarz inequality gives
\begin{align}\label{appB-eq:bootstrap_linear_noise_product}
E\left|\,n^{-1}\sum_{i=1}^n v_i(\widecheck{u}_i - u_i)\right|
  &\le n^{-1}\sum_{i=1}^n \sqrt{E[v_i^2] \, E[(\widecheck{u}_i - u_i)^2]}.
\end{align}
Here $E[v_i^2]<\infty$ under the condition $E\!\left[\,h_v^2(\cdot)\,\right] < \infty$, implied by $E(Y^4)<\infty$, $E\|X\|^4<\infty$, and $w(x,y)\leq C$, while $E[(\widecheck{u}_i - u_i)^2]=o(1)$ by Lemma~\ref{appB-lem:bootstrap-lemma-1}. 
Thus, by Markov's inequality, the bound in 
\eqref{appB-eq:bootstrap_linear_noise_product} is $o(1)$. 
Combining both components, III is $o_P(1)$. 
Finally, note that I is of lower order than III, as it involves the plug-in 
error $O_P(n^{-1/2})$ multiplied by III, and is therefore $o_P(n^{-1/2})$.

For II, it suffices to show that
\[
n^{-1}\sum_{i=1}^n u_i(v_i - V_n) 
= n^{-1}\sum_{i=1}^n u_i v_i - U_n V_n
\]
is $O_P(1)$. Under the assumptions $E(u_i^2)<\infty$ and $E(v_i^2)<\infty$, by Cauchy--Schwarz and Markov,
\[
n^{-1}\sum_{i=1}^n u_i v_i = O_P(1),
\]
and by the convergence of $U$-statistics, $U_nV_n = O_P(1)$. Therefore, the entire expression multiplied by $O_P(n^{-1/2})$ is $o_P(1)$.

Finally, for analyzing C, we separate the plug-in effect and expand it into four terms as:
\begin{align}\label{appB:term-C-expand}
\frac{1}{n}\sum_{i=1}^n (\widehat{u}_i - u_i)(\widehat{v}_i - v_i)
&=
\Biggl(\Bigl(\tfrac{\rho_n}{\widehat{\rho}_n}\Bigr)^{s_u} - 1\Biggr)
\Biggl(\Bigl(\tfrac{\rho_n}{\widehat{\rho}_n}\Bigr)^{s_v} - 1\Biggr)
\left(\frac{1}{n}\sum_{i=1}^n \widecheck{u}_i \widecheck{v}_i \right) \\
&\quad +
\Biggl(\Bigl(\tfrac{\rho_n}{\widehat{\rho}_n}\Bigr)^{s_u} - 1\Biggr)
\left(\frac{1}{n}\sum_{i=1}^n \widecheck{u}_i(\widecheck{v}_i - v_i) \right) \notag \\
&\quad +
\Biggl(\Bigl(\tfrac{\rho_n}{\widehat{\rho}_n}\Bigr)^{s_v} - 1\Biggr)
\left(\frac{1}{n}\sum_{i=1}^n (\widecheck{u}_i - u_i)\widecheck{v}_i \right) \notag \\
&\quad +
\frac{1}{n}\sum_{i=1}^n (\widecheck{u}_i - u_i)(\widecheck{v}_i - v_i). \notag
\end{align}

Furthermore, by Cauchy--Schwarz, Markov’s inequality, and Lemma~\ref{appB-lem:bootstrap-lemma-1}, 
the last term of \eqref{appB:term-C-expand} is \( o_P(1) \), since 
\begin{align}\label{appB-eq:bootstrap_product_noise}
E \left| n^{-1} \sum_{i=1}^n (\widecheck{u}_i - u_i)(\widecheck{v}_i - v_i) \right|
\leq \sqrt{ E(\widecheck{u}_i - u_i)^2 } \cdot \sqrt{ E(\widecheck{v}_i - v_i)^2 }
= o(1).    
\end{align}

For the first term of \eqref{appB:term-C-expand}, Proposition~\ref{main:prop1} 
and Theorem~\ref{theorem 1} give
\[
n^{-1}\sum_{i=1}^n \widecheck{u}_i \widecheck{v}_i = O_P(1).
\]
Since it is multiplied by two vanishing factors, the term is \( o_P(1) \).

Next, note that the second and third terms in \eqref{appB:term-C-expand} are symmetric. We can decompose
\[
n^{-1}\sum_{i=1}^n \widecheck{u}_i (\widecheck{v}_i - v_i)
= n^{-1}\sum_{i=1}^n (\widecheck{u}_i - u_i)(\widecheck{v}_i - v_i)
+ n^{-1}\sum_{i=1}^n u_i (\widecheck{v}_i - v_i).
\]
By \eqref{appB-eq:bootstrap_product_noise} and 
\eqref{appB-eq:bootstrap_linear_noise_product}, both terms are \(o_P(1)\).  
With the additional vanishing plug-in factor, this shows that
\[
\left( \left(\frac{\rho_n}{\widehat{\rho}_n}\right)^{s_u} - 1 \right)
\cdot \left( \frac{1}{n} \sum_{i=1}^n \widecheck{u}_i (\widecheck{v}_i - v_i) \right) = o_P(n^{-1/2}),
\]
which concludes the argument.
\end{proof}

The following lemma establishes properties of the local $U$-statistics in the latent space. Our strategy is closely related to that used in the proof of Lemma~3.1 in \cite{zhang2022edgeworth}, but since it pertains to covariances rather than variances, we include a proof here for completeness.

\begin{lemma}\label{appB-lem:bootstrap-lemma-4} 
Under the conditions $E(Y^4)<\infty$, $E\|X\|^4<\infty$, 
the following convergence result holds:
\begin{align*}
\widehat{\Sigma}_{u,v} \to \Sigma_{u,v}  \quad \text{in probability}.
\end{align*}
\end{lemma}

\begin{proof}
We will show that:
\begin{align*}
\frac{\widehat{\Sigma}_{u,v}}{r_u r_v} 
= \frac{1}{n} \sum_{i=1}^n (u_i - U_n)(v_i - V_n)
\to
\operatorname{cov}(h_{1,u}(\tau_i), h_{1,v}(\tau_i)) \quad \textrm{in probability},
\end{align*}
where the term on the LHS decomposes as
\begin{align}\label{appB:latent-uv-expand}
\frac{1}{n} \sum_{i=1}^n (u_i - U_n)(v_i - V_n) 
&= \frac{1}{n} \sum_{i=1}^n \big( (u_i - \mu_U) + (\mu_U - U_n) \big) \big( (v_i - \mu_V) + (\mu_V - V_n) \big) \nonumber\\
&=\frac{1}{n} \sum_{i=1}^n (u_i - \mu_U)(v_i - \mu_V)+\underbrace{(\mu_U - U_n)(\mu_V - V_n)}_{O_P(n^{-1})}.
\end{align}
The two intermediate terms in \eqref{appB:latent-uv-expand} are zero, since 
$n^{-1}\sum_{i=1}^n u_i = U_n$ and $n^{-1}\sum_{i=1}^n v_i = V_n$. 
The last term is $O_P(n^{-1})$, because both $U_n$ and $V_n$ are subgraph $U$-statistics satisfying 
$U_n - \mu_U = O_P(n^{-1/2})$ and $V_n - \mu_V = O_P(n^{-1/2})$; 
see \cite{serfling2009approximation} for details, under the conditions $E[h_u^2(\cdot)] < \infty$ and $E[h_v^2(\cdot)] < \infty$, 
which are implied by $E(Y^4) < \infty$, $E\|X\|^4 < \infty$ and $w(x,y) \leq C$. 

Thus, it remains to analyze the first term of \eqref{appB:latent-uv-expand}, which can be further decomposed as
\begin{align*}
&\frac{1}{n} \sum_{i=1}^n (u_i - \mu_U)(v_i - \mu_V)\\
&= \frac{1}{n} \sum_{i=1}^n \left[(u_i - h_{1,u}(\tau_i)) + (h_{1,u}(\tau_i) - \mu_U)\right]
\left[(v_i - h_{1,v}(\tau_i)) + (h_{1,v}(\tau_i) - \mu_V)\right] \\
&= \underbrace{\frac{1}{n} \sum_{i=1}^n (u_i - h_{1,u}(\tau_i))(v_i - h_{1,v}(\tau_i))}_{\text{A}} 
+ \underbrace{\frac{1}{n} \sum_{i=1}^n (u_i - h_{1,u}(\tau_i))(h_{1,v}(\tau_i) - \mu_V)}_{\text{B}} \\
&\quad + \underbrace{\frac{1}{n} \sum_{i=1}^n (h_{1,u}(\tau_i) - \mu_U)(v_i - h_{1,v}(\tau_i))}_{\text{C}} + \frac{1}{n} \sum_{i=1}^n (h_{1,u}(\tau_i) - \mu_U)(h_{1,v}(\tau_i) - \mu_V).
\end{align*}
Note that, the weak law of large numbers implies that the last term converges in probability to 
$\operatorname{cov}(h_{1,u}(\tau_i), h_{1,v}(\tau_i))=E\left[(h_{1,u}(\tau_i)-\mu_U)(h_{1,v}(\tau_i)-\mu_V)\right]$. 
Furthermore, B and C can be viewed as special cases of A since  $h_{1,v}(\tau_i)$ is a local U--statistic of order one. Therefore, it only remains to bound A.

Following the Hoeffding decomposition, the cross–product in the local $U$–statistic, 
\((u_i - h_{1,u}(\tau_i))(v_i - h_{1,v}(\tau_i))\), can be expanded as
\begin{align*}
&\left\{\sum_{k=1}^{r_u-1}\frac{\binom{n-k-1}{r_u-k-1}}{\binom{n-1}{r_u-1}}\sum_{\substack{1 \leq j_1 \leq \ldots \leq j_k \leq n \\ j_1,\ldots,j_k \neq i}}g_{k+1,u}(\tau_i,\tau_{j_1},\ldots,\tau_{j_k})\right\}\notag\\
&\qquad \qquad \times
\left\{\sum_{k=1}^{r_v-1}\frac{\binom{n-k-1}{r_v-k-1}}{\binom{n-1}{r_v-1}}\sum_{\substack{1 \leq j_1 \leq \ldots \leq j_k \leq n \\ j_1,\ldots,j_k \neq i}}g_{k+1,v}(\tau_i,\tau_{j_1},\ldots,\tau_{j_k})\right\} \nonumber\\
&=\left\{\frac{r_u-1}{n-1}\sum_{\substack{1 \leq j \leq n \\ j \neq i}} g_{2,u}(\tau_i,\tau_j)+O_P(n^{-1})\right\} \times
\left\{\frac{r_v-1}{n-1}\sum_{\substack{1 \leq j \leq n \\ j \neq i}} g_{2,v}(\tau_i,\tau_j)+O_P(n^{-1})\right\}.
\end{align*}

Next, we verify that the leading component of A satisfies
\begin{align*}
\frac{1}{n} \sum_{i=1}^n 
\left\{\frac{r_u-1}{n-1}\sum_{\substack{1 \leq j \leq n \\ j \neq i}} g_{2,u}(\tau_i,\tau_j)\right\} 
\left\{\frac{r_v-1}{n-1}\sum_{\substack{1 \leq j \leq n \\ j \neq i}} g_{2,v}(\tau_i,\tau_j)\right\} 
= o_P(1).
\end{align*}
It follows by observing that
\begin{align*}
\operatorname{var}\!\left(\frac{r_u-1}{n-1}\sum_{j \neq i} g_{2,u}(\tau_i,\tau_j)\right) 
&= \left(\frac{r_u-1}{n-1}\right)^2 
   \operatorname{var}\!\left(\sum_{j \neq i} g_{2,u}(\tau_i,\tau_j)\right) \nonumber \\
&= O(n^{-1}),
\end{align*}
Since, by the Hoeffding canonical property, $E[g_{2,u}(\tau_i,\tau_j)\mid\tau_i]=0$ and 
the terms $g_{2,u}(\tau_i,\tau_j)$ are uncorrelated across $j$, 
this variance bound holds under the condition $E[g_{2,u}^2(\tau_1,\tau_2)]<\infty$, 
which is implied by the standing moment assumptions $E(Y^4)<\infty$, $E\|X\|^4<\infty$. 

Hence, by Cauchy-Schwarz, we obtain that  $A=O_P(n^{-1})$, and the claim follows.
\end{proof}

The next lemma establishes the order of $\operatorname{var}(\widecheck{u}_i - u_i)$ for the general local subgraph estimates, and establishes that it is asymptotically negligible. The proof follows a similar strategy as in the global case we had before in Lemma \ref{appB:lemma2}. 

\begin{lemma}[Node-wise Variance Bound]\label{appB-lem:bootstrap-lemma-1}
Under the sparsity condition as Condition \ref{appB:sparsity-cond-weak},
\begin{align*}
E[(\widecheck{u}_i - u_i)^2]=\operatorname{var}(\widecheck{u}_i - u_i) = o(1).    
\end{align*}
\end{lemma}

\begin{proof}
Using the law of total variance, conditioning on the $\tau=(\tau_1,\ldots,\tau_n)$, $\tau_i=(X_i,Y_i,\xi_i)$, we get
$$\operatorname{var}(\widecheck{u}_i - u_i) = E[\operatorname{var}(\widecheck{u}_i - u_i | \tau)] + \operatorname{var}(E[\widecheck{u}_i - u_i | \tau]).$$
Since $E[\widecheck{u}_i - u_i | \tau] = 0$, the proof reduces to analyzing the conditional variance term. Note that we can expand the conditional variance as
\[
\operatorname{var}(\widecheck{u}_i \mid \tau)
= \frac{1}{\binom{n-1}{r_u-1}^2\,\rho_n^{2 s_u}}
\sum_{\substack{\pi,\pi'\subset [n]\setminus\{i\}\\ |\pi|=|\pi'|= r_u-1}}
\operatorname{cov}\left(
  h_u\big(A_{\{i\}\cup \pi}\big),
  h_u\big(A_{\{i\}\cup \pi'}\big)
  \,\middle|\, \tau
\right).
\]

Following a similar argument as in Lemma~\ref{appB:lemma2}, 
we analyze the double summation by grouping pairs of subgraphs 
by the number of shared nodes \(c\) (including node \(i\)) and shared edges \(d\). 
For a term that shares \(d\) edges, the individual covariance is upper-bounded by
\begin{align}\label{appB:boot-single_cov_bound}
\operatorname{cov}\!\left(
  h_u(A_{\{i\} \cup \pi_i}),\,
  h_u(A_{\{i\} \cup \pi'_i}) \,\middle|\, \tau
\right) 
\leq \rho_n^{2s_u-d}\, E h_{\mathcal{M}_{c,d}}(\tau).  
\end{align}

There are $O\!\left(n^{2(r_u-1)-(c-1)}\right)$ such terms, each with covariance bounded by \eqref{appB:boot-single_cov_bound}. 
Therefore, combining these bounds and assuming $E h_{\mathcal{M}_{c,d}}(\cdot)<\infty$, which holds under our standing conditions $E(Y^4)<\infty$, $E\|X\|^4<\infty$ and $w(x,y)\leq C$, we obtain
\[
\operatorname{var}(\widecheck{u}_i - u_i) 
= E\big[\operatorname{var}(\widecheck{u}_i \mid \tau)\big] 
= \sum_{c=1}^{r_u} \sum_{d} 
O\!\left( \frac{n^{2(r_u-1)-(c-1)}}{n^{2(r_u-1)} \rho_n^{2s_u}} \cdot \rho_n^{2s_u-d} \right) 
= \sum_{c=1}^{r_u} \sum_{d} O\!\left(n^{1-c}\rho_n^{-d}\right).
\]
The sparsity condition in Condition~\ref{appB:sparsity-cond-weak} guarantees that 
$O\!\left(n^{1-c}\rho_n^{-d}\right)=o(1)$; 
hence, 
\[
\operatorname{var}(\widecheck{u}_i - u_i)=o(1).
\]
The claim follows.
\end{proof}

The following lemma establishes appropriate smoothness of a functional related to the OLS estimator.  

\begin{lemma}[Property of the OLS Functional]\label{appB:OLS_functional_property}
Let $(\mathrm{M}_0, b_0) \in \mathbb{R}^{m \times m} \times \mathbb{R}^m$ be a point where $\mathrm{M}_0$ is invertible and symmetric.
Then there exists an open neighborhood of $(\mathrm{M}_0,b_0)$ on which $f(\mathrm{M}, b) = \mathrm{M}^{-1} b$ is continuously differentiable. 
\end{lemma}

\begin{proof}
Since $\mathrm{M}_0$ is invertible and symmetric, its minimum singular value satisfies $\sigma_m(\mathrm{M}_0)= \v >0$.  By Weyl's inequality for singular values, it then follows that, for any $\|\mathrm{M} - \mathrm{M}_0\|_{\mathrm{op}} < \v$, $\sigma_{min}(\mathrm{M}) > 0$; therefore, any $\mathrm{M}$ in this neighborhood is invertible.  By equivalence of norms on finite-dimensional spaces, for any choice of norm, there exists an open neighborhood $\mathcal{N}$ of $\mathrm{M}_0$ such that $\mathrm{M}$ is invertible for all $\mathrm{M} \in \mathcal{N}$.      

The total differential of the functional $f(\mathrm{M}, b) = \mathrm{M}^{-1} b$ is given by:
\[
d f = d(\mathrm{M}^{-1} b) = d\mathrm{M}^{-1} \cdot b + \mathrm{M}^{-1} \cdot db,
\]
see for example, \cite{magnus1988matrix}. By a matrix calculus identity, we have $d\mathrm{M}^{-1} = - \mathrm{M}^{-1} (d\mathrm{M}) \mathrm{M}^{-1}$. Substituting this above, we have:
\[
df = - \mathrm{M}^{-1} (d\mathrm{M}) \mathrm{M}^{-1} b + \mathrm{M}^{-1} db.
\]
From the identity $\mathrm{M}^{-1} = \frac{1}{\mathrm{det}(\mathrm{M})} \mathrm{adj}(\mathrm{M})$, it follows that the map $\mathrm{M} \mapsto \mathrm{M}^{-1}$ is continuous on $\mathcal{N}$.  Thus, the claim follows.  
\end{proof}

\subsection{Pre-computation Example for Linear Multiplier Bootstrap Procedure}\label{bootstrap-precom}
We now present a concrete example of how to compute the quantity $\widehat{g}_1$.

In Algorithm \ref{alg:ALG-pre-computation} below, we illustrate how to compute \(n^{-1} \sum_{i=1}^n Y_i \widehat{Z}_{ij}\). In particular, we consider a rooted star as an example, where fast computation can be achieved through appropriate matrix multiplication. According to the global representation of \(n^{-1} \sum_{i=1}^n Y_i \widehat{Z}_{ij}\) in \eqref{appB-eq:h_Y_rooted}, when \(\widehat{Z}_{ij}\) is a rooted star, the corresponding \(\widehat{g}_1(i)\), following the definition in \ref{appB-eq:sample-projection-unit}, is:
\begin{align*}
\widehat{g}_1(i)
&= \frac{1}{\binom{n-1}{r_j-1}\,\hat{\rho}_n^{s_j}}
\sum_{\substack{I \subset [n]\setminus\{i\}\\ |I|=r_j-1}}\frac{1}{r_j}
\left(
  Y_i \prod_{j\in I} A_{ij}
  \;+\;
  \sum_{l\in I} Y_l\, A_{li}\!\!\prod_{j\in I\setminus\{l\}} \! A_{lj}
\right)
- \frac{1}{n}\sum_{i=1}^{n} Y_i \widehat{Z}_{ij}.
\end{align*}

\begin{algorithm}[H]
\caption{Pre-computation of $\widehat{g}_1$ for $n^{-1} \sum_{i=1}^{n} Y_i \widehat{Z}_{ij}$ in the Rooted-star Case}
\label{alg:ALG-pre-computation}
\KwIn{Adjacency matrix $\mathrm{A} \in \{0,1\}^{n\times n}$, response vector $Y \in \mathbb{R}^n$, star-type subgraph $\mathcal{R}_j$ with node-edge parameter as $(r_j,s_j)$}
\KwOut{Vector $\big(\widehat{g}_1(i)\big)_{i=1}^n$}
Estimate \(\hat{\rho}_n\) by \(\frac{2}{n(n-1)} \sum_{1 \le i < j \le n} A_{ij}\), \quad
$d_i \gets \sum_{j=1}^n A_{ij}$ for $i=1,\dots,n$\;
$A^Y \gets A \cdot \mathrm{diag}(Y)$ \tcp*{Pre-compute response-weighted adjacency}

$\widehat{Z}_{ij} \gets \frac{\binom{d_i }{r_j-1}}{\hat{\rho}_n^{s_j}\binom{n - 1}{r_j-1}}$ \tcp*{Pre-compute rooted star frequency}

\ForEach{$i \in \{1,\dots,n\}$}{
    $C^{\mathrm{rooted}}_{y}(i) \gets Y_i\binom{d_i}{r_j-1}/r_j$ \tcp*{Rooted star count at $i$}
    
    $C^{\mathrm{peripheral}}_{y}(i) \gets \sum_{j : A_{ij}=1} Y_j\binom{d_j - 1}{r_j-2}/r_j$ \tcp*{Star count from neighbors}

    $C_y(i) \gets C^{\mathrm{rooted}}_{y}(i) + C^{\mathrm{peripheral}}_{y}(i)$\;
    
    $\widehat{g}_1(i) \gets \dfrac{C_y(i)}{\hat{\rho}_n^{s_j}\binom{n-1}{r_j-1}} - n^{-1} \sum_{i=1}^{n} Y_i \widehat{Z}_{ij}$\;
}

\Return $\big(\widehat{g}_1(i)\big)_{i=1}^n$
\end{algorithm}
\begin{remark}
Pre-computation is required because the local quantities 
$Y_i\widehat{Z}_{ij}$ and $X_i\widehat{Z}_{ij}$ do not coincide with 
$\widehat{g}_1(i)$ unless the weights are properly distributed. 
The same issue arises for $\widehat{Z}_{ij}$ in the rooted-star case. 
In the quadratic setting, where only the leading term is retained, 
such pre-computations are also naturally necessary.
\end{remark}

\section{Proofs for Section \ref{subsec:rdpg-reg}}\label{App-C}
We start by introducing some notations. Let $\mathrm{P}$ be the $n \times n$ matrix for which $P_{ij} = \operatorname{pr}(A_{ij}=1 \ | \ \xi_i, \xi_j)$. By construction, the $n \times d$ matrix $Z$, where the $i$th row is $Z_i$, satisfies
$(\mathrm{P}^{\T}\mathrm{P})^{1/2}=\rho_nZZ^{\T}$. Thus, if $\mathrm{P}$ admits the spectral decomposition $\mathrm{P}=\mathrm{U}_{\mathrm{P}}\mathrm{S}_{\mathrm{P}}\mathrm{U}_{\mathrm{P}}^{\T}$, then 
\begin{align*}
\rho_nZZ^{\T}=\mathrm{U}_{\mathrm{P}}|\mathrm{S}_{\mathrm{P}}|\mathrm{U}_{\mathrm{P}}^{\T}.   
\end{align*}

Our choice of target $Z_i$ is closely related to the term $\widetilde{Z}_i$ discussed briefly in
\cite{rubin2022statistical}. As the authors point out, with this choice of target, the underlying latent positions are estimable up to orthogonal rotation rather than indefinite orthogonal rotation.  Since their main theorems involve a different target and consequently involve indefinite orthogonal rotations, we restate some of their results for our choice of target parameter below. In what follows, let $\widecheck{Z}$ be a $n \times d$ matrix where the $i$th row is the ASE scaled by $(\hat{\rho}_n/\rho_n)^{1/2}$.

\begin{lemma}[Representation of the Embedding Error under Orthogonal Rotation]\label{app-rdpg:ase-error representation}

There exists a sequence of orthogonal matrices $\mathrm{Q}_n$ such that
\begin{align}\label{app-grdpg:scaled-error}
\widecheck{Z}\,\mathrm{Q}_n - Z
= \rho_n^{-1} (\mathrm{A} - \mathrm{P})\,Z\,(Z^\T Z)^{-1} + \rho_n^{-1/2}\mathrm{R}\,\mathrm{Q}_n,
\end{align}
where:
\begin{align}
\label{eq-2-infinity-bounds}
\left\| \rho_n^{-1} (\mathrm{A} - \mathrm{P})\,Z\,(Z^\T Z)^{-1} \right\|_{\ 2 \rightarrow \infty}& = O_{p} \left( \frac{\log^c n}{(n\rho_n)^{1/2}} \right),\quad
\left\| \rho_n^{-1/2}\mathrm{R}\mathrm{Q}_n  \right\|_{ \ 2 \rightarrow \infty} = O_{p} \left( \frac{\log^{2c} n}{n\rho_n} \right).
\end{align}
\end{lemma}

We are now ready to prove the CLT result in Theorem \ref{theorem 5}. 

\begin{proof}[Proof of Theorem \ref{theorem 5}.]
Define the following quantities:
\begin{align*}
\widecheck{J}^{M}_n &=
\left[
\operatorname{vec}\left( \mathrm{M}^\T_n \left( n^{-1} \widecheck{\mathrm{L}}^{\T} \widecheck{\mathrm{L}} \right) \mathrm{M}_n \right)^\T,
\quad
\left( \mathrm{M}_n^{\T} \left( n^{-1} \widecheck{\mathrm{L}}^{\T} Y\right)\right)^{\T},
\quad
\hat{\rho}_n/\rho_n
\right]^{\T}, \\
J_n &=\left[
\operatorname{vec}\left(n^{-1}\sum_{i=1}^n L_i L_i^\T\right)^\T,
\quad
\left(n^{-1} \sum_{i=1}^n L_i Y_i\right)^\T,
\quad
1+ \frac{2}{n} \sum_{i=1}^n \left( E[w(\xi_i, \xi_j) \mid \xi_i] - 1\right)
\right]^{\T},
\end{align*}
where
\[
L_i = 
\begin{bmatrix}
X_i \\
Z_i
\end{bmatrix} \in \mathbb{R}^{p+d}, 
\quad 
\mathrm{M}_n\widecheck{L}_i = 
\begin{bmatrix}
X_i \\
\mathrm{Q}_n \widecheck{Z}_i
\end{bmatrix} \in \mathbb{R}^{p+d}.
\]
And again, let $\eta^*=E(J_n)$. Consider the following decomposition:
\begin{align}\label{appB-grdpg:clt-decomp}
n^{1/2}( \widecheck{J}^{M}_n - \eta^*)
= n^{1/2}(\widecheck{J}^{M}_n - J_n) + n^{1/2}(J_n - \eta^*). 
\end{align}

To prove $\|n^{1/2}(\widecheck{J}^{M}_n - J_n)\|=o_P(1)$, it suffices to show:
\begin{align}
\label{eq-conv-prob-grdpg}
& n^{-1/2}\sum_{i=1}^n (\mathrm{Q}_n \widecheck{Z}_i - Z_i) X_i^\T \to 0,\quad   
n^{-1/2} \sum_{i=1}^n (\mathrm{Q}_n \widecheck{Z}_i - Z_i) Y_i \to 0,  \notag \\
&\quad n^{-1/2} \sum_{i=1}^n (\mathrm{Q}_n \widecheck{Z}_i \widecheck{Z}_i^\T \mathrm{Q}_n^\T - Z_i Z_i^\T) \to 0, \quad \textrm{in probability}.
\end{align}
We start 
$n^{-1/2}\sum_{i=1}^n (\mathrm{Q}_n \widecheck{Z}_i - Z_i) X_i^\T \to 0$ in probability; the second statement can be proven using analogous reasoning. Substituting \eqref{app-grdpg:scaled-error} from Lemma~\ref{app-rdpg:ase-error representation} into the average, we obtain:  
\begin{align*}
n^{-1/2} \sum_{i=1}^n (\mathrm{Q}_n \widecheck{Z}_i - Z_i) X_i^\T
&= \rho_n^{-1}n^{-1/2} (Z^\T Z)^{-1} \sum_{i=1}^n \left( (\mathrm{A} - \mathrm{P}) Z \right)_i X_i^\T 
+ \rho_n^{-1/2} n^{-1/2} \sum_{i=1}^n \mathrm{Q}_n R_i X_i^\T \notag \\
&=: T_1 + T_2.
\end{align*}

For \(T_2\), we apply the triangle inequality for the Frobenius norm and use the fact that the Frobenius norm of an outer product equals the product of the Euclidean norms of the two vectors. Specifically,  
\begin{align*}
\left\| T_2 \right\|_{\mathrm{F}}
= \left\| \rho_n^{-1/2} n^{1/2} \cdot \frac{1}{n} \sum_{i=1}^n \mathrm{Q}_n R_i X_i^\T \right\|_{\mathrm{F}} 
&\leq n^{1/2} \cdot \frac{1}{n} \sum_{i=1}^n \left\| \rho_n^{-1/2} \mathrm{Q}_nR_i X_i^\T \right\|_{\mathrm{F}} \\
&= n^{1/2} \cdot \frac{1}{n} \sum_{i=1}^n \| \rho_n^{-1/2} \mathrm{Q}_nR_i \|_2 \cdot \| X_i \|_2 \\
&\leq \max_{1 \leq i \leq n} \| \rho_n^{-1/2} \mathrm{Q}_nR_i \|_2 \cdot n^{1/2} \cdot \frac{1}{n} \sum_{i=1}^n \| X_i \|_2 \\
&= O_P\!\left( \frac{\log^{2c} n}{\rho_n n^{1/2}} \right)=  o_P(1), 
\end{align*}
where the last line follows from the  Lemma~\ref{app-rdpg:ase-error representation} and the assumption that \(E\|X_i\|_2^2 < \infty\), which implies \(n^{-1} \sum_{i=1}^n \| X_i \|_2 = O_P(1)\). Under the sparsity condition \(\rho_n = \omega( \log^{2c} n/n^{1/2})\), this term vanishes in probability.

For $T_1$, we have:
\begin{align*}
T_1 
&= \rho_n^{-1}n^{-1/2} (Z^\T Z)^{-1} \sum_{i=1}^n \left( (\mathrm{A} - \mathrm{P}) Z \right)_i X_i^\T \nonumber\\
&=\underbrace{(n\rho_n)^{-1/2}}_{I} \underbrace{(n^{-1}Z^\T Z)^{-1}}_{II} \underbrace{n^{-1}\sum_{i=1}^n \sum_{j=1}^n \rho_n^{-1/2}(A_{ij} - P_{ij}) Z_j X_i^\T}_{III}.
\end{align*}

First, under assumed sparsity condition, $I = o(1)$. Second, since the latent positions \( Z_i \) are i.i.d., the weak law of large numbers implies $n^{-1}Z^\T Z=n^{-1}\sum_{i=1}^nZ_i Z_i^\T \to E(Z_1Z_1^{\T})$ in probability, and hence the inverse, $(n^{-1}Z^\T Z)^{-1} $, converges in probability to \( \left[ E(Z_1Z_1^{\T})\right]^{-1} \). It remains to control $III$.  

Conditioned on \( \{X_i\}_{1 \leq i \leq n} \) and \( \{\xi_j\}_{1 \leq j \leq n} \), \( (A_{ij} - P_{ij}) Z_j X_i^\T \in \mathbb{R}^{d \times p} \) are independent, mean-zero random matrices. Therefore, it follows that:
\begin{align*}
\Pr\bigl(\|III\|_F^2>\v\bigr) &\leq \frac{1}{n^2 \v^2 \rho_n}\sum_{i=1}^n  \sum_{j=1}^n E[\operatorname{var}(A_{ij} \ | \ \xi) \|Z_j \|^2 \|X_i\|^2] \\ 
&\leq C \v^{-2}E\|Z\|^4 E\|X\|^4, 
\end{align*}
where last inequality follows from the Cauchy-Schwarz inequality and the assumption that $w(x, y) \leq C $; stronger moment conditions will also yield the same conclusion for appropriate unbounded graphons.

We will now establish the third claim in \eqref{eq-conv-prob-grdpg}. Notice that:
\begin{align*}
Q_n\widecheck Z_i\widecheck Z_i^\T Q_n^\T -Z_iZ_i^\T
=(Q_n\widecheck Z_i -Z_i)(Q_n\widecheck Z_i -Z_i)^\T
+(Q_n\widecheck Z_i -Z_i)Z_i^\T
+Z_i(Q_n\widecheck Z_i -Z_i)^\T.    
\end{align*}
Therefore, applying the Frobenius norm to the average, we have:
\begin{align}
\label{eq-quadratic-rdgp-bound}
\Bigl\|\,n^{-1/2}\sum_{i=1}^n\bigl(Q_n\widecheck Z_i\widecheck Z_i^\T Q_n^\T -Z_iZ_i^\T\bigr)\Bigr\|_F
&\le 
\,n^{-1/2}\sum_{i=1}^n\|Q_n\widecheck Z_i -Z_i\|^2
+2\;n^{-1/2}\sum_{i=1}^n\|Q_n\widecheck Z_i -Z_i\|\,\|Z_i\|.
\end{align}

By \eqref{eq-2-infinity-bounds} in Lemma \ref{app-rdpg:ase-error representation},
\[
n^{-1/2}\sum_{i=1}^n\|Q_n\widecheck Z_i -Z_i\|^2
\le
n^{1/2}\,\|Q_n\widecheck Z -Z\|_{2\to\infty}^2
=O_P\!\left(\frac{\log^{2c}n}{\rho_nn^{1/2}}\right).
\]

The argument for the last term in \eqref{eq-quadratic-rdgp-bound}  
follows directly from the same analysis used to bound the linear terms.  Hence, the third claim is shown.  Now, it is clear that $n^{1/2}(J_n - \eta^*)$ in \eqref{appB-grdpg:clt-decomp} converges to a multivariate normal distribution.

Furthermore, observe that:
\begin{align*}
 (\mathrm{M}_n^\T \widehat{\mathrm{L}}^\T \widehat{\mathrm{L}} \mathrm{M}_n)^{-1}(\mathrm{M}_n^\T \widehat{\mathrm{L}}^\T Y)=\mathrm{M}_n\left[(\widehat{\mathrm{L}}^\T \widehat{\mathrm{L}})^{-1}\widehat{\mathrm{L}}^\T Y\right]=\mathrm{M}_n\widehat{\beta}. 
\end{align*}
Thus, by the delta method,
$n^{1/2}  (\mathrm{M}_n\widehat{\beta} - \beta^*)$ is asymptotically normal. 
The claim follows.
\end{proof}

\subsection{Proofs for Section \ref{bootstrap:rdpg}}
Before beginning the proof of Theorem~\ref{rdpg:theorem6}, we introduce a few quantities used in the proof. Define
\begin{align*}
\widehat{J}^{\flat}_n &:=
\left[
\operatorname{vec}\!\left( \mathrm{M}_n^{\T} \left( n^{-1} \widehat{\mathrm{L}}^{\T} \widehat{\mathrm{L}} \right)^{\flat} \mathrm{M}_n \right)^{\T},\;
\left( \mathrm{M}_n^{\T} \left( n^{-1} \widehat{\mathrm{L}}^{\T} Y \right)^{\flat} \right)^{\T},\;
\hat{\rho}^\flat_n/\hat{\rho}_n
\right]^{\T}.
\end{align*}

Following the bootstrap procedure in Section~\ref{bootstrap:rdpg}, its bootstrap covariance is given by
\begin{align*}
\operatorname{cov}^{\flat}(\widehat{J}_n^{\flat}) 
&= \frac{1}{ n^2 } \sum_{i=1}^{n}
 \begin{bmatrix}
 \widehat{A}_i \\
 \widehat{B}_i \\
 \widehat{C}_i
 \end{bmatrix}
 \begin{bmatrix}
 \widehat{A}_i^{\T} & \widehat{B}_i^{\T} & \widehat{C}_i
 \end{bmatrix},
\end{align*}
where
\begin{align*}
\widehat{A}_i &= \operatorname{vec}\!\left( \mathrm{M}_n \left( \widehat{L}_i \widehat{L}_i^\T - \tfrac{1}{n} \sum_{j=1}^n \widehat{L}_j \widehat{L}_j^\T \right) \mathrm{M}_n^\T \right), \quad
\widehat{B}_i =  \mathrm{M}_n \left( \widehat{L}_i Y_i - \tfrac{1}{n} \sum_{j=1}^n \widehat{L}_j Y_j \right), \\
\widehat{C}_i &= 2\left( \frac{1}{(n - 1)\hat{\rho}_n} \sum_{j \ne i} A_{ij} - 1 \right).
\end{align*}
Finally, let
\[
\widehat{\Sigma}^\flat := n\operatorname{cov}^{\flat}(\widehat{J}_n^{\flat}).
\] 
A key step in establishing the desired bootstrap consistency is establishing the consistency of $\widehat{\Sigma}^\flat$, which is proven below.  

\begin{lemma}\label{lem-grdpg-boot-variance-consistency}
Under the conditions of Theorem \ref{rdpg:theorem6},

\begin{align*}
\widehat{\Sigma}^\flat \to \Sigma_\Psi \quad \text{in probability}.
\end{align*}  
\end{lemma}

\begin{proof}
Define
\begin{align*}
A_i = \operatorname{vec}\left( L_i L_i^\T - \frac{1}{n} \sum_{j=1}^n L_j L_j^\T \right), \quad
B_i =  L_i Y_i - \frac{1}{n} \sum_{j=1}^n L_j Y_j,
\quad
C_i = 2 \left( E[w(\xi_i, \xi_j) \mid \xi_i] - 1\right).
\end{align*}
And let 
\begin{align*}
\widehat{\Sigma}:=\frac{1}{n} \sum_{i=1}^n
\begin{bmatrix}
A_i \\
B_i \\
C_i
\end{bmatrix}
\begin{bmatrix}
A_i^\T & B_i^\T & C_i
\end{bmatrix}.       
\end{align*} 

Since $(A_i, B_i, C_i)$ are i.i.d., we have $\widehat{\Sigma} \to \Sigma_\Psi$ in probability by the consistency of empirical covariances. It only remains to establish the following result:
\begin{align*}
\widehat{\Sigma}^\flat - \widehat{\Sigma} \to 0 \quad \text{in probability}.
\end{align*}

The term $n^{-1} \sum_{i=1}^{n} (\widehat{C}^2_i-C_i^2) \to 0$ in probability by Lemma \ref{appB-lem:bootstrap-lemma-3} and Lemma \ref{appB-lem:bootstrap-lemma-4}.  Moreover, by Lemma \ref{appB-grdpg:lem-cross-noise}, 
\begin{align*}
\frac{1}{n} \sum_{i=1}^{n} \left( \widehat{A}_i \widehat{C}_i - A_i C_i\right) \to 0
\quad
\frac{1}{n} \sum_{i=1}^{n} \left( \widehat{B}_i \widehat{C}_i - B_i C_i\right) \to 0 \quad \textrm{in probability}.
\end{align*}
It remains to show:
\begin{align}\label{appB-grdpg:boot_var_noise_3terms}
&\frac{1}{n} \sum_{i=1}^{n} \left( \widehat{A}_i \widehat{A}_i^\T - A_i A_i^\T \right) \to 0, \quad
\frac{1}{n} \sum_{i=1}^{n} \left(\widehat{A}_i\widehat{B}_i^\T - A_iB_i^\T\right) \to 0, \notag\\
& \text{and}\,\,
\frac{1}{n} \sum_{i=1}^{n} \left(\widehat{B}_i\widehat{B}_i^\T - B_iB_i^\T\right) \to 0, \quad \textrm{in probability}.
\end{align}

We begin by expanding the first term of \eqref{appB-grdpg:boot_var_noise_3terms}:
\begin{align}\label{appB-grdpg:A_squared_decomp}
\frac{1}{n} \sum_{i=1}^{n} \left( \widehat{A}_i \widehat{A}_i^\T - A_i A_i^\T \right)
&= \frac{1}{n} \sum_{i=1}^{n} (\widehat{A}_i - A_i)(\widehat{A}_i - A_i)^\T \nonumber\\
&\quad + \frac{1}{n} \sum_{i=1}^{n} A_i (\widehat{A}_i - A_i)^\T 
+ \frac{1}{n} \sum_{i=1}^{n} (\widehat{A}_i - A_i) A_i^\T.
\\ &= I_A + II_A + III_A, \text{ say.} \nonumber
\end{align}

Define
\begin{align*}
\Delta_i &:= \operatorname{vec}\left( \mathrm{M}_n \widehat{L}_i \widehat{L}_i^\T \mathrm{M}_n^\T - L_i L_i^\T \right), \\
\widebar{\Delta}_n &:= \operatorname{vec}\left( \frac{1}{n} \sum_{j=1}^n \mathrm{M}_n \widehat{L}_j \widehat{L}_j^\T \mathrm{M}_n^\T - \frac{1}{n} \sum_{j=1}^n L_j L_j^\T \right),
\end{align*}
so that
\begin{align*}
\widehat{A}_i - A_i = \Delta_i - \widebar{\Delta}_n \quad \text{and} \quad \widebar{\Delta}_n = \frac{1}{n} \sum_{j=1}^n \Delta_j.
\end{align*}
Now, observe that,
\begin{align}
I_A =\frac{1}{n} \sum_{i=1}^n (\Delta_i - \widebar{\Delta}_n)(\Delta_i - \widebar{\Delta}_n)^\T \nonumber 
= \frac{1}{n} \sum_{i=1}^n \Delta_i \Delta_i^\T - \widebar{\Delta}_n \widebar{\Delta}_n^\T. 
\end{align}

Thus, to show $I_A \rightarrow 0$ in probability, it suffices to show both $n^{-1} \sum_{i=1}^n \Delta_i \Delta_i^\T \to 0$ and $\widebar{\Delta}_n \widebar{\Delta}_n^\T \to 0$ hold in probability. Notice that:

\begin{align*}
\left| \frac{1}{n} \sum_{i=1}^n \Delta_{i,j} \Delta_{i,k} \right| \leq \frac{1}{n} \sum_{i=1}^n \left| \Delta_{i,j} \Delta_{i,k} \right| \leq \frac{1}{n} \sum_{i=1}^n \| \Delta_i \|_2^2.
\end{align*}
Therefore, it suffices to show that:
\begin{align}\label{appB-grdpg:2-block-norm}
\frac{1}{n} \sum_{i=1}^n \left\| X_i (\mathrm{Q}_n \widehat{Z}_i - Z_i)^\T \right\|_F^2 \to 0
\quad \text{and} \quad
\frac{1}{n} \sum_{i=1}^n \left\| \mathrm{Q}_n \widehat{Z}_i \widehat{Z}_i^\T \mathrm{Q}_n^\T - Z_i Z_i^\T \right\|_F^2 \to 0, \quad  \textrm{in probability}. 
\end{align}

We further bound the first term in \eqref{appB-grdpg:2-block-norm} as follows:
\begin{align*}
\frac{1}{n} \sum_{i=1}^n \left\| X_i (\mathrm{Q}_n \widehat{Z}_i - Z_i)^\T \right\|_F^2 &=\frac{1}{n} \sum_{i=1}^n 
   \| X_i \|_2^2 \cdot \| \mathrm{Q}_n \widehat{Z}_i - Z_i \|_2^2 \notag\\
&\quad \le 
3\!\left(\frac{\rho_n}{\widehat{\rho}_n}-1\right)^{\!2}
\left[
\frac{1}{n} \sum_{i=1}^n 
   \| X_i \|_2^2
   \left\| \mathrm{Q}_n \widecheck{Z}_i - Z_i \right\|_2^2
\right]\\
&\qquad +
3\!\left(\frac{\rho_n}{\widehat{\rho}_n}-1\right)^{\!2}
\left[
\frac{1}{n} \sum_{i=1}^n 
   \| X_i \|_2^2
   \| Z_i \|_2^2
\right] \notag
+
\frac{3}{n} \sum_{i=1}^n 
   \| X_i \|_2^2
   \left\| \mathrm{Q}_n \widecheck{Z}_i - Z_i \right\|_2^2.
\end{align*}
The first term in the above expression is of lower order than the third term since $(\rho_n / \hat{\rho}_n - 1)=O_P(n^{-1/2})$. Moreover, for the second term, we have:
\begin{align*}
\left(\frac{\rho_n}{\widehat{\rho}_n}-1\right)\cdot 
\frac{1}{n}\sum_{i=1}^n \|X_i\|_2^{2}\,\|Z_i\|_2^{2}
&= O_P(n^{-1/2}) \cdot O_P(1)=  o_P(1). 
\end{align*}

For the third term, we upper-bound it by 
\begin{align*}
\frac{3}{n} \sum_{i=1}^n \left( \| X_i \|_2^2 \cdot \| \mathrm{Q}_n \widecheck{Z}_i - Z_i \|_2^2 \right)
&\leq \left( \frac{3}{n} \sum_{i=1}^n \| X_i \|_2^2 \right) \cdot \max_{1 \leq i \leq n} \| \mathrm{Q}_n \widecheck{Z}_i - Z_i \|_2^2\\
&= O_P(1) \cdot O_P\left( \frac{\log^{2c} n}{n \rho_n} \right)= o_P(1), 
\end{align*}
by applying the 2-to-$\infty$ norm bound in Lemma \ref{app-rdpg:ase-error representation}.

Next, we decompose the other term, $n^{-1} \sum_{i=1}^n \left\| \mathrm{Q}_n \widehat{Z}_i \widehat{Z}_i^\T \mathrm{Q}_n^\T - Z_i Z_i^\T \right\|_F^2$, of \eqref{appB-grdpg:2-block-norm} as follows:
\begin{align*}
\left\| \mathrm{Q}_n \widehat{Z}_i \widehat{Z}_i^\T \mathrm{Q}_n^\T - Z_i Z_i^\T \right\|_F^2 
&= \left\| (\mathrm{Q}_n \widehat{Z}_i - Z_i)(\mathrm{Q}_n \widehat{Z}_i - Z_i)^\T 
+ (\mathrm{Q}_n \widehat{Z}_i - Z_i) Z_i^\T 
+ Z_i (\mathrm{Q}_n \widehat{Z}_i - Z_i)^\T \right\|_F^2 \\
&\leq 3 \left\| (\mathrm{Q}_n \widehat{Z}_i - Z_i)(\mathrm{Q}_n \widehat{Z}_i - Z_i)^\T \right\|_F^2 
+ 3 \left\| (\mathrm{Q}_n \widehat{Z}_i - Z_i) Z_i^\T \right\|_F^2 \\
& \qquad  + 3 \left\| Z_i (\mathrm{Q}_n \widehat{Z}_i - Z_i)^\T \right\|_F^2 \\
&= 3 \| \mathrm{Q}_n \widehat{Z}_i - Z_i \|_2^4 
+ 6 \| \mathrm{Q}_n \widehat{Z}_i - Z_i \|_2^2 \; \| Z_i \|_2^2.
\end{align*}

Average the above expression over \( i \), 
we obtain:
\begin{align*}
\frac{1}{n} \sum_{i=1}^n \left\| \mathrm{Q}_n \widehat{Z}_i \widehat{Z}_i^\T \mathrm{Q}_n^\T - Z_i Z_i^\T \right\|_F^2 
&\leq 3 \max_{1 \leq i \leq n}\| \mathrm{Q}_n \widecheck{Z}_i - Z_i \|^4 
+ 6 \max_{1 \leq i \leq n}\| \mathrm{Q}_n \widecheck{Z}_i - Z_i \|^2 \;  \left( \frac{1}{n} \sum_{i=1}^n \| Z_i \|_2^2 \right).
\end{align*}

Under the assumption that $w(x,y) \leq C$, we conclude:
\[
\frac{1}{n} \sum_{i=1}^n \left\| \mathrm{Q}_n \widehat{Z}_i \widehat{Z}_i^\T \mathrm{Q}_n^\T - Z_i Z_i^\T \right\|_F^2 
= O_P\left( \frac{\log^{4c} n}{(n \rho_n)^2} + \frac{\log^{2c} n}{n \rho_n} \right)
= O_P\left( \frac{\log^{2c} n}{n \rho_n} \right)= o_P(1).
\]

Next, we claim that \( \widebar{\Delta}_n \widebar{\Delta}_n^\T \to 0 \) in probability. Recall that \( \widebar{\Delta}_n = n^{-1} \sum_{j=1}^n \Delta_j \). To handle the outer product, we directly appeal to the vector norm:
\begin{align*}
\left\| \widebar{\Delta}_n \widebar{\Delta}_n^\T \right\|_{F} = \| \widebar{\Delta}_n \|_2^2\leq \frac{1}{n} \sum_{j=1}^n \| \Delta_j \|_2^2,
\end{align*}
we conclude 
$n^{-1} \sum_{i=1}^{n} (\widehat{A}_i - A_i)(\widehat{A}_i - A_i)^\T \to 0$ in probability.

We next analyze the remaining linearized noise terms of \eqref{appB-grdpg:A_squared_decomp}:
\begin{align*}
\frac{1}{n} \sum_{i=1}^{n} A_i (\widehat{A}_i - A_i)^\T 
\quad \text{and} \quad 
\frac{1}{n} \sum_{i=1}^{n} (\widehat{A}_i - A_i) A_i^\T.    
\end{align*}

By the decomposition \( \widehat{A}_i - A_i = \Delta_i - \widebar{\Delta}_n \), we rewrite:
\begin{align*}
\frac{1}{n} \sum_{i=1}^{n} A_i (\widehat{A}_i - A_i)^\T 
&= \frac{1}{n} \sum_{i=1}^{n} A_i \Delta_i^\T 
- \left( \frac{1}{n} \sum_{i=1}^{n} A_i \right) \widebar{\Delta}_n^\T \\
&= \frac{1}{n} \sum_{i=1}^{n} A_i \Delta_i^\T,
\end{align*}
since \( n^{-1} \sum_{i=1}^{n} A_i = 0 \) by construction.

To bound the term \( \frac{1}{n} \sum_{i=1}^{n} A_i \Delta_i^\T \), we apply the triangle inequality and the Cauchy–Schwarz inequality:
\begin{align*}
\left\| \frac{1}{n} \sum_{i=1}^{n} A_i \Delta_i^\T \right\|_F 
&\leq \frac{1}{n} \sum_{i=1}^{n} \|A_i \Delta_i^\T\|_{F} \nonumber \\
&\leq \left( \frac{1}{n} \sum_{i=1}^{n} \| A_i \|_2^2 \right)^{1/2} 
\left( \frac{1}{n} \sum_{i=1}^{n} \| \Delta_i \|_2^2 \right)^{1/2},
\end{align*}
for which, we have already established that $n^{-1} \sum_{i=1}^{n} \| \Delta_i \|_2^2 \to 0$ in probability. Hence, under $E \| A_i \|_2^2 < \infty$, which is guaranteed by $E\|X_1\|_2^4<\infty$ and the assumption that $w(x,y) \leq C$, it then follows that:
\[
\left\| \frac{1}{n} \sum_{i=1}^{n} A_i (\widehat{A}_i - A_i)^\T \right\|_F= O_P\left( \frac{\log^{c} n}{(n\rho_n)^{1/2}} \right)= o_P(1),
\]
similar for the transpose term, which is $n^{-1} \sum_{i=1}^{n} (\widehat{A}_i - A_i) A_i^\T \to 0$ in probability.

This verifies those two linearized terms in \eqref{appB-grdpg:A_squared_decomp} and concludes $n^{-1} \sum_{i=1}^{n} \left( \widehat{A}_i \widehat{A}_i^\T - A_i A_i^\T \right) \to 0$ in probability.

The proof of the second claim in \eqref{appB-grdpg:boot_var_noise_3terms} largely follows the same reasoning; we describe the first few steps in this argument below. We decompose this term as: 
\begin{align}
\begin{split}
\frac{1}{n} \sum_{i=1}^{n} \left(\widehat{A}_i\widehat{B}_i^\T - A_iB_i^\T\right)&= \frac{1}{n} \sum_{i=1}^{n} \left(\widehat{A}_i - A_i\right)\left(\widehat{B}_i - B_i\right)^\T\\ 
& \quad + \frac{1}{n} \sum_{i=1}^{n} A_i\left(\widehat{B}_i - B_i\right)^\T 
+ \frac{1}{n} \sum_{i=1}^{n} \left(\widehat{A}_i - A_i\right) B_i^\T. \label{cross-prodct:remainder term}
\end{split}
\end{align}

Define the quantities:
\begin{align*}
\Upsilon_i := \operatorname{vec}\left( \mathrm{M}_n \widehat{L}_i Y_i - L_i Y_i \right), \quad
\widebar{\Upsilon}_n := \frac{1}{n} \sum_{j=1}^{n} \Upsilon_j,    
\end{align*}
so that:
\begin{align*}
\widehat{B}_i - B_i = \Upsilon_i  - \widebar{\Upsilon}_n.    
\end{align*}

Thus, the first term in \eqref{cross-prodct:remainder term} becomes:
\begin{align*}
\frac{1}{n} \sum_{i=1}^{n} (\widehat{A}_i - A_i)(\widehat{B}_i - B_i)^\T
= \frac{1}{n} \sum_{i=1}^{n} \Delta_i \Upsilon_i ^\T - \widebar{\Delta}_n \widebar{\Upsilon}_n^\T.    
\end{align*}

We start by controlling the cross-product term $\frac{1}{n} \sum_{i=1}^{n} \Delta_i \Upsilon_i ^\T$:
\begin{align*}
\left| \frac{1}{n} \sum_{i=1}^n \Delta_{i,j} \Upsilon_{i,k} \right| \leq \frac{1}{n} \sum_{i=1}^n \left| \Delta_{i,j} \Upsilon_{i,k} \right| \leq \frac{1}{n} \sum_{i=1}^{n} \| \Delta_i \|_2 \cdot \| \Upsilon_i  \|_2 \leq \left( \frac{1}{n} \sum_{i=1}^{n} \| \Delta_i \|_2^2 \right)^{1/2}
\left( \frac{1}{n} \sum_{i=1}^{n} \| \Upsilon_i  \|_2^2 \right)^{1/2}. 
\end{align*}

We have already shown that $
\frac{1}{n} \sum_{i=1}^{n} \| \Delta_i \|_2^2 \to 0$ in  probability. Averaging over $i$, we upper bound it by:
\begin{align*}
\frac{1}{n} \sum_{i=1}^{n}\| \Upsilon_i  \|_2^2= \frac{1}{n} \sum_{i=1}^{n} \| \mathrm{Q}_n \widehat{Z}_i - Z_i \|_2^2 \cdot Y_i^2
&\leq \left( \frac{1}{n} \sum_{i=1}^{n} Y_i^2 \right) \cdot \max_{1 \leq i \leq n} \| \mathrm{Q}_n \widecheck{Z}_i - Z_i \|_2^2 \\
&= O_P(1) \cdot O_P\left( \frac{\log^{2c} n}{n \rho_n} \right) = O_P\left( \frac{\log^{2c} n}{n \rho_n} \right). \nonumber
\label{eq:Y_cross_term_bound}
\end{align*}
Thus, $n^{-1}\sum_{i=1}^{n} \Delta_i \Upsilon_i ^\T \rightarrow 0 $ in probability.  The rest of the argument is analogous to the first term.  The third claim in \eqref{appB-grdpg:boot_var_noise_3terms} follows identical reasoning. 
\end{proof}

\begin{lemma}[Cross--term Noise]\label{appB-grdpg:lem-cross-noise}
With the quantities defined as before, and under the conditions of Theorem \ref{rdpg:theorem6}, the following holds:
\begin{align*}
\frac{1}{n} \sum_{i=1}^{n} \!\left( \widehat{A}_i \widehat{C}_i - A_i C_i \right) \to  0,
\qquad
\frac{1}{n} \sum_{i=1}^{n} \!\left( \widehat{B}_i \widehat{C}_i - B_i C_i \right) \to  0 \quad\textrm{in probability}.    
\end{align*}
\end{lemma} 
\begin{proof}
Denote
\begin{align*}
\widehat a_i&:=\frac{1}{(n-1)\,\widehat\rho_n}\sum_{j\ne i}A_{ij},&
a_i&:=E[w(\xi_i,\xi_j)\mid \xi_i].&
\textrm{Hence,}\quad \widehat C_i&=2(\widehat a_i-1),\; C_i=2(a_i-1).
\end{align*}

We start by proving the first claim. Note that:
\begin{align*}
\frac{1}{n}\sum_{i=1}^n(\widehat A_i\widehat C_i-A_iC_i)
&=\frac{1}{n}\sum_{i=1}^n\big[(\widehat A_i-A_i)C_i+A_i(\widehat C_i-C_i)+(\widehat A_i-A_i)(\widehat C_i-C_i)\big].
\end{align*}

Thus, $n^{-1}\sum_{i=1}^{n} \big( \widehat{A}_i \widehat{C}_i - A_i C_i \big) \to 0$ in probability whenever 
\begin{align*}
&n^{-1}\sum_{i=1}^n \|\Delta_i\|_2^2= o_P(1),  \quad
n^{-1}\sum_{i=1}^n \|A_i\|^2=O_P(1),\quad
n^{-1}\sum_{i=1}^n (a_i-1)^2=O_P(1),\quad\\
&\text{and}\,\,
n^{-1}\sum_{i=1}^n (\widehat a_i-a_i)^2= o_P(1).    
\end{align*}

For the last term, introduce $\widetilde a_i:=\sum_{j\ne i}w(\xi_i,\xi_j)/(n-1)$ and note
\begin{align*}
\widehat a_i-a_i&=(\widehat a_i-\widetilde a_i)+(\widetilde a_i-a_i),\\
n^{-1}\sum_{i=1}^n(\widehat a_i-a_i)^2
&\le 2\Big\{n^{-1}\sum_{i=1}^n(\widehat a_i-\widetilde a_i)^2+n^{-1}\sum_{i=1}^n(\widetilde a_i-a_i)^2\Big\},
\end{align*}
where the two terms on the upper bound are handled separately in Lemma \ref{appB-lem:bootstrap-lemma-3} and Lemma \ref{appB-lem:bootstrap-lemma-4}. 

Furthermore, $n^{-1} \sum_{i=1}^{n} \!\left( \widehat{B}_i \widehat{C}_i - B_i C_i \right) \to  0$ in probability, by observing  $n^{-1}\sum_{i=1}^n \|\Upsilon_i\|_2^2=o_P(1)$  as well as $n^{-1}\sum_{i=1}^n\|B_i\|^2=O_P(1)$  under the conditions of Theorem \ref{rdpg:theorem6}. The claim follows.
\end{proof}

We are now ready to prove Theorem~\ref{rdpg:theorem6}.

\begin{proof}[Proof of Theorem \ref{rdpg:theorem6}.]
By Lemma \ref{lem-grdpg-boot-variance-consistency}, $\widehat{\Sigma}^\flat \to \Sigma_\Psi$.  Since the multipliers are Gaussian, this in turn implies:
\begin{align*}
\sup_{u \in \mathbb{R}^{q+1}} \Big| 
\operatorname{pr}^{\flat} \!\left( n^{1/2}(\widehat{J}_n^{\flat} - \widehat{J}^{M}_n) \leq u \right)
-
\operatorname{pr} \!\left( n^{1/2}(\widecheck{J}^{M}_n - \eta^*) \leq u \right)
\Big| &\;\to 0 \quad \text{in probability}.   
\end{align*}
where $\widehat{J}^{M}_n=\big[ \operatorname{vec}\!\big( \mathrm{M}_n^{\T}( n^{-1} \widehat{\mathrm{L}}^{\T} \widehat{\mathrm{L}} ) \mathrm{M}_n \big)^{\T},\; 
\big( \mathrm{M}_n^{\T}( n^{-1} \widehat{\mathrm{L}}^{\T} Y )\big)^{\T},\; 1 \big]^{\T}$.

Furthermore, the delta method for the bootstrap (Theorem~23.5 of \cite{vanderVaart1998}), provided \(f \circ g\) is continuously differentiable in a neighborhood of \(\eta^*\), as ensured by Lemma \ref{appB:OLS_functional_property}, the claim of Theorem \ref{rdpg:theorem6} follows.
\end{proof}

\begin{proof}[Proof of Corollary~\ref{hypothesis testing network}.]
By Theorem~\ref{rdpg:theorem6},  
\[
   \sup_{u \in \mathbb{R}^{p+d}} 
   \left| \operatorname{pr}^{\flat} \left(n^{1/2}\mathrm{M}_n(\widehat{\beta}_{\mathrm{z}}^{\flat}-\widehat{\beta}_{\mathrm{z}}) \leq u\right) 
   - \operatorname{pr} \left( \Sigma_\beta^{1/2} Z \leq u \right)\right| \to 0 
   \quad \text{in probability}.
\]
Moreover, Theorem~\ref{theorem 5} yields
\[
n^{1/2}\mathrm{M}_n(\widehat{\beta}_{\mathrm{z}}-\beta^*_{\mathrm{z}}) \to \Sigma_\beta^{1/2}Z \quad \textrm{in distribution}.
\]
 where $Z \sim N(0, I_{p+d})$.

Applying the continuous mapping theorem under $H_0:\beta^{*}_{\mathrm{z}}=0$,
\[
\sup_{u\in\mathbb{R}}
\Big|
\operatorname{pr}^{\flat}\!\big(\widehat{T}^{\flat}_{test}(\widehat{\beta}_{\mathrm{z}})\le u\big)
-
\operatorname{pr} \big(\widehat{T}_{test}(0)\le u\big)
\Big| \to 0
\quad\text{in probability},
\]
where $\widehat{T}_{test}(0)=n\|\widehat{\beta}_{\mathrm{z}}\|_2^2$. Let $c_{1-\alpha}^\flat$ be the $(1-\alpha)$ bootstrap quantile of $\widehat{T}^{\flat}_{test}(\widehat{\beta}_{\mathrm{z}})$.  By standard arguments, weak convergence in probability implies convergence in probability to quantiles at a continuity point of the limiting distribution. Therefore, 
\[
\operatorname{pr} \big(\widehat{T}_{test}(\beta^{*}_{\mathrm{z}})>c_{1-\alpha}^\flat\big) \to \alpha
\quad\text{under } H_0.
\]
The claim follows.
\end{proof}

\section{Proofs For Section \ref{sec:down-sampling}}\label{App-D}

\subsection{Proof of Theorem \ref{thm:down-sampling} and Proposition \ref{coro:2}}
\label{subsec-appedix-down-sampling}

\begin{proof}[Proof of Theorem \ref{thm:down-sampling}.]
Define 
\begin{align*}
\widehat{\Psi}_{m}=
\begin{pmatrix}
\operatorname{vec}(\widehat{\Lambda}_m)\\
\widehat{\gamma}_m
\end{pmatrix},
\qquad \textrm{where}\quad
\widehat{\Lambda}_m := \frac{1}{m} \sum_{i=1}^m \widehat{L}_i \widehat{L}_i^\T, \quad \text{and} 
\quad \widehat{\gamma}_m:= \frac{1}{m} \sum_{i=1}^m \widehat{L}_i Y_i.    
\end{align*}

Analogously, define 
\begin{align*}
\Psi_{m}=
\begin{pmatrix}
\operatorname{vec}(\Lambda_m)\\
\gamma_m
\end{pmatrix},
\quad \textrm{where}\quad
\Lambda_m := \frac{1}{m} \sum_{i=1}^m L_i L_i^\T,
\quad \gamma_m:= \frac{1}{m} \sum_{i=1}^m L_i Y_i,
\quad \textrm{and}\quad 
\psi_{i}=
\begin{pmatrix}
\operatorname{vec}(L_iL_i^\T)\\
L_iY_i
\end{pmatrix}.
\end{align*}
Therefore, the population parameter $\Psi=(\Lambda,\gamma)$ satisfies $\Psi = E[\psi_1]$.

We start by studying the $m^{1/2}$-consistency of $(\widehat{\Psi}_{m}-\Psi)$ by the following decomposition:
\begin{align}\label{appC:decomp_ds}
m^{1/2}\left(\widehat{\Psi}_{m}-\Psi_{m}+\Psi_{m}-\Psi\right).
\end{align}
To establish $m^{1/2}\left(\widehat{\Psi}_{m}-\Psi_{m}\right) \to 0$ in probability, it suffices to show:
\begin{align}\label{appC:ds_noise_claim}
m^{-1/2} \sum_{i=1}^m (\widehat{L}_i Y_i - L_i Y_i) = o_P(1),
\qquad
m^{-1/2} \sum_{i=1}^m (\widehat{L}_i \widehat{L}_i^\T - L_i L_i^\T) = o_P(1).   
\end{align}

Moreover, \eqref{appC:ds_noise_claim} would follow if
\begin{align}
m^{-1/2} \sum_{i=1}^m Y_i\, (\widehat{Z}_{ij} - Z_{ij}) &= o_P(1)\qquad \forall j, \label{appC:linear_noise}\\
m^{-1/2} \sum_{i=1}^m X_i\, (\widehat{Z}_{ij} - Z_{ij}) &= o_P(1)\qquad \forall j, \label{appC:linear_noise_X}\\
m^{-1/2} \sum_{i=1}^m \left( \widehat{Z}_{ij} \widehat{Z}_{ik} - Z_{ij} Z_{ik} \right) &= o_P(1) \qquad \forall j, k. \label{appC:cross_noise}
\end{align}

We start by upper-bounding \eqref{appC:linear_noise} and \eqref{appC:linear_noise_X} would follow analogously,
\begin{align*}
\left| m^{-1/2} \sum_{i=1}^m Y_i \, (\widehat{Z}_{ij} - Z_{ij}) \right| 
&\leq m^{-1/2} \sum_{i=1}^m |Y_i| \, |\widehat{Z}_{ij} - Z_{ij}|  \\
&\leq m^{1/2} \max_{1 \leq i \leq n} |\widehat{Z}_{ij} - Z_{ij}|\cdot\frac{1}{m}\sum_{i=1}^m |Y_i|.
\end{align*}

Since, by assumption,
\begin{align}\label{appC:a_n}
\max_{1 \leq i \leq n} |\widehat{Z}_{ij} - Z_{ij}| = O_P(a_n) \quad \forall j.    
\end{align}
It follows that:
\[
\left| m^{-1/2} \sum_{i=1}^m Y_i \, (\widehat{Z}_{ij} - Z_{ij}) \right|
\leq m^{1/2} \cdot O_P(a_n) \cdot O_P(1),
\]
\noindent
by noting that \( m^{-1}\sum_{i=1}^m |Y_i| = O_P(1) \) under \( E(Y_i^2) < \infty \). Therefore, the entire expression is \( o_P(1) \) when \( a_n = o(m^{-1/2}) \), which corresponds precisely to Condition~\ref{down-sampling:condition2}.

We now bound \eqref{appC:cross_noise}, we write
\begin{align}\label{appC-eq:corss_upperbound}
m^{-1/2} \sum_{i=1}^m \left( \widehat{Z}_{ij} \widehat{Z}_{ik} - Z_{ij} Z_{ik} \right) 
= m^{-1/2} \sum_{i=1}^m \left[ (\widehat{Z}_{ij} - Z_{ij})(\widehat{Z}_{ik} - Z_{ik}) + Z_{ij}(\widehat{Z}_{ik} - Z_{ik}) + Z_{ik}(\widehat{Z}_{ij} - Z_{ij}) \right].   
\end{align}

Each term can be handled separately. The first term on the RHS of \eqref{appC-eq:corss_upperbound} is upper-bounded by
\begin{align*}
\left| m^{-1/2} \sum_{i=1}^m (\widehat{Z}_{ij} - Z_{ij})(\widehat{Z}_{ik} - Z_{ik}) \right| 
\leq m^{1/2} \max_i |\widehat{Z}_{ij} - Z_{ij}| \cdot \max_i |\widehat{Z}_{ik} - Z_{ik}| = O(m^{1/2}a_n^2),    
\end{align*}
according to \eqref{appC:a_n}. The second and third terms are upper-bounded similarly as
\begin{align*}
\left| m^{-1/2} \sum_{i=1}^m Z_{ij}(\widehat{Z}_{ik} - Z_{ik}) \right| \leq m^{1/2} \max_{1\leq i\leq n} |\widehat{Z}_{ik} - Z_{ik}|\cdot \frac{1}{m}\sum_{i=1}^m |Z_{ij}|  = O_P(m^{1/2}a_n)\cdot O_P(1).    
\end{align*}

All three upper-bounds above vanish as $o_P(1)$ under the same regime given in Condition \ref{down-sampling:condition2}. Therefore, \eqref{appC:ds_noise_claim} holds and verifies $m^{1/2}\left(\widehat{\Psi}_{m}-\Psi_{m}\right) \to 0$ in probability under Condition \ref{down-sampling:condition2} as well as $E(Y^2)<\infty$, $E\|X\|^2<\infty$ and $w(x,y) \leq C$.

To conclude the proof, we first invoke the central limit theorem  for $m^{1/2}(\Psi_{m}-\Psi)$ due to i.i.d. $\psi_{i}$ that are averaged in $\Psi_{m}$. Then under the moment conditions assumed in Theorem~\ref{thm:down-sampling}, 
\begin{align}
m^{1/2} \left( \widehat{\Psi}_{m} - \Psi \right) \to N(0, \Sigma_{\Psi,m} ) \quad \textrm{in distribution},    
\end{align}
where $\Sigma_{\Psi,m}=\operatorname{var}(\psi_i)$.

Further, under the assumption that \(\Lambda = E[L_i L_i^\T]\) is finite and positive definite, 
the delta method applied to the OLS functional \(f\) in 
Lemma~\ref{appB:OLS_functional_property} yields $m^{1/2}$--consistency for 
\(\widehat{\beta}^{(m)}\), which is
\[
m^{1/2} \left( \widehat{\beta}^{(m)} - \beta^* \right) 
\to N\!\left(0, \, \Sigma^{(m)}_\beta \right) \quad \textrm{in distribution},
\]
where
\[
\Sigma^{(m)}_\beta = D_{f}(\Lambda,\gamma)\,\Sigma_{\Psi,m}\,D_{f}(\Lambda,\gamma)^\T.
\]
This completes the proof of Theorem~\ref{thm:down-sampling}.
\end{proof}

We next prove Proposition \ref{coro:2}.

\begin{proof}[Proof of Proposition \ref{coro:2}.]
Under the decomposition in \eqref{appC:decomp_ds}, we establish that 
$m^{1/2}\left(\widehat{\Psi}_m - \Psi_m\right) \to 0$ in probability by proving 
\eqref{appC:ds_noise_claim} using a different argument, which still corresponds 
to verifying \eqref{appC:linear_noise}--\eqref{appC:cross_noise}.

For \eqref{appC:linear_noise} and \eqref{appC:linear_noise_X}, we apply Markov's inequality together with the Cauchy--Schwarz inequality to obtain:
\begin{align}\label{appC:prop2-linear}
\left| E(\,Y_i(\widehat{Z}_{ij} - Z_{ij})) \right|
\ \leq\ 
\left( E(Y_i^2) \cdot E[(\widehat{Z}_{ij} - Z_{ij})^2] \right)^{1/2}.    
\end{align}
Under the assumption \(E[(\widehat{Z}_{ij} - Z_{ij})^2] = o(m^{-1})\) from Proposition~\ref{coro:2} as well as \(E(Y^2)<\infty\) and \(E\|X\|^2<\infty\), the right-hand side of \eqref{appC:prop2-linear} is \(o(m^{-1/2})\), thereby establishing \eqref{appC:linear_noise} and \eqref{appC:linear_noise_X}.

For \eqref{appC:cross_noise}, application of Markov’s inequality and the Cauchy--Schwarz inequality to each term yields
\begin{align*}
E\big|(\widehat{Z}_{ij}-Z_{ij})(\widehat{Z}_{ik}-Z_{ik})\big|
&\le \big(E(\widehat{Z}_{ij}-Z_{ij})^{2}\big)^{1/2}
     \big(E(\widehat{Z}_{ik}-Z_{ik})^{2}\big)^{1/2}
     = o(m^{-1}),\\[4pt]
E\big|Z_{ij}(\widehat{Z}_{ik}-Z_{ik})\big|
&\le \big(E Z_{ij}^{2}\big)^{1/2}
     \big(E(\widehat{Z}_{ik}-Z_{ik})^{2}\big)^{1/2}
     = o(m^{-1/2}),
\end{align*}
and similarly for the symmetric term with $Z_{ik}(\widehat{Z}_{ij}-Z_{ij})$.  
Under the assumptions of Proposition~\ref{coro:2} and $w(x,y) \leq C$, this shows that \eqref{appC:cross_noise} is $o_P(1)$.

Hence, we conclude that $m^{1/2}(\widehat{\Psi}_m-\Psi_m)\to 0$ in probability.

Therefore, the above bounds are sufficient to establish the same $m^{1/2}$--consistency:
\[
m^{1/2} \big( \widehat{\beta}^{(m)} - \beta^* \big) 
\to N\!\left(0, \Sigma^{(m)}_\beta \right) \quad \text{in distribution}.
\]

Furthermore, under the same conditions,
\begin{align}\label{target_shift}
m^{1/2} (\widetilde{\beta} - \beta^*) = o(1),    
\end{align}
due to \( E\|\widehat{Z}_i - Z_i\|^2 = o(m^{-1}) \). To prove \eqref{target_shift}, it suffices to show that $E(\widehat{Z}^2_{ij}-Z_{ij}^2)=o(m^{-1/2})$ for any $j$, and $E(\widehat{Z}_{ik}\widehat{Z}_{ij}-Z_{ij}Z_{ik})=o(m^{-1/2})$ for any $j \neq k$. We have, by Cauchy-Schwarz,
\begin{align*}
E(\widehat{Z}_{ij}^2 - Z_{ij}^2)
= E\big[(\widehat{Z}_{ij}-Z_{ij})(\widehat{Z}_{ij}+Z_{ij})\big]
\le \big\{E(\widehat{Z}_{ij}-Z_{ij})^2\big\}^{1/2}\,
     \big\{E(\widehat{Z}_{ij}+Z_{ij})^2\big\}^{1/2}. 
\end{align*}
under the assumption $w(x,y)\leq C$ and $E\|\widehat{Z}_i - Z_i\|^2 = o\left(m^{-1}\right)$, and also
\begin{align*}
E[(\widehat{Z}_{ij}+Z_{ij})^2] = E[((\widehat{Z}_{ij}-Z_i)+2Z_{ij})^2] \le 2E[(\widehat{Z}_{ij}-Z_{ij})^2]+8E[Z_{ij}^2]=O(1).
\end{align*}
Similarly, by the Cauchy-Schwarz inequality,
\begin{align*}
E(\widehat{Z}_{ik}\widehat{Z}_{ij}-Z_{ij}Z_{ik})
&= E\left[ (\widehat{Z}_{ij}-Z_{ij})(\widehat{Z}_{ik}-Z_{ik})
+ Z_{ij}(\widehat{Z}_{ik}-Z_{ik})
+ Z_{ik}(\widehat{Z}_{ij}-Z_{ij}) \right]\\
&\le
\{E(\widehat{Z}_{ij}-Z_{ij})^2\}^{1/2}
\{E(\widehat{Z}_{ik}-Z_{ik})^2\}^{1/2}\\
&\quad
+ \{E Z_{ij}^2\}^{1/2}\{E(\widehat{Z}_{ik}-Z_{ik})^2\}^{1/2}
+ \{E Z_{ik}^2\}^{1/2}\{E(\widehat{Z}_{ij}-Z_{ij})^2\}^{1/2}.
\end{align*}
Thus, under the same assumed conditions,
\eqref{target_shift} holds.

Therefore, the same central limit theorem holds for $\widetilde{\beta}$:
\[
m^{1/2} \big( \widehat{\beta}^{(m)} - \widetilde{\beta} \big) 
\to N\!\left(0, \Sigma^{(m)}_\beta \right) \quad \text{in distribution},
\]
the claim follows.
\end{proof}

\subsection{Proof of Example \ref{ex:graphon_ratio}}\label{ds-example1}
Before presenting the examples, we justify that the plug-in error from the estimator $\widehat{\rho}_n$ is negligible. 
For both cases, suppose \( \widehat{Z}_{ij} = \left(  \rho_n/ \widehat{\rho}_n \right)^{e_j} \widecheck{Z}_{ij} \). We decompose
\begin{align*}
|\widehat{Z}_{ij} - Z_{ij}|
&= \left| \left( \frac{\rho_n}{\widehat{\rho}_n} \right)^{e_j} \widecheck{Z}_{ij} - Z_{ij} \right| \\
&\leq \left| \left( \frac{\rho_n}{\widehat{\rho}_n} \right)^{e_j} - 1\right|\,|Z_{ij}|
+ \left|\left( \frac{\rho_n}{\widehat{\rho}_n} \right)^{e_j}-1\right|\,|\widecheck{Z}_{ij} - Z_{ij}| 
+ |\widecheck{Z}_{ij} - Z_{ij}| \\
&\equiv \Delta_\rho |Z_{ij}|+\Delta_\rho |\widecheck{Z}_{ij} - Z_{ij}|+|\widecheck{Z}_{ij} - Z_{ij}|,
\end{align*}
where $\Delta_\rho = \bigl| \left(  \rho_n/ \widehat{\rho}_n \right)^{e_j} - 1\bigr| = O_P(n^{-1/2})$. 
It follows that the plug-in error does not affect the upper bounds for \eqref{appC:linear_noise}-\eqref{appC:cross_noise}, since $\Delta_\rho$ converges at a faster rate.

Explicitly, the upper bound of \eqref{appC:linear_noise} now becomes
\begin{align*}
\left| m^{-1/2} \sum_{i=1}^m Y_i (\widehat{Z}_{ij} - Z_{ij}) \right|
&\leq m^{-1/2} \sum_{i=1}^m |Y_i|\,|\widehat{Z}_{ij} - Z_{ij}| \\
&\leq m^{1/2}\max_{1\leq i\leq m}\! \left|\widecheck{Z}_{ij} - Z_{ij}\right|
       \cdot \frac{1}{m}\sum_{i=1}^m |Y_i| \\
&\quad + m^{1/2}\Delta_\rho \Biggl(
       \frac{1}{m}\sum_{i=1}^m |Y_i||Z_{ij}|
       + \max_{1\leq i\leq m}\! \left|\widecheck{Z}_{ij} - Z_{ij}\right|
         \cdot \frac{1}{m}\sum_{i=1}^m |Y_i| \Biggr).
\end{align*}

And similarly for \eqref{appC:cross_noise}, the second and third term in \eqref{appC-eq:corss_upperbound} follow similar argument as above, and for the first term,
\begin{align*}
\left| m^{-1/2} \sum_{i=1}^m (\widehat{Z}_{ij} - Z_{ij})(\widehat{Z}_{ik} - Z_{ik}) \right|
&\leq m^{1/2}\max_{1\leq i\leq m}\! \left|\widecheck{Z}_{ij} - Z_{ij}\right|\cdot \max_{1\leq i\leq m}\left|\widecheck{Z}_{ik} - Z_{ik}\right|+o_P(m^{1/2}a^2_n).
\end{align*}

For all the three claims, since 
$m^{-1} \sum_{i=1}^m |Y_i||Z_{ij}| = O_P(1)$, $m^{-1} \sum_{i=1}^m |X_i||Z_{ij}| = O_P(1)$ and 
$m^{-1} \sum_{i=1}^m |Z_{ij}||Z_{ik}| = O_P(1)$ 
under the assumptions of Theorem \ref{thm:down-sampling} that $E\|Z\|^4 < \infty$, $E(Y^4) < \infty$ and $w(x,y) \leq C$, the leading term that does not involve $\Delta_\rho$ continues to dominate, while the plug-in error contributes only lower-order terms. Hence, the upper bound remains $O_P(m^{1/2}a_n)$. 

In the following examples, the argument will be applied to $\widecheck{Z}_i - Z_i$.

\begin{proof}[Proof of Proposition \ref{prop:composite_control}.]
Without loss of generality, consider a single network covariate that admits the representation 
\[
\widecheck{Z}_i = \frac{ \widecheck{U}_i(\mathcal{R}_1)}{ \widecheck{U}_i(\mathcal{R}_2)}.
\]

We start by deriving $\widetilde{a}_n$ for
\begin{align}
\max_{1 \leq i \leq m} |\widecheck{Z}_i - Z_i| = O_P(\widetilde{a}_n).
\end{align}

Denote $U_i(\mathcal{R}_1)=E(\widecheck{U}_i(\mathcal{R}_1)\mid (X_i,\xi_i))$, and similarly for $U_i(\mathcal{R}_2)$. 
A Taylor expansion of $f(x,y) = x/y$ around $(U_i(\mathcal{R}_1), U_i(\mathcal{R}_2))$ yields
\begin{align*}
\widecheck{Z}_i - Z_i = \nabla f\big(\kappa_i(\mathcal{R}_1), \kappa_i(\mathcal{R}_2)\big) \cdot 
\begin{pmatrix}
\widecheck{U}_i(\mathcal{R}_1) - U_i(\mathcal{R}_1) \\
\widecheck{U}_i(\mathcal{R}_2) - U_i(\mathcal{R}_2)
\end{pmatrix},
\quad\textrm{where}\quad
\nabla f(x,y) = \left( \frac{1}{y}, \, -\frac{x}{y^2} \right)^{\T},
\end{align*}
and  $\big(\kappa_i(\mathcal{R}_1), \kappa_i(\mathcal{R}_2)\big)$  lie between $(U_i(\mathcal{R}_1), U_i(\mathcal{R}_2))$ and $(\widecheck{U}_i(\mathcal{R}_1), \widecheck{U}_i(\mathcal{R}_2))$.

Expanding  $\widecheck{Z}_i - Z_i $  yields,
\begin{align}\label{appC:taylor_expansion_ratio}
\widecheck{Z}_i - Z_i 
= \frac{ \widecheck{U}_i(\mathcal{R}_1) - U_i(\mathcal{R}_1) }{  \kappa_i(\mathcal{R}_2) }
- \frac{  \kappa_i(\mathcal{R}_1) }{  \big( \kappa_i(\mathcal{R}_2) \big)^2  } 
\big( \widecheck{U}_i(\mathcal{R}_2) - U_i(\mathcal{R}_2) \big).
\end{align}

We analyze each noise component of \eqref{appC:taylor_expansion_ratio} and start with $\widecheck{U}_i(\mathcal{R}_1) - U_i(\mathcal{R}_1)$; the argument for $\widecheck{U}_i(\mathcal{R}_2) - U_i(\mathcal{R}_2)$ is analogous. For each $i$,
\[
\widecheck{U}_i(\mathcal{R}_1) - U_i(\mathcal{R}_1)
= \big( \widecheck{U}_i(\mathcal{R}_1) - U_i(\mathcal{R}_1;\xi) \big)
+ \big( U_i(\mathcal{R}_1;\xi) - U_i(\mathcal{R}_1) \big),
\]
where $U_i(\mathcal{R}_1;\xi)=E(\widecheck{U}_i(\mathcal{R}_1)\mid(X,\xi))$, and similarly for $U_i(\mathcal{R}_2;\xi)$.

For any $t>0$,
\begin{align}\label{appC:ratio-noise}
\operatorname{pr}\big( |\widecheck{U}_i(\mathcal{R}_1) - U_i(\mathcal{R}_1)| > 2t \big)
\le \operatorname{pr}\big( |\widecheck{U}_i(\mathcal{R}_1) - U_i(\mathcal{R}_1;\xi)| > t \big)
+ \operatorname{pr}\big( |U_i(\mathcal{R}_1;\xi) - U_i(\mathcal{R}_1)| > t \big).    
\end{align}
Conditional on $\xi_i$, the local $U_i(\mathcal{R}_1;\xi)$ is a U-statistic of rank $r_1-1$ that also has a bounded kernel.

By applying Hoeffding's inequality for U-statistics, for some constant $B>0$, it follows that:
\begin{align*}
\operatorname{pr}\big( |U_i(\mathcal{R}_1;\xi) - E[U_i(\mathcal{R}_1;\xi)]| > t \big) &= E_{\xi_i}\left[ \operatorname{pr}\big( |U_i(\mathcal{R}_1;\xi) - E[U_i(\mathcal{R}_1;\xi)]| > t \mid \xi_i \big) \right] \\
&\le 2\exp\left( -\frac{nt^2}{B} \right).
\end{align*}

For the first term on the RHS of \eqref{appC:ratio-noise}, applying Chebyshev's inequality, and following Lemma~\ref{appB-lem:bootstrap-lemma-1} for $\operatorname{var}\left(\widecheck{U}_i(\mathcal{R}_1) - U_i(\mathcal{R}_1;\xi)\right)$, it satisfies
\begin{align*}
\operatorname{pr}\big( |\widecheck{U}_i(\mathcal{R}_1) - U_i(\mathcal{R}_1;\xi)| > t \big)
\le \frac{ \operatorname{var}\left(\widecheck{U}_i(\mathcal{R}_1) - U_i(\mathcal{R}_1;\xi)\right) }{t^2}
= O\!\left( \frac{n^{1-c_1}\rho_n^{-d_1}}{t^2} \right),  
\end{align*}
where $c_1,d_1$, as the number of nodes and edges shared from the overlap of $\mathcal{R}_1$ at node $i$. Therefore,
\begin{align*}
\operatorname{pr}\big( |\widecheck{U}_i(\mathcal{R}_1) - U_i(\mathcal{R}_1)| > 2t \big)
\le 2\exp\left(- \frac{n t^2}{B}\right) + O\!\left( \frac{n^{1-c_1}\rho_n^{-d_1}}{t^2} \right).    
\end{align*}
An analogous bound holds for $\mathcal{R}_2$. Now, for the denominator, under the assumption that \( 0 < c \leq w(\xi_i,\xi_j) \leq C <\infty \), there also exists some $\v_0'$ such that $U_i(\mathcal R_2)\ge \v_0'>0$ a.s. Take $\v_0=\v_0'/2$,
\begin{align*}
\operatorname{pr}\big( |\widecheck{Z}_i - Z_i| > t \big)
&\le \operatorname{pr}\big( \kappa_i(\mathcal{R}_2)  < \v_0 \big)
+ \operatorname{pr}\big( |\widecheck{Z}_i - Z_i| > t,  \kappa_i(\mathcal{R}_2)  \ge \v_0 \big).
\end{align*}

On $\{\kappa_i(\mathcal{R}_2) \ge \v_0\}$,  the Taylor's expansion in \eqref{appC:taylor_expansion_ratio} gives, for some $K > 0$ and $\lambda \in [0,1]$,
\begin{align*}
\bigl|\widecheck{Z}_i - Z_i\bigr|
&\le K \left\{ \bigl|\widecheck{U}_i(\mathcal{R}_1) - U_i(\mathcal{R}_1)\bigr|
+ \left(\bigl|U_i(\mathcal{R}_1)\bigr| + \lambda \bigl|\widecheck{U}_i(\mathcal{R}_1) - U_i(\mathcal{R}_1)\bigr|\right)
\,\bigl|\widecheck{U}_i(\mathcal{R}_2) - U_i(\mathcal{R}_2)\bigr| \right\}.
\end{align*}

Clearly, since the graphon is bounded away from $0$ and $\infty$, for an appropriate choice of $k$ and $\v_0$,
\[
\operatorname{pr}\big( |\widecheck{Z}_i - Z_i| > t,\,  \kappa_i(\mathcal{R}_2) \ge \v_0 \big)
\le \operatorname{pr}\big( |\widecheck{U}_i(\mathcal{R}_1) - U_i(\mathcal{R}_1)| > kt \big)
+ \operatorname{pr}\big( |\widecheck{U}_i(\mathcal{R}_2) - U_i(\mathcal{R}_2)| > kt \big).
\]
where $k=k(\v_0)$.

For $\{\kappa_i(\mathcal{R}_2) < \v_0\}$, since $\{\kappa_i(\mathcal R_2)<\v_0\}
\subset \{\,|\widecheck{U}_i(\mathcal{R}_2) - U_i(\mathcal{R}_2)|>\v_0'/2\,\}$, we have 
\begin{align*}
\operatorname{pr}\big( \kappa_i(\mathcal{R}_2) < \v_0 \big) \leq \operatorname{pr}\big(|\widecheck{U}_i(\mathcal{R}_2) - U_i(\mathcal{R}_2)|>\v_0'/2\big).   
\end{align*} 

We apply the union bound over $1\le i\le m$ and obtain
\begin{align*}
\operatorname{pr}\!\left( \max_{1\le i\le m} |\widecheck{Z}_i - Z_i| > t \right)
\le 2mK_1 e^{-a_3 n t^2} + mK_2\frac{n^{1-c_1}\rho_n^{-d_1} + n^{1-c_2}\rho_n^{-d_2}}{t^2}.
\end{align*}

Therefore, choosing
\[
t_n = \left(\frac{\log m}{n}\right)^{1/2} + (m\,n^{1-c_1}\rho_n^{-d_1})^{1/2} + (m\,n^{1-c_2}\rho_n^{-d_2})^{1/2},
\]
ensures
\[
\max_{1\le i\le m} |\widecheck{Z}_i - Z_i| = O_P\!\left( \left(\frac{\log m}{n}\right)^{1/2} + (m\,n^{1-c_1}\rho_n^{-d_1})^{1/2} + (m\,n^{1-c_2}\rho_n^{-d_2})^{1/2} \right).
\]

Next, we deal with the two examples in the main text. For $\widehat{Z}_i$ as local transitivity, this leads to the condition on down-sample size as
\begin{align*}
m^2=o(\lambda_n \wedge \lambda_n^3/n), \quad \textrm{and}\quad m\log m=o(n).
\end{align*}

And for $\widehat{Z}_i$ as the neighborhood average, this leads to the down-sample size regime as
\begin{align*}
m^2=o(\lambda_n), \quad \textrm{and}\quad m\log m=o(n).
\end{align*}
That completes the proof.
\end{proof}

\subsection{Proof of Example \ref{ds-example:rdpg}}\label{appC-section:down-sampling_grdpg}

\begin{proof}
The first claim follows immediately from the 2-to-$\infty$ norm bound \eqref{eq-2-infinity-bounds} in Lemma \ref{app-rdpg:ase-error representation}.

For the second claim, the argument made in the proof of Theorem \ref{rdpg:theorem6} implies that it suffices to show $m^{1/2} \|\mathrm{Q}_n \widecheck{Z}- Z \|_{2 \rightarrow \infty}^2 = o_P(1)$.  Thus, it suffices to choose $m=o\left((\lambda_n^2/\log^{4c} n) \wedge n\right)$, for some constant $c$.   
\end{proof}

\end{document}